\documentclass[12pt]{article}
\usepackage[english]{babel}    
\usepackage[ansinew]{inputenc} 
\usepackage[T1]{fontenc}       
\usepackage{lmodern,microtype} 
\usepackage{amsmath,amsthm,amssymb,amsfonts} 
\usepackage{dsfont,mathrsfs,ushort} 

\usepackage{titlesec,titling}  
\usepackage[nohead]{geometry}  
\usepackage{setspace,caption}  
\usepackage{enumitem,booktabs} 
\usepackage{tabularx}
\usepackage[comma]{natbib} 
\usepackage{pstricks,pst-all}   
\usepackage{pst-plot,pst-node,pst-3dplot}
\usepackage{sgamevar}           
\usepackage{hyperref,breakurl}
\usepackage{pgfplots}
\usepackage{longtable}
\usepackage{algorithmic}

\setenumerate{label=\small(\roman*)}
\hypersetup{colorlinks=true,pdfnewwindow=true,pdfstartview=FitH,%
	pdftitle="SRP",pdfauthor="Victor Aguiar and Nail Kashaev",%
	urlcolor=blue!90!red!45!black,citecolor=blue!90!red!45!black,linkcolor=red!90!black}
\DeclareCaptionFont{fancy}{\bfseries\sffamily}
\gamemathtrue
\allowdisplaybreaks

\pgfplotsset{compat=1.14}

\providecommand{\psreset}{\psset{%
		linewidth=0.3pt,linestyle=solid,linecolor=black,
		dotsize=2.5pt,dotsep=2.5pt,arrowsize=4pt,
		fillstyle=none,fillcolor=white,
		showpoints=false,arrows=-,linearc=0,framearc=0,
		hatchsep=2pt,hatchwidth=0.2pt,nodesep=4pt,opacity=1}
	\psset{gridcolor=black!60, subgridcolor=black!30}
}

\psreset

\usepackage{graphicx}
\graphicspath{ {figures/} }

\titleformat{\section}[block]{\centering\large\bfseries\sffamily}{\thesection.}{0.5em}{}
\titleformat{\subsection}[block]{\flushleft\bfseries}{\thesubsection.}{0.5em}{}
\titleformat{\subsection}[block]{\flushleft\bfseries\sffamily}{\thesubsection.}{0.5em}{}
\titleformat{\subsubsection}[runin]{\normalsize\itshape}{\bfseries\upshape\sffamily\thesubsubsection.}{0.5em}{}[.--\:]
\renewcommand{\thesubsubsection}{\arabic{section}.\arabic{subsection}.\alph{subsubsection}}
\titlespacing{\section}{0ex}{10ex}{5ex}
\titlespacing{\subsection}{0in}{6ex}{3ex}
\titlespacing{\subsubsection}{0mm}{2ex}{0.5em}
\pretitle{\begin{center}\LARGE\bfseries\sffamily}
	\posttitle{\par\end{center}\vskip 0.5em}
\preauthor{\begin{center} \large \lineskip 0.5em\begin{tabular}[t]{c}}
		\postauthor{\end{tabular}\par\end{center}}
\predate{\begin{center}\small}
	\postdate{\par\end{center}}
\providecommand{\abstitle}[1]{{\par\vspace*{2ex}\small\bfseries\sffamily #1}\hspace*{1ex}}
\renewenvironment{abstract}%
{\begin{center}\begin{minipage}{0.8\linewidth}%
			\setlength{\parindent}{0.0em}\abstitle{Abstract}\small}%
		{\end{minipage}\end{center}\vfill\clearpage}

\providecommand{\Expt}[2][\!]{\mathds{E}_{#1}\left[\,#2\,\right]}
\providecommand{\Char}[1]{\mathds{1}\left(\,#1\,\right)}
\providecommand{\Real}{{\mathds{R}}}
\providecommand{\Natural}{{\mathds{N}}}
\providecommand{\tr}{^{\prime}}
\providecommand{\as}{\ensuremath{\mathrm{a.s.}}}
\providecommand{\rand}[1]{\mathbf{#1}}
\providecommand{\rands}[1]{\boldsymbol{#1}}
\providecommand{\norm}[1]{\left\lVert#1\right\rVert}
\providecommand{\Prob}[1]{\mathds{P}\left(#1\right)}
\providecommand{\Exp}[1]{\mathds{E}\left[#1\right]}
\providecommand{\abs}[1]{\left\lvert#1\right\rvert}

\newcommand{\m}{\ensuremath{\mathrm{m}}}
\newcommand{\rat}{\ensuremath{\mathrm{R}}}

\newcommand{\ed}{\ensuremath{\mathrm{ED}}}
\newcommand{\s}{\ensuremath{\mathrm{s}}}
\newcommand{\cf}{\ensuremath{\mathrm{C}}}
  \theoremstyle{remark}
  \newtheorem{rem}{\protect\remarkname}
  \theoremstyle{plain}
  \newtheorem{lem}{\protect\lemmaname}
  \newtheorem{proposition}{\protect\propositionname}
  \theoremstyle{definition}
  \newtheorem{defn}{\protect\definitionname}
\theoremstyle{plain}
\newtheorem{thm}{\protect\theoremname}
\newtheorem*{thmnonumber}{\protect\theoremname}
  \theoremstyle{plain}
  
 \theoremstyle{definition}
  \newtheorem{example}{\protect\examplename}
  \theoremstyle{plain}
  \newtheorem{assumption}{\protect\assumptionname}
  
\newtheoremstyle{subplain}
  {\topsep}   
  {\topsep}   
  {\itshape}  
  {0pt}       
  {\bfseries} 
  {.}         
  {5pt plus 1pt minus 1pt} 
  {\protect\assumptionname\ \arabic{assumption}.\arabic{subassumption}}          

 \theoremstyle{subplain}
 \newtheorem{subassumption}{\protect\assumptionname}

  \providecommand{\assumptionname}{Assumption}

  \providecommand{\definitionname}{Definition}
  \providecommand{\lemmaname}{Lemma}
  \providecommand{\propositionname}{Proposition}
  \providecommand{\remarkname}{Remark}
\providecommand{\corollaryname}{Corollary}
\providecommand{\theoremname}{Theorem}
\providecommand{\examplename}{Example}

\usepackage[capitalise]{cleveref}

\makeatother

\begin{document}

\title{Stochastic Revealed Preferences with\\ Measurement Error\thanks{\scriptsize We thank the editor and three anonymous referees for comments and suggestions that have greatly improved the manuscript. We thank for their comments seminar and conference participants at UC Davis (ARE Department), University of Toronto, Universit\'e libre de Bruxelles, Caltech (DHSS), Arizona State University, Cornell University, NASMES 2018, MEG Meeting 2018, IAAE Conference 2018, BRIC 2018, CEA Meeting 2019, SAET Conference 2019, and the Conference on Counterfactuals with Economic Restrictions at UWO. We are grateful to Roy Allen, Andres Aradillas-Lopez, Elizabeth Caucutt, Laurens Cherchye, Ian Crawford, Tim Conley, Mark Dean, Rahul Deb, Thomas Demuynck, Federico Echenique, Yuichi Kitamura, Lance Lochner, Jim MacGee, Nirav Mehta, Salvador Navarro, Joris Pinkse, David Rivers, Bram De Rock, Susanne Schennach, J\"org Stoye, Al Slivinsky, Dan Silverman, and Todd Stinebrickner for useful comments and encouragement. We thank Vered Kurtz-David for providing validation data; Lance Lochner, Compute Ontario, and Compute Canada for computational resources; and the Social Sciences and Humanities Research Council for financial support.} }

\author{ 
	Victor H. Aguiar\thanks{Department of Economics, University of Western Ontario.}\\ vaguiar@uwo.ca
	\and
	Nail Kashaev\protect\footnotemark[2]\\
	nkashaev@uwo.ca
	}
\date{First version: September 2017\\
This version: June 2020 }
\maketitle
\begin{abstract}
A long-standing question about consumer behavior is whether individuals' observed purchase decisions satisfy the revealed preference (RP) axioms of the utility maximization theory (UMT). Researchers using survey or experimental panel data sets on prices and consumption to answer this question face the well-known problem of measurement error. We show that ignoring measurement error in the RP approach may lead to overrejection of the UMT. To solve this problem, we propose a new statistical RP framework for consumption panel data sets that allows for testing the UMT  in the presence of measurement error. Our test is applicable to all consumer models that can be characterized by their first-order conditions. Our approach is nonparametric, allows for unrestricted heterogeneity in preferences, and requires only a centering condition on measurement error.  We develop two applications that provide new evidence about the UMT.  First, we find support in a survey data set for the dynamic and time-consistent UMT in single-individual households, in the presence of \emph{nonclassical} measurement error in consumption. In the second application, we cannot reject the static UMT in a widely used experimental data set in which measurement error in prices is assumed to be the result of price misperception due to the experimental design. The first finding stands in contrast to the conclusions drawn from the deterministic RP test of \citet{browning1989anonparametric}. The second finding reverses the conclusions drawn from the deterministic RP test of \citet{afriat1967construction} and \citet{varian1982nonparametric}. 


\end{abstract}

\section{Introduction}\label{sec:intro}
One well-known feature of consumer panel data sets---whether they are based on surveys, experiments, or scanners---is measurement error in prices or consumption.\footnote{See \citet{mathiowetz2002measurement}, \citet{echenique2011money}, \citet{carroll2014introduction}, and \citet{gillen2015experimenting}.} This is of significant concern because measurement error is responsible for one important limitation of the standard revealed-preference (RP) framework. Specifically, RP tests tend to overreject the utility maximization theory (UMT). To the best of our knowledge, this paper proposes the first fully nonparametric statistical RP framework  for consumer panel data sets in the presence of measurement error. We show that taking measurement error into account can significantly change the conclusions about the validity of the UMT in a given context. In the two applications we develop, we cannot reject the validity of the UMT, a finding which contradicts the conclusions of the deterministic RP framework. 
\par
Measurement error is the difference between the unobserved but true value of the variable of interest and its observed but mismeasured counterpart. If the UMT is valid, the corresponding RP conditions must be satisfied by the true prices and consumption. However, there is no reason to believe that either the mismeasured consumption or mismeasured prices that we usually observe are consistent with the RP conditions. A key concern for RP practitioners is that, in the presence of measurement error, a deterministic RP test may overreject the null hypothesis that the UMT is valid.
We provide Monte Carlo evidence that this concern may be relevant in practice.
Measurement error in consumption may arise in survey data due to misreporting, in experimental data due to trembling-hand errors, and in scanner data due to recording errors.\footnote{\textit{Trembling-hand errors} are nonsystematic mistakes incurred by a subject when trying to implement a decision because she has difficulties with the interface of an experiment.} Measurement error in prices may arise in experimental data due to misperception errors and in scanner data due to unobserved coupons.\footnote{\textit{Misperception errors} are nonsystematic mistakes incurred by a subject because she misperceives information due to the experimental design.}
\par
Our methodology covers as special cases the static UMT, and the classic dynamic UMT with exponential discounting.\footnote{We also cover firm-cost minimization \citep{varian1984nonparametric}, dynamic rationalizability with quasi-hyperbolic discounting \citep{blownonparametric2017}, homothetic rationalizability \citep{varian1985non}, quasilinear rationalizability \citep{brownquasilinear}, expected utility maximization \citep{diewert2012afriat}, and static utility maximization with nonlinear budget constraints \citep{forges2009afriats}.} When we apply our methodology to a consumer panel survey of Spanish households and allow for measurement error in consumption, we find that the dynamic UMT with exponential discounting cannot be rejected for single-individual households. This first finding contradicts the conclusions of the deterministic RP test of \citet{browning1989anonparametric}. When we apply our framework to a widely used experimental data set \citep{ahn2014estimating} and allow for measurement error in prices that arises due to misperception, we find that the static UMT cannot be rejected. This second finding is the opposite of the conclusions that must be drawn when one applies the deterministic RP test of \citet{afriat1967construction} and \citet{varian1982nonparametric} to the same data set. Taken together, these findings suggest that the negative conclusions about the validity of the UMT drawn from the deterministic RP framework may not be robust to measurement error. 

\par
The leading solution for dealing with measurement error in the RP framework usually consists of perturbing (minimally) any observed individual consumption streams in order to satisfy the conditions of an RP test \citep{adams2014consume}. However, this approach does not allow for standard statistical hypothesis testing. In particular, one cannot control the probability of erroneously rejecting a particular model when such a rejection could be an artifact of noisy measurements. Other works on RP with measurement error, such as the seminal contribution of \citet{varian1985non}, allow for statistical hypothesis testing but require knowledge of the distribution of measurement error; this may be impractical because this does not align with the nonparametric nature of the RP framework. In contrast, our procedure does not suffer from these issues.  
\par
Our main result is the formulation of a statistical test for the null hypothesis that a random data set of mismeasured prices and consumption is consistent with any given model that can be characterized by first-order conditions. Based on this test we provide and implement a general methodology to make out-of-sample predictions or counterfactual analyses with minimal assumptions (e.g., sharp bounds for average or quantile demand). Our approach takes advantage of the work of \citet{schennach2014entropic} on Entropic Latent Variable Integration via Simulation (ELVIS) to provide a practical implementation of our test.
\par
Our RP methodology is fully nonparametric and admits unrestricted heterogeneity in preferences. In addition, we require only a centering condition on the unobserved measurement error. The centering condition captures the application-specific knowledge we have about measurement error. Moreover, our framework is general enough to allow for (i) \emph{nonclassical} measurement errors in consumption in survey environments, (ii) trembling-hand errors in experimental setups, (iii) misperception of prices due to experimental designs (price measurement error), and (iv) different forms of measurement error in prices in scanner environments.
\par
In our first application in particular, we require that consumers be accurate, on average, in recalling and reporting their total expenditures. This assumption is compatible with systematic misreporting of consumption in surveys. For our second application, to an experimental data set, we require that errors in consumption or prices be centered around zero, which is compatible with trembling-hand and misperception errors in subjects' behavior. Measurement error in experimental data sets may arise because the experimental design may fail to elicit the intended choices of consumers. In such a case, classical measurement-error assumptions that allow for nonsystematic mistakes must be taken into account to ensure the external validity of the conclusions drawn from applying any RP test to this type of data set.    
\par
Our empirical contribution is to apply our methodology to a well-known consumer panel survey of single-individual and couples' households in Spain\footnote{This data set has been used in \citet{beatty2011howdemanding,adams2014consume}, and \citet{blownonparametric2017}.} in order to test for the dynamic UMT, and to reexamine the static UMT in a widely used experimental data set \citep{ahn2014estimating}; in doing this reexamination, we allow (separately) for trembling-hand errors and misperceived prices. 
\par
For the first application, we note that under the exponential discounting model, the consumer's time preferences are captured by a time-invariant discount factor and a time-invariant instantaneous utility. The main feature of this model is the time-consistency of the exponential discounting consumer. In other words, if the consumer prefers consumption bundle $c$ at time $t$ to $x$ at time $t+k$, then she will always prefer $c$ at time $\tau$ to $x$ at time $\tau+k$.
The exponential discounting model remains the workhorse of a large body of applied work in economics.
However, several authors, such as 
\citet{browning1989anonparametric}, \citet{dellavigna2006payingnot}, and \citet{blownonparametric2017},
have provided suggestive evidence against the validity of this model. Our methodology addresses, in a nonparametric fashion, the presence of measurement error in survey data, in order to examine the robustness of these findings.
\par
We find support for exponential discounting behavior for single-individual households. This contrasts with the results of applying the deterministic methodology of \citet{browning1989anonparametric} to the same sample. At the same time, in line with the findings of \citet{blownonparametric2017} (who also use  the deterministic methodology of \citealp{browning1989anonparametric}), we reject the null hypothesis of exponential discounting for the case of couples. When compared with the single-household evidence, these results suggest that time inconsistencies in consumer behavior in the couples' case arise due to preference aggregation. 
\par
In our second application, we apply our methodology to  test for the validity of the classical static UMT. Since the work of \citet{afriat1967construction} and \citet{varian1982nonparametric}, researchers have used the deterministic RP framework to examine the validity of the UMT in experimental data sets. The experimental design proposed by \citet{ahn2014estimating} is particularly useful for this task since it provides a controlled environment with substantial price variation, variation that guarantees that the UMT has empirical bite. 
\par
When we apply the deterministic RP test to this experimental data set, we conclude that the UMT is rejected for most subjects. Nonetheless, the external validity of the conclusions drawn from the deterministic RP tests applied to these experimental data sets may be limited. One key reason for this is that the elicitation of consumer behavior may have been subject to measurement error. \citet{gillen2015experimenting} argue that experimental elicitations of choices are subject to random variation in participants' perception and focus. 
Moreover, RP practitioners since \citet{afriat1967construction} have recognized that the deterministic RP test for the static UMT may be too demanding in the presence of imperfect devices for the elicitation of choices. Many researchers have studied how to allow for optimization mistakes in the RP framework and how to measure the intensity of any departure from rationality.\footnote{See \citet{afriat1967construction}, \citet{varian1990goodness}, and \citet{echenique2011money}.} However, none of the existing approaches designed to introduce the possibility of mistakes in the RP framework has allowed for a fully nonparametric approach to doing standard statistical hypothesis testing. 
In our application, we allow for the possibility of nonsystematic mistakes by requiring that the measurement error in consumption or prices be mean zero. 
\par
We cannot reject the null hypothesis of the validity of the static UMT with misperception of prices. However, when we allow only for trembling-hand errors in consumption, we must strongly reject the static UMT. Our findings call into question the robustness of the deterministic RP test that is due to \citet{afriat1967construction} and \citet{varian1982nonparametric} to measurement error in prices.
\subsection*{Outline} The paper proceeds as follows. 
Section~\ref{sec:determin} presents the first-order conditions approach to the deterministic RP methodology. Section~\ref{sec:edrum} contains our statistical test. Section~\ref{sec:recoverabilitycounterfactual} presents a framework for recoverability and counterfactual analysis on the basis of our testing methodology. Section~\ref{sec:econometric framework} provides an econometric framework for our methodology. Section~\ref{sec:error in data} provides a guide specifying the centering condition in different environments.
Section~\ref{sec:empirical application} implements our empirical test for the case of the dynamic UMT in a consumer panel-survey data set. Section~\ref{sec:ApplicationStaticUMT} implements our methodology for the case of the static UMT in an experimental data set. Section~\ref{sec:literature} presents a brief
discussion of related literature. Finally, we conclude in Section~\ref{sec:conclusion}.
All proofs can be found in Appendix~\ref{sec:appendix}. 

\section{The Revealed-Preference Methodology and the First-Order Conditions Approach}\label{sec:determin}
The main objective of this section is to provide a brief summary in a united fashion of two very important deterministic consumer models and their RP characterization. In particular, we study the static UMT or rational model ($\rat$), and the dynamic UMT with exponential discounting ($\ed$). These models are at the center of many applied and theoretical works. We show that they can be completely characterized by their first-order conditions in an RP fashion. All quantities used here
are assumed to be measured precisely.
\par
Let the consumption space be $\Real_{+}^{L}\setminus\{0\}$, where $L\in\Natural$ is the number of commodities.\footnote{We use $\Natural$ to denote the set of natural numbers. The expression $\Real_{+}^L$ denotes the set of componentwise nonnegative elements of the $L$-dimensional Euclidean space $\Real^L$, and $\Real_{+}^L\setminus\{0\}$ denotes the set of vectors $v\in\Real_{+}^L$ that are distinct from zero ($v\neq0$). Similarly, $\Real^L_{++}$ denotes the set of componentwise positive elements of $\Real^L_{+}$. The inner product of two vectors $v_1,v_2\in\Real^L$ is denoted by $v_1\tr v_2$.} Consider a consumer who is endowed with a utility function $u:\Real_{+}^{L}\to\Real$  that is assumed to be concave, locally
nonsatiated, and continuous. The consumer faces a sequence of decision problems indexed by $t\in \mathcal{T}$, where $\mathcal{T}=\{0,\cdots,T\}$, with a known and finite $T\in\Natural$. At each decision problem $t\in \mathcal{T}$, the consumer faces the price vector $p_{t}\in\Real_{++}^{L}$. 
\begin{defn}[Static UMT, $\rat$-rationalizability] 
A deterministic array $(p_{t},c_{t})_{t\in\mathcal{T}}$  is $\rat$-rationalizable (in a static sense) if there exists a concave, locally nonsatiated, and continuous function $u$, and some constants $y_{t}>0$, $t\in\mathcal{T}$, such that the consumption bundle $c_{t}$ solves:
\begin{align*}
    &\max_{c\in\Real_{+}^{L}} u(c),\\
    &\text{s.t. }p_{t}\tr c=y_{t},
\end{align*}
for all $t\in\mathcal{T}$.
\end{defn}
\par
Next we focus on the dynamic UMT. We assume that an individual consumer has preferences over a stream
of dated consumption bundles $(c_{t})_{t\in\mathcal{T}}$, where $\mathcal{T}=\{0,\cdots,T\}$,  $T\in\Natural$, and $c_{t}\in\Real_{+}^{L}\setminus\{0\}$. (The number
of goods, $L$, is kept the same across the time interval.) At time $\tau$, the consumer chooses how much $c_{\tau}$ she will consume by maximizing
\[
V_{\tau}(c)=u(c_{\tau})+\sum_{j=1}^{T-\tau}d^{j}u(c_{\tau+j}),
\]
subject to the linear budget or flow constraints shown here:
\[
p_{t}\tr c_{t}-y_{t}+s_{t}-a_{t}=0,\quad t=\tau,\dots,T,
\]
where $d\in(0,1]$ is the discount factor; $p_{t}\in\Real_{++}^{L}$ is the price vector as before; $y_{t}\in\Real_{++}$ is income received by the individual at time $t$; $s_{t}$ is the amount of savings held by the consumer at the end of time $t$; and $a_{t}$ is the volume of assets held at the start of time $t$. The consumer invests all her savings. Moreover, the assets evolve according to the following law of motion:
\[
a_{t}=(1+r_{t})s_{t-1},
\]
where $r_{t+1}>-1$ is the interest rate that is accessible for the consumer. The holdings of assets in the last period ($t=T$) are set to be zero. 
\par
The intertemporal value function, $V_{t}:\Real_{+}^{L\times(T-t+1)}\to\Real_{++}$, represents the consumer preferences
at a given time $t$. The components of this representation are the parameters of the model. First, $d\in(0,1]$ is a scalar number that measures the degree of discount that the consumer gives to the future. Second, $u:\Real_{+}^{L}\to\Real_{++}$ is an instantaneous utility function that is assumed to be concave, locally nonsatiated, and continuous. The exponential discounting consumer is time-consistent, that is, she will solve the dynamic problem above the same way at any point of the time window.   

\begin{defn}[Dynamic UMT, $\ed$-rationalizability]
A deterministic array $(p_{t},r_{t},c_{t})_{t\in\mathcal{T}}$  is $\ed$-rationalizable if there exist a concave, locally nonsatiated, and continuous function $u$, a vector $(y_{t})_{t\in\mathcal{T}}\in\Real^{\abs{\mathcal{T}}}_{++}$, and a scalar $a_{0}\geq0$ such that the consumption stream $(c_{t})_{t\in\mathcal{T}}$ solves:
\[
    \max_{z\in\Real_{+}^{L\times|\mathcal{T}|}} u(z_{0})+\sum_{t=1}^{T}d^{t}u(z_{t}),
\]
subject to 
\begin{align*}
    &p_{0}\tr z_{0}+\sum_{t=1}^{T}\frac{p_{t}\tr z_{t}}{\prod_{i=1}^{t}[1+r_{i}]}=
\sum_{t=1}^{T}\frac{y_{t}}{\prod_{i=1}^{t}[1+r_{i}]}+a_{0}.
\end{align*}
\end{defn}
$\ed$-rationalizability implicitly assumes perfect foresight (e.g., individuals know their future income) and homogeneity of consumers within a household (e.g., household members have the same discount factors). In Appendix~\ref{appen: extensions of ED } we show that our methodology covers two extensions of this model: the dynamic UMT with income uncertainty, and the collective model of \citet{adams2014consume}.

\subsection{The First-Order Conditions Approach}\label{subsec:determin,bench}
Now we establish that any consumer model $\m\in\{\rat,\ed\}$ can be completely characterized in terms of its first-order conditions with respect to (i) a concave, locally nonsatiated, and continuous utility function $u:\Real_{+}^{L}\to\Real$, (ii) the effective (or transformed) prices $\rho_t^\m\in\Real_{++}^{L}$, and (iii) restrictions on some constants $\lambda_t^\m \in \Real_{++}$ and $\delta^{\m}_t\in(0,1]$, interpreted as the marginal utility of income and the discount rate, respectively. We call this \emph{the first-order conditions approach}. Observe that the utility function is model-independent, but the effective prices, the marginal utility of income, and the discount rate are not. We define the effective prices in Table~\ref{table:Lambdarhodef}. 
\begin{longtable}{cccccc}
\bottomrule
\midrule\midrule
\endfoot
\caption{\label{table:Lambdarhodef} Definition of $\rho^\m_t$  }\\
\toprule
$\m$ &  $\rat$  &    $\ed$   \\
\midrule
$\rho^\m_t$ &   $p_t$   &      $p_t/\prod_{j=1}^{t}(1+r_j)$  \\
\end{longtable}

The following lemma summarizes the results in \citet{browning1989anonparametric} for the exponential discounting case, and it is trivial for the static rationalizability case. Let $\nabla u(c_t)$ denote a supergradient of $u$ at the point $c_t$.\footnote{The supergradient is $\nabla u(c_t)=\{\xi\in\Real^{L}\::\:u(c)-u(c_t)\leq \xi\tr(c-c_t),\:\forall\:c\in\Real_{+}^{L}\setminus\{0\}\}$. Under differentiability, $\nabla u(c_t)$ is a gradient.}
\begin{lem}
	\label{lem:Browiningparametric}
	For any model $\m\in\{\rat,\ed\}$, a deterministic array $(\rho_{t}^{\m},c_{t})_{t\in\mathcal{T}}$ is $\m$-rationalizable if and only if
	there exists $(u,(\lambda_{t}^{\m},\delta_t^{\m})_{t\in\mathcal{T}})$ such that
	\begin{enumerate}
	\item $u:\Real_{+}^{L}\to\Real$ is a concave, locally nonsatiated, and continuous utility function;
	\item $\delta_t^{\m}\nabla u(c_{t})\leq\lambda_{t}^{\m}\rho_{t}^{\m}$ for every $t\in\mathcal{T}$. If $c_{t,j}\neq0$, then $\delta_t^{\m}\nabla u(c_{t})_j=\lambda_{t}^{\m}\rho_{t,j}^{\m}$, where $c_{t,j}$, $\nabla u(c_{t})_j$, and $\rho_{t,j}$ are the $j$-th components of $c_{t}$, $\nabla u(c_{t})$, and $\rho_{t}$, respectively;
	\item $\lambda^\rat_t=\lambda_t>0$ and $\delta_t^{\rat}=1$ for all $t\in\mathcal{T}$;
	\item $\lambda^\ed_t=1$ and $\delta^{\ed}_t=d^{t}$, where $d\in(0,1]$, for all $t\in\mathcal{T}$.
	\end{enumerate}
\end{lem}
We want to highlight that while we focus on these two models for expositional and motivational purposes, our methodology is applicable to any model that can be characterized using the first-order conditions approach. 

\begin{rem}
	Lemma~\ref{lem:Browiningparametric} allows for nondifferentiable utility functions. So, the supergradient of ${u}( c_{t})$ may be set-valued. In this case one should read the condition $\delta_t^{\m}\nabla u(c_{t})\leq\lambda_{t}^{\m}\rho_{t}^{\m}$ as ``there exists $\xi\in\nabla u(c_t)$ such that $\delta_t^{\m}\xi\leq\lambda_{t}^{\m}\rho_{t}^{\m}$.''  
\end{rem}

\begin{rem} Lemma~\ref{lem:Browiningparametric} specialized for $\ed$-rationality implies that we do not need to observe consumers over all periods of their lives. The first-order conditions are the same irrespective of whether the consumer lives any finite number of time periods containing the observed time-window, or is alive only during the latter period.
\end{rem}

\subsection{The Elimination of a Latent Infinite-Dimensional Parameter}\label{subsec:determin,latent}
Since our objective is not to estimate but to test $\m$-rationalizability, we will eliminate the utility function $u$ from its characterization. We follow the theorists of RP to eliminate the latent infinite-dimensional parameters by exploiting their shape restrictions.
\par
In particular, we follow \citet{afriat1967construction}, \citet{varian1985non}, \citet{browning1989anonparametric},
and \citet{rockafellar1970convexanalysis} to formulate
a result that eliminates the utility function  $u$ (an infinite dimensional parameter) from the first-order conditions.
The cost of doing this is that we have to replace the first-order conditions by a set of inequalities that require only the concavity of $u$. As a result, the inequalities are exact and
do not involve any form of approximation; this is an advantage compared to other nonparametric methods (e.g., sieves, kernel estimators) or the parametric approach used in many applied papers.
\par
To formulate our result, we first recall the definition of the concavity of $u$.
\begin{defn}[Concavity]\label{def:concavity}
	A utility function $u$ is said to be concave if and only
	if $u(\tilde{c})-u(c)\leq\nabla u(c)\tr(\tilde{c}-c)$, for all $c,\tilde{c}\in\Real_{+}^{L}\setminus\{0\}$.
\end{defn}
\begin{rem}
	In Definition~\ref{def:concavity} we implicitly assume the existence of the supergradient of $u$. Since the supergradient may be set-valued, one should read the condition $u(\tilde{c})-u(c)\leq\nabla u(c)\tr(\tilde{c}-c)$ as ``$u(\tilde{c})-u(c)\leq\xi\tr(\tilde{c}-c)$ for all $\xi\in\nabla u(c)$.''  
\end{rem}
The nonparametric characterization of the $\m$-rationalizability of observed consumption and prices without measurement error is captured by the following result.

\begin{thm}\label{thm:Deterministic-thm_exponentialdiscount} 
For any $m\in\{\rat,\ed\}$, the following are equivalent:
	\begin{enumerate}
		\item The deterministic array $(\rho_{t}^\m,c_{t})_{t\in\mathcal{T}}$ is $\m$-rationalizable.
		\item There exist vectors $(\lambda^{\m}_{t})_{t\in\mathcal{T}}$, $(\delta_t)_{t\in\mathcal{T}}$, and a positive vector $(v_{t})_{t\in\mathcal{T}}$
		such that:
		\[
		v_{t}-v_{s}\geq \dfrac{\lambda^\m_{t}}{\delta_t^{\m}}\rho_{t}^{\m\prime}(c_{t}-c_{s}),
		\]
		with $\lambda^\rat_t=\lambda_t>0$, $\delta^{\rat}_t=1$, $\lambda^{\ed}_t=1$, and   $\delta^\ed_t=d^{t}$, where $d\in(0,1]$, for all $t,s\in\mathcal{T}$.
	\end{enumerate}
\end{thm}

Theorem~\ref{thm:Deterministic-thm_exponentialdiscount}  summarizes known results from the RP literature.\footnote{The proof is a consequence from the results in  \citet{afriat1967construction}, \citet{varian1985non}, and \citet{browning1989anonparametric} taken together.} 
Observe that Theorem~\ref{thm:Deterministic-thm_exponentialdiscount} has transformed the first-order conditions  that depend on  the infinite-dimensional $u$ to a set of inequalities that depend only on a deterministic finite-dimensional array $(v_{t},\delta_t^{\m},\lambda^{\m}_{t})_{t\in\mathcal{T}}$. Nonetheless, this set of conditions is satisfied if and only if we can find a utility function that satisfies the conditions in Lemma~\ref{lem:Browiningparametric}.\footnote{In contrast to Afriat's theorem for the static UMT, the assumption of concavity of the utility function is necessary in our framework. The reason is that concavity is testable in some cases that are different from static utility maximization (e.g., expected utility maximization, \citealp{polisson2020revealed}). Concavity guarantees that first-order conditions are necessary and sufficient in a wide variety of models beyond the static UMT. For the additional generality our result requires this additional constraint.} Checking the set of inequalities is a parametric problem and it tells us whether a consumption stream is $\m$-rationalizable. This methodology is traditionally applied at the individual level in panel data sets, assuming that the data contains no measurement error. In the next section we extend the RP framework to a noisy or stochastic environment.

\section{The Revealed-Preference Approach with Measurement Error}\label{sec:edrum}
In this section, we introduce a new statistical notion of $\m$-rationalizability (henceforth, $\s/\m$-rationalizability) with mismeasured consumption or prices, and provide a result similar to Theorem~\ref{thm:Deterministic-thm_exponentialdiscount} in the presence of measurement error. From here on, we use boldface font to denote random objects and regular font for deterministic ones.
\subsection{Statistical Rationalizability}\label{subsec:edrum,model}
We are interested in testing a statistical model of consumption such that each individual is an independent, identically distributed (i.i.d.) draw from some stochastic consumption rule. Note that by Lemma~\ref{lem:Browiningparametric} the choice of a particular model $\m$ only affects the definition of the effective price, and the restrictions on the marginal utility of income and the discount rate. Henceforth, we fix some model such that the effective prices, and the restrictions on the marginal utility of income and the discount rate are known, and we omit the superscript $\m$ from the notation. 
\par
Using Lemma~\ref{lem:Browiningparametric} as motivation, we directly define $\s/\m$-rationalizability as follows. Let $\rands{\rho}^*_t\in P^*_t\subseteq\mathbb{R}^L_{++}$ and $\rand{c}^*_t\in C_t^*\subseteq\mathbb{R}^L_+ \setminus{\{0\}}$ denote random vectors of true effective prices and true consumption at time $t$, respectively.\footnote{For short, we use $\as$ instead of ``almost surely.'' We denote (i) the probability of an event $A$ by the expression $\Prob{A}$; (ii) the indicator function by $\Char{A}=1$ when the statement $A$ is true, otherwise it is zero; (iii) the mathematical expectation of any random vector $\rand{z}$ by the expression $\Exp{\rand{z}}$; (iv) the cardinality of a set $\mathcal{A}$ is given by the expression $|\mathcal{A}|$; and (v) the norm of a vector $v$ is given by $\norm{v}$.}

\begin{defn} [$\s/\m$-rationalizability]
	A random array $(\rands{\rho}^*_{t},\rand c_{t}^{*})_{t\in\mathcal{T}}$
	is $\s/\m$-rationalizable if there exists a tuple $(\rand{u},(\rands{\lambda}_t,\rands{\delta}_t)_{t\in\mathcal{T}})$ such that
	\begin{enumerate}
	    \item $\rand u$ is a random, concave, locally nonsatiated, and continuous utility function;
	    \item $(\rands{\lambda_t})_{t\in\mathcal{T}}$ is a positive random vector, interpreted as the marginal utility of income, supported on or inside a known set $\Lambda\subseteq\Real_{++}^{|\mathcal{T}|}$;
	    \item $(\rands{\delta}_t)_{t\in\mathcal{T}}$ is a positive random vector, interpreted as time-varying discount factor, supported on or inside a known set $\Delta\subseteq(0,1]^{|\mathcal{T}|}$;
	    \item $\rands{\delta}_t\nabla \rands{u}(\rand c_{t}^{*})\leq\rands{\lambda_t}\rands{\rho}^*_{t}\:\as$ for all $t\in\mathcal{T}$;
	    \item For every $j=1,\dots,L$ and $t\in\mathcal{T}$, it must be the case that $\Prob{\rand{c}^*_{t,j}\neq 0,\rands{\delta}_t\nabla \rands{u}(\rand c_{t}^{*})_j<\rands{\lambda_t}\rands{\rho}^*_{t,j}}=0$, where $c_{t,j}^*$, $\rho^*_{t,j}$, and $\nabla {u}(c_{t}^{*})_j$ denote the $j$-th components of $c_t^*$, $\rho^*_t$, and $\nabla{u}(c_{t}^{*})$, respectively.
	\end{enumerate}
\end{defn}

This definition means that for a given realization of (i) the utility function, (ii)  the marginal utility of income, and (iii) the discount rate, the realized effective prices and the realized true consumption should fulfill the inequality $\delta_t\nabla u(c_{t}^{*})\leq\lambda_t \rho_{t}^*$. This is a special case of the dynamic random utility model in which the preferences (captured by $\rand u$), the random discount factor (captured by $(\rands{\delta}_t)_{t\in \mathcal{T}})$, and the distribution of the marginal utility of income (captured by $(\rands{\lambda}_t)_{t\in\mathcal{T}})$ are drawn at some initial time for each consumer, and then are kept fixed over time. 
\par
Several consumer models can be characterized by their first-order conditions and by restrictions on the marginal utility of income, as we observed in Section~\ref{sec:determin}. For instance, we define the statistical version of $\rat$-rationalizability or $\s/\rat$-rationalizability by requiring that the support of the marginal utility of income be strictly positive (i.e., $\Lambda=\Real_{++}^{|\mathcal{T}|}$), and the discount rate to be one (i.e., $\Delta=\{1\}^{|\mathcal{T}|}$). Similarly, we define $\s/\ed$-rationalizability by imposing that $\Lambda=\{1\}^{|\mathcal{T}|}$, and the support $\Delta$ be given by the restriction $\rands{\delta}_t=\rand{d}^{t}$, where $\rand{d}$ is a random variable supported on $(0,1]$. The effective prices in each case are to be defined according to Table~\ref{table:Lambdarhodef}. 
\par
Given the definition of $\s/\m$-rationalizability, we can now formulate the stochastic version of Theorem~\ref{thm:Deterministic-thm_exponentialdiscount}.
\begin{lem}
	\label{lem:RandomExponentialdiscountingCstar} For a given random array $(\rands{\rho}^*_{t},\rand c_{t}^{*})_{t\in\mathcal{T}}$, the following are equivalent:
	\begin{enumerate}
		\item The random array $(\rands{\rho}^*_{t},\rand c_{t}^{*})_{t\in\mathcal{T}}$ is $\s/\m$-rationalizable.
		\item There exist positive random vector $(\rand v_{t})_{t\in\mathcal{T}}$, $(\rands{\lambda_{t}})_{t\in\mathcal{T}}$ supported on or inside $\Lambda$, and $(\rands{\delta}_t)_{t\in\mathcal{T}}$ supported on or inside $\Delta$ such
		that
		\[
		\rand{v}_{t}-\rand{v}_{s}\geq\dfrac{\rands{\lambda}_{t}}{\rands{\delta}_t}\rands{\rho}^{*\prime}_{t}(\rand{c}_{t}^{*}-\rand{c}_{s}^{*})\quad\as,\quad\forall s,t\in\mathcal{T}.
		\]
	\end{enumerate}
\end{lem}

Lemma~\ref{lem:RandomExponentialdiscountingCstar} allows us to statistically test the $\s/\m$-rationalizability of $(\rands{\rho}^*_{t},\rand c_{t}^{*})_{t\in\mathcal{T}}$. However, as the following example demonstrates, any test based on this notion of rationalizability cannot differentiate between ``almost'' $\s/\m$-rationalizability and exact $\s/\m$-rationalizabilty (an issue first identified by \citealp{Galichon2013db}).

\begin{example}[Hyperbolic Discounting]\label{ex: hyperbolic}
Consider the case of a consumer who maximizes
\[
V_{\tau}(c)=u(c_{\tau})+\beta \sum_{j=1}^{T-\tau}d^{j}u(c_{\tau+j}),
\]
where $\beta\in(0,1]$ is the present-bias parameter. It is easy to see that if $\beta\rightarrow1$, then the consumption stream generated by this model is arbitrarily close to the $\ed$-rationalizable behavior.
\end{example}
Example~\ref{ex: hyperbolic} presents a random choice rule that is not $\s/\m$-rationalizable, but is arbitrarily close to being $\s/\m$-rationalizable. In other words, there may exist sequences of random arrays that are not $\s/\m$-rationalizable that converge to random arrays that are $\s/\m$-rationalizable.
That is why we need to extend the notion of the consistency of a data set that is characterized by $\s/\m$-rationalizability.

\begin{defn} [Approximate $\s/\m$-rationalizability]
We say that $(\rands{\rho}^*_{t},\rand c_{t}^{*})_{t\in\mathcal{T}}$ is \textbf{approximately consistent} with $\s/\m$-rationalizability if there exists a sequence of random variables $(\rand v_j\tr,\rands{\lambda}_{j}\tr,\rands{\delta}_{j}\tr)\tr\in \Real^{\abs{\mathcal T}}_{+}\times\Lambda\times\Delta$, $j=1,2,\dots$, such that
\[
\Prob{\Char{\rand v_{j,t}-\rand v_{j,s}\geq\dfrac{\rands{\lambda}_{j,t}}{\rands{\delta}_{j,t}}\rands{\rho}^{*\prime}_{t}[\rand c_{t}^{*}-\rand c_{s}^{*}]}=1}\to_{j\to +\infty}1,
\]
for all $s,t\in\mathcal{T}$.
\end{defn}

\subsection{Introducing Measurement Error}\label{subsec:edrum,meas er}
Theorem~\ref{thm:Deterministic-thm_exponentialdiscount} and Lemma~\ref{lem:RandomExponentialdiscountingCstar} provide testable implications of $\s/\m$-rationalizability. These implications depend solely on the distribution of $\rands{\lambda}=(\rands{\lambda}_t)_{t\in\mathcal{T}}$, $\rands{\delta}=(\rands{\delta}_t)_{t\in\mathcal{T}}$ and $\rand v=(\rand v_{t})_{t\in\mathcal{T}}$. The usual approach to testing $\s/\m$-rationalizability would amount to solving a (non)linear programming problem corresponding to Theorem~\ref{thm:Deterministic-thm_exponentialdiscount} at the level of individual consumers. However, this common practice does not work any more in the presence of measurement error. When \textbf{true consumption} or \textbf{true prices} are measured erroneously, we observe not $\rand c_{t}^{*}$ or $\rands{\rho}^*_t$ but rather perturbed versions of them. (See Section~\ref{sec:error in data} for a discussion of the reasons measurement error in consumption and prices arises in survey, experimental, and scanner data sets.)
\par
Define the \textbf{measurement error} $\rand{w}=(\rand w_{t})_{t\in\mathcal{T}}\in W$ as the difference between reported consumption and prices, $\rand{c}=(\rand c_t)_{t\in\mathcal{T}}$ and $\rands{\rho}=(\rands{\rho}_t)_{t\in\mathcal{T}}$; and true consumption and prices, $(\rand c_t^*)_{t\in\mathcal{T}}$ and $(\rands{\rho}_t^*)_{t\in\mathcal{T}}$. That is,
\[
\rand w_t=\left(\begin{array}{c}\rand w_t^c\\
\rand w_t^p\end{array}
\right),
\]
where $\rand w_{t}^{c}=\rand c_{t}-\rand c_{t}^{*}$ and $\rand w_{t}^{p}=\rands \rho_{t}-\rands \rho_{t}^{*}$ for all $t\in\mathcal{T}$.
\par
It is important to note that we \emph{define} the measurement error. We do not make any assumptions about how the difference between observed and true quantities arises (i.e., we allow for measurement error to be multiplicative or additive).\footnote{Formally, this makes the support $W$ depend on the support of both the observed and the true quantities. For simplicity we omit this dependency from the notation.} Moreover, we \textbf{do not} need to assume that measurement error is independent of other variables, independent within time periods, or independent across goods. 
\par
By Lemma~\ref{lem:RandomExponentialdiscountingCstar} we can immediately conclude that the observed $\rand{x}=(\rands{\rho}_{t},\rand c_{t})_{t\in\mathcal{T}}$ can be  $\s/\m$-rationalized  if and only if there exist $(\rands{\lambda}_t, \rands{\delta}_t,\rand v_t,\rand{w}_t)_{t\in\mathcal{T}}$, with $(\rands{\lambda}_t)_{t\in\mathcal{T}}$ supported on or inside $\Lambda$, and $(\rands{\delta}_t)_{t\in\mathcal{T}}$ supported on or inside $\Delta$ such that
\[
\rand{v}_t-\rand{v}_s\geq\dfrac{\rands{\lambda}_t}{\rands{\delta}_t}(\rands{\rho}_{t}-\rand w^p_t)\tr(\rand{c}_{t}-\rand{c}_{s}+\rand{w}^c_{s}-\rand{w}^c_{t})\quad\as,\forall s,t\in\mathcal{T}.
\]
However, we know that without restrictions on the distribution of measurement error, RP tests have no power. That is, there always exists a measurement error $\rand w$ such that the observed $\rand x$ is consistent with $\s/\m$-rationalizability. Hence, we require access to additional validation information about measurement error. The source of measurement error is different in different applications. That is why in this section we formulate a general restriction on the measurement error distribution that can be tailored for a given empirical application. 
\par
Recall that $\rand{x}\in X$ denotes observed quantities. Let $\rand{e}=(\rands\lambda\tr,\rands{\delta}\tr,\rand v\tr,\rand w\tr)\tr\in E|X$ denote the vector of latent random variables, supported on or inside the conditional support $E|X$. We say that a mapping $g_{M}:X\times E|X\to\Real^{d_{M}}$ is a \emph{measurement error moment}. We only require the following condition on measurement error.

\begin{assumption}[Centered Measurement Error]\label{assu: general measurement error}
    (i) The random vector $\rand{e}$ is supported on or inside the known support $E|X$. (ii) There exists a known measurement error moment $g_{M}:X\times E|X\to\Real^{d_{M}}$ such that
    \[
    \Exp{g_M(\rand{x},\rand{e})}=0.
    \]
\end{assumption}

The choice of function $g_M$ depends on the application and the assumptions the researcher is willing to make on the basis of the available knowledge about the nature of measurement error. In Section~\ref{sec:error in data} we provide examples of moment conditions in data sets that are often used in the RP literature. The objects of interest for us are measurement error in consumption, expenditure, and prices.

\section{Recoverability and Counterfactuals}\label{sec:recoverabilitycounterfactual}
\citet{varian1982nonparametric,varian1984nonparametric} exploits the
connections between empirical content and
counterfactuals. In particular, \citet{varian1982nonparametric} seems to be the first to think of nonparametric counterfactual analysis as specification testing.\footnote{Recent work building on these connections includes \citet{blundell2003nonparametric} and \citet{allen2019identification} in demand analysis, and \citet{norets2014semiparametric} in the analysis of dynamic binary choice models.} Following these ideas, our formulation of rationalizability allows us to answer questions about the recoverability of, and counterfactual predictions for, different objects of interest. 
\par
In Section~\ref{sunsec:Recoverability} we show how to recover different quantities of interest (e.g., average true consumption at a given $t=\tau$) from the $\s/\m$-rationalizable data set. In Section~\ref{subsec:Counterfactuals} we demonstrate how to make out-of-sample predictions for expected consumption in a way that is analogous to \citet{blundell2014bounding}. In the presence of measurement error, distributional information about the primitives of the model of interest  is inevitably lost. Hence, we cannot apply the traditional approach proposed by \citet{varian1982nonparametric} to recover preferences and to do counterfactual analysis on an individual basis. Instead, we use this section to pose questions about the primitives of the model at the level of the population.

\subsection{Recoverability}\label{sunsec:Recoverability}
Assume that $\rand{x}=(\rands{\rho}_{t},\rand c_{t})_{t\in\mathcal{T}}$ can be  $\s/\m$-rationalized and Assumption~\ref{assu: general measurement error} holds. Suppose that there is a finite-dimensional parameter of interest $\theta_0\in\Theta$, where $\Theta$ is a compact subset of the Euclidean space. The parameter of interest is related to the model via the user-specified moment condition
\[
\Exp{g_{R}(\rand{x},\rand{e};\theta_0)}=0\in\Real^{d_{R}}.
\]
The function $g_{R}$ can take different forms depending on the different questions the user wants to answer. We provide some examples here.
\begin{example} [Expected True Consumption and Expected True Consumption Change]  If $\theta_0$ is the expected true consumption at $t=\tau$, then $g_{R}(x,e;\theta_0)=c_{\tau}-w^c_{\tau}-\theta_0$. If $\theta_0$ is the expected difference in true consumption at $t=\tau+1$ and $t=\tau$, then  $g_{R}(x,e;\theta_0)=c_{\tau+1}-w^c_{\tau+1}-c_{\tau}+w^c_{\tau}-\theta_0$. 
\end{example}
The user may also be interested in testing the joint null hypothesis that (i) the consumer is $\s/\ed$-rationalizable and (ii) the random discount factor distribution has certain properties. 
\begin{example}[Average Random Discount Factor] The user may be interested in testing whether the average value of the random discount factor is equal to a certain fixed value, in which case $g_{R}(x,e;\theta_0)=d-\theta_0$.
\end{example}
In addition, our framework allows us to have, as a special case, latent random variables with flexible support. 
\begin{example}[Support of the Random Discount Factor] The user may be interested in whether the random time-discount factor $\rand{d}$ has a support on or inside $[\theta_{01},\theta_{02}]\subseteq(0,1]$. Then, for $\theta_0=(\theta_{01},\theta_{02})\tr$, one can define $g_{R}(x,e;\theta_0)=\Char{\theta_{01}\leq d\leq\theta_{02}}-1$.
\end{example}

\subsection{Counterfactual Out-of-Sample Predictions}\label{subsec:Counterfactuals}
We consider a counterfactual situation in which the user is given an
out-of-sample effective random price vector $\rands{\rho}^*_{T+1}$ (supported in $\Real_{++}^L$), a data set $\rand{x}=(\rands{\rho}_{t},\rand c_{t})_{t\in\mathcal{T}}$ such that Assumption~\ref{assu: general measurement error} holds, and she then asks two related questions. 
First, the user wants to know if there exists a counterfactual random consumption vector $\rand c^*_{T+1}$ such that the augmented random array $\{(\rands{\rho}^*_{t},\rand{c}^*_{t})_{t\in\mathcal{T}}, (\rands{\rho}^*_{T+1},\rand{c}^*_{T+1})\}$
is approximately $\s/\m$-rationalizable, where $\rand{c}^*_t=\rand{c}_t-\rand{w}^c_t$ and $\rands{\rho}^*_t=\rands{\rho}_t-\rand{w}^p_t$.  
\par
Second, if the answer to the first question is affirmative, then the user will be interested in constructing confidence sets for some counterfactual finite-dimensional parameter $\theta_0\in\Theta$. 
The parameter $\theta_0$ satisfies the user-specified moment condition 
\[
\Exp{g_C\left((\rands{\rho}^*_{t},\rand c^*_{t})_{t\in\mathcal{T}},\rand{c}^*_{T+1};\rands{\rho}^*_{T+1},\theta_0\right)}=0\in\Real^{d_{C}}.
\]
Both questions can be answered simultaneously with our characterization of $\s/\m$-rationalizability. Observe that the answer to the first question is negative if the random array $(\rands{\rho}^*_{t},\rand c^*_{t})_{t\in\mathcal{T}}$ is not $\s/\m$-rationalizable. In contrast, if the random array $(\rands{\rho}^*_{t},\rand c^*_{t})_{t\in\mathcal{T}}$ is $\s/\m$-rationalizable, then the counterfactual exercise is equivalent to checking that the counterfactual price/consumption distribution is simultaneously compatible with $\s/\m$-rationalizability and the user-specified moment condition. Formally, to answer both questions, we define what it means for a random array $(\rands{\rho}^*_{t},\rand c^*_{t})_{t\in\mathcal{T}}$ to be counterfactually rationalizable ($\cf/\m$-rationalizability) for a given $\rands{\rho}^*_{T+1}$, $\theta_0$, and $g_C$.

\begin{defn}[$\cf/\m$-rationalizability]\label{def:counterfactuals}
For a given $\rands{\rho}^*_{T+1}$, $g_C$, and $\theta_{0}$, a random array $(\rands{\rho}^*_{t},\rand c^*_{t})_{t\in\mathcal{T}}$ is approximately $\cf/\m$-rationalizable if there exist $\rand{c}^*_{T+1}$ such that 
\begin{enumerate}
\item The augmented random array $\{(\rands{\rho}^*_{t},\rand c^*_{t})_{t\in\mathcal{T}},(\rands{\rho}^*_{T+1},\rand{c}^*_{T+1})\}$ is approximately $\s/\m$-rationalizable;
\item $\Exp{g_C((\rands{\rho}^*_{t},\rand c^*_{t})_{t\in\mathcal{T}},\rand{c}^*_{T+1};\rands{\rho}^*_{T+1},\theta_0)}=0$.
\end{enumerate}
\end{defn}


Observe that if a random array $(\rands{\rho}^*_{t},\rand c^*_{t})_{t\in\mathcal{T}}$ is $\cf/\m$-rationalizable for a given $(\rands{\rho}^*_{T+1},\rand{c}_{T+1}^*)$ and $\theta_0$, then it is also $\s/\m$-rationalizable. However, the opposite is not always true. 
\par
We can apply Lemma~\ref{lem:RandomExponentialdiscountingCstar} to Definition~\ref{def:counterfactuals} and get an extended system of the RP inequalities coupled with the counterfactual moment conditions $g_C$. Moreover, we can define an identified set for counterfactual parameter values $\Theta_0$. Formally,
\[
\Theta_0=\left\{\theta_0\in\Theta\::\: \text{$(\rands{\rho}^*_{t},\rand c^*_{t})_{t\in\mathcal{T}}$ is approximately $\cf/\m$-rationalizable given $\rands{\rho}^*_{T+1}$, $\theta_0$, and $g_C$}\right\}.
\]
We highlight that our framework can accommodate additional support restrictions on the counterfactual objects. A classical support constraint is a target out-of-sample expenditure level (i.e., $\rands{\rho}^{*\prime}_{T+1}\rand{c}^{*}_{T+1}=1\:\as$) as in \citet{varian1982nonparametric}. We omit these constraints from our discussion to simplify exposition. 
\begin{example}[Average Varian Support Set]
We consider a moment 
\[
g_C((\rho^*_{t},c^*_{t})_{t\in\mathcal{T}},c^*_{T+1};\rho^*_{T+1},\theta_0)=c^*_{T+1}-\theta_0,
\]
with $\theta_0\in\Theta=\Real_{+}^L\setminus{\{0\}}$ as a hypothesized average-demand vector. Thus, $\Theta_0$ is the Average Varian Support Set.
Given $\rands{\rho}^*_{T+1}$ this set describes the bounds on the average demand that is compatible with the $\s/\m$-rationalizability of the random array $(\rands{\rho}^*_{t},\rand c^*_{t})_{t\in\mathcal{T}}$.
\end{example}

\begin{example}[Quantile Varian Support Set] For $\s/\rat$-rationalizability, we can consider the following moment condition:
\[ 
g_C((\rho^*_{t},c^*_{t})_{t\in\mathcal{T}},c^*_{T+1};\rho^*_{T+1},\theta)=\Char{\rho_{T+1}^{*\prime} c^*_{T+1}\leq \bar{e}_c}-\phi,
\]
where $\theta=(\bar e_c,\phi)\tr\in\Real_{++}\times[0,1]$, $\bar e_c$ is a fixed $\phi$-quantile of the counterfactual expenditure distribution. Next we can define the $\phi$-quantile Varian Support Set:
\[
\left\{c\in\Real_{+}^L\setminus\{0\}\::\:\rho_{T+1}\tr {c}=\bar e_c,\: (\bar e_c,\phi)\tr\in\Theta_0\right\}.
\]
This set describes the bounds of the counterfactual demand for a given $\rands{\rho}^*_{T+1}$ and $\phi$-quantile of $\rand{u}(\rand{c}^*_{T+1})$ that is compatible with $\s/\rat$-rationalizability.
\end{example}
\par
Some counterfactual questions (e.g., the Average Varian Support Set) lead to convex identified sets $\Theta_0$.
\begin{proposition}\label{prop:convexity}
If the parameter space $\Theta$ is convex and $g_{C}$ is such that
\[
g_C((\rho^*_{t},c^*_{t})_{t\in\mathcal{T}},c^*_{T+1};\rho^*_{T+1},\theta)=\tilde g((\rho^*_{t}, c^*_{t})_{t\in\mathcal{T}},c^*_{T+1};\rho^*_{T+1})-A((\rho^*_{t}, c^*_{t})_{t\in\mathcal{T}};\rho^*_{T+1})\tr\theta
\]
for some $\tilde g$ and $A$, then $\Theta_0$ is convex. 
\end{proposition}
Proposition~\ref{prop:convexity} imposes two restrictions on $g_C$: (i) additive separability between $\rand{c}^*_{T+1}$ and $\theta$; (ii) affinity of the moment condition in $\theta$. In Section~\ref{sec:econometric framework} we provide a framework to construct confidence sets for the counterfactual parameters by means of the test inversion. Convexity of $\Theta_0$ can substantially simplify the computation of the confidence sets since one does not need to conduct test inversion at every point of the parameter space. 
\section{Econometric Framework}\label{sec:econometric framework}
In Sections~\ref{subsec:edrum,meas er} and~\ref{subsec:Counterfactuals} we showed how testing, recoverability, and counterfactuals in RP models with measurement error can be framed in the form of moment conditions. In this section we recast the empirical content of the RP inequalities in a form amenable to statistical testing. To simplify the exposition we will focus on testing $\s/\m$-rationalizability in the presence of measurement error ($g_M$ and $g_I$ only). 
\subsection{Characterization of the Model via Moment Conditions}\label{subsec:edrum,moments}
First, we write a set of moment conditions that will summarize the empirical content of $\s/\m$-rationalizability.
Recall that $\rand x\in X$ denotes observed quantities and $\rand e=(\rands\lambda\tr,\rands \delta\tr,\rand v\tr,\rand w\tr)\tr\in E|X$ denote the vector of latent random variables. The support $E|X$ depends on the fixed supports $\Lambda$ and $\Delta$ that characterizes the particular model of interest. We use $\mathcal{P}_{X}$, $\mathcal{P}_{E,X}$, and $\mathcal{P}_{E|X}$ to denote the set of all probability measures defined over the support of $\rand x$, $(\rand e\tr,\rand x\tr)\tr$, and $\rand{e}|\rand{x}$, respectively. (Recall that the boldface font letters denote random objects.) Define the following moment functions:
\begin{align*}
g_{I,t,s}(\rand{x},\rand{e})&=\Char{\rand v_{t}-\rand v_{s}-\frac{\rands\lambda_t}{\rands{\delta}_t}(\rands{\rho}_{t}-\rand{w}_t^p)\tr[\rand c_{t}-\rand w^c_{t}-\rand c_{s}+\rand w^c_{s}]\geq 0}-1,\quad t\neq s\in\mathcal{T},\\
g(\rand{x},\rand{e})&=(g_{I}(\rand{x},\rand{e})\tr,g_{M}(\rand{x},\rand{e})\tr)\tr.
\end{align*}
We have $k=\abs{\mathcal T}^2-\abs{\mathcal T}$ and $q=d_M$ moment functions which correspond to inequality conditions ($g_{I}$) and the measurement error centering conditions ($g_{M}$), respectively.\footnote{If in addition, the user includes moments $g_{R}$ or $g_C$, then $q=d_M+d_R$ or $q=d_M+d_C$, respectively.} 
Define $\Expt[\mu\times\pi]{g(\rand{x},\rand{e})}=\int_{X}\int_{E|X}g(x,e)d\mu d\pi$, where $\mu\in\mathcal{P}_{E|X}$ and $\pi\in\mathcal{P}_{X}$.
\begin{thm}
	\label{thm:DistributionalRP}The following are equivalent:
	\begin{enumerate}
		\item A random vector $\rand x=(\rands{\rho}_{t},\rand c_{t})_{t\in\mathcal{T}}$
		is approximately $\s/\m$-rationalizable  such that Assumption
		\ref{assu: general measurement error} holds.
		\item
		\[
		\inf_{\mu\in\mathcal{P}_{E|X}}\norm{\Expt[\mu\times\pi_0]{g(\rand x,\rand e)}}=0,
		\]
		where $\pi_0\in\mathcal{P}_{X}$ is the observed distribution of $\rand x$.
	\end{enumerate}
\end{thm}

Theorem~\ref{thm:DistributionalRP} establishes the equivalence between (i) $\s/\m$-rationalizability with the centered measurement error condition and (ii) a system of moment conditions. In other words, the observed consumption pattern, captured by the random array $(\rands{\rho}_{t},\rand c_{t})_{t\in\mathcal{T}}$, can be $\s/\m$-rationalized under the restrictions on measurement error if and only if there exists a distribution of latent variables conditional on observables that satisfies the RP inequalities with probability 1, for the given supports $\Lambda$ and $\Delta$.
\par
Our notion of $\s/\m$-rationalizability makes clear that when one is dealing
with measurement error, no RP test can decide whether a finite sample is consistent with model $\m$. We can decide only that the data set is asymptotically consistent with the model
as the sample size goes to infinity. Moreover, even asymptotically, there is no way to differentiate between the notion of approximate $\s/\m$-rationalizability and the notion of exact $\s/\m$-rationalizability. Nonetheless, we can do traditional hypothesis testing and decide at a fixed significance level whether we reject the null hypothesis of (approximate) model $\m$ consistency under Assumption~\ref{assu: general measurement error} for a given sample. Conceptually, our notion of rationalizability corresponds to the extended notion of an identified set in \citet{schennach2014entropic}.
\par
Note that the test is not yet formally established. We have a set of latent random variables $\rand e$ distributed according to an unknown $\mu\in\mathcal{P}_{E|X}$. This problem can be solved nonparametrically using the Entropic Latent Variable Integration via Simulation (ELVIS) of \citet{schennach2014entropic}. The main advantage of the ELVIS approach is that it allows us to formulate a test that can be implemented in panel data sets suffering from measurement error of the type described only in terms of observables.

\subsection{ELVIS and Its Implications for Testing and Inference}\label{sec:elvis}
We start this section by showing how the nonparametric results of Theorem~\ref{thm:DistributionalRP} can be used to construct a set of (equivalent) parametric maximum-entropy moment conditions using \citet{schennach2014entropic}. Next, we provide a semi-analytic solution to the these conditions. Finally, we propose a procedure to test for $\s/\m$-rationalizability.
\par
Following \citet{schennach2014entropic}, we define the maximum-entropy moment as follows.
\begin{defn}[Maximum-Entropy Moment]\label{def:MEM}
    The maximum-entropy moment of the moment $g(x,\cdot)$,
	for a fixed $x$, is
	\[
	h(x;\gamma)=\dfrac{\int_{e\in E|X}g(x,e)\exp(\gamma\tr g(x,e))d\eta(e|x)}{\int_{e\in E|X}\exp(\gamma\tr g(x,e))d\eta(e|x)},
	\]
	where $\gamma\in\Real^{k+q}$ is a nuisance parameter, and $\eta\in\mathcal{P}_{E|X}$ is an arbitrary user-input distribution function supported on $E|X$ such that $\Expt[\pi_0]{\log \Expt[\eta]{\exp(\gamma\tr g(\rand x,\rand e))|\rand x}}$
	exists and is twice continuously differentiable in $\gamma$ for all
	$\gamma\in\mathbb{R}^{k+q}$.
\end{defn}
Note that 
\[
\left\{d\eta^*(\cdot|x;\gamma)=\dfrac{\exp(\gamma\tr g(x,\cdot))d\eta(\cdot|x)}{\int_{e\in E|X}\exp(\gamma\tr g(x,e))d\eta(e|x)},\:\gamma\in\Real^{k+q}\right\}
\]
is a family of exponential conditional probability measures. Thus, the maximum-entropy moment $h$ is the marginal moment of the function $g$,
at which the latent variable has been integrated out using one of the members from the above exponential family. 
The importance of the maximum-entropy moment is captured in the following result.

\begin{thm}\label{thm:ELVISexponentialdiscounting}
The following are equivalent:
	\begin{enumerate}
		\item A random array $\rand x=(\rands{\rho}_{t},\rand c_{t})_{t\in\mathcal{T}}$
		is approximately $\s/\m$-rationalizable such that Assumption
		\ref{assu: general measurement error} holds.
		\item 
		\[
		\inf_{\gamma\in\mathbb{R}^{k+q}}\norm{\Expt[\pi_0]{h(\rand x;\gamma)}}=0,
		\]
		where $\pi_0\in\boldsymbol{P}_{X}$ is the observed distribution
		of $\rand x$.
	\end{enumerate}
\end{thm}

The idea behind Theorem~\ref{thm:ELVISexponentialdiscounting} is that if there exists a distribution that satisfies the moment condition, then there must be another distribution that (i) belongs to a particular finite-dimensional exponential family and (ii) satisfies the same moment condition. Since we are only interested in the existence of the former distribution, instead of searching over the set of all possible distributions we are going to only search over this ``smaller'' exponential family.
\par
We emphasize that Theorem~\ref{thm:ELVISexponentialdiscounting} provides both \emph{necessary and sufficient} conditions for the observed data to be (approximately) $\s/\m$-rationalizable. This represents an important gain in power with respect to any of the averaging-based tests of RP models that are usually used in the presence of measurement error.
\par
High-level technical assumptions can ensure that the sequence of random latent variables that approximates model $\m$ converges to a proper random variable. Thus, this limiting random variable would ensure (i) that the infimum in Theorem~\ref{thm:ELVISexponentialdiscounting} is attained, and (ii) that the notion of approximate rationalizability collapses to exact rationalizability. However, this obscures the fact that any assumption made in that direction has no testable implications.
\par
The remarkable advantage of applying the results of \citet{schennach2014entropic} to the RP approach is that it marginalizes out the latent random variables. More importantly, we have a robust statistical framework with which to test our models in the presence of measurement error. In particular, we have not made any strong distributional assumptions about $\rands{\lambda}$, $\rands{\delta}$ or $\rand u$ (the heterogeneous tastes). The only restrictions are the concavity assumption on the utility function, and a centering condition on measurement error.\footnote{At this point, we can use as an alternative the methodology presented by \citet{ekeland2010optimal} to deal with latent variables in our moment conditions.}
\begin{rem}
	Theorem~\ref{thm:ELVISexponentialdiscounting} does not imply that the distribution of the latent variables (or their support) is point-identified. In fact, it will always be set-identified.
\end{rem}
\begin{rem}
    The maximum-entropy moment implicitly depends on the choice of the user-specified distribution $\eta$. However, $\eta$ does not have any effects on the set of values $\Expt[\pi_0]{h(\rand x;\cdot)}$ or its sample analogue can take. That is, the choice of $\eta$ will not affect the value $\inf_{\gamma\in\mathbb{R}^{k+q}}\norm{\Expt[\pi_0]{h(\rand x;\gamma)}}$ takes both asymptotically and in finite samples. The choice of $\eta$ affects \emph{only} the nuisance parameter ($\gamma$) value. See Remark 2.3 in \citet{schennach2014entropic} for further details. 
\end{rem}
\subsection{Semi-analytic Solution for the Maximum-entropy Moment}\label{subsec:elvis,mesurm}
One can directly employ the maximum-entropy moment in Theorem~\ref{thm:ELVISexponentialdiscounting} to test model $\m$. However, doing so is potentially problematic. One possible concern is the fact that the number of maximum-entropy moments corresponding to the $(g_I)$ conditions, $k=\abs{\mathcal{T}}^{2}-\abs{\mathcal{T}}$, grows quadratically with $\abs{\mathcal{T}}$. Moreover, $\gamma_0$, the nuisance parameter value at which infimum is achieved, may be set-identified when unbounded (e.g., some of the components of $\gamma_0$ may be equal to infinity\footnote{By ``equal to infinity'' we mean that the infimum is achieved along the sequence of $\gamma$ that diverges to infinity along some coordinates.}), which would therefore lead to nonstandard testing procedures.
\par
Here we show that there exists a semi-analytic solution to the optimization problem where every component of $\gamma_0$ that corresponds to the RP inequality constraints is equal to $+\infty$, and every component of $\gamma_0$ that corresponds to the measurement error centering constraint is finite and unique. Thus, for testing purposes (under the null hypothesis of model $\m$), we can minimize an objective function over a parameter space of lower dimensionality.
\begin{assumption}\label{assu:nondegenerate w} (Nondegeneracy)
    There exist two subsets of $E|X$, $E'$ and $E''$, with a positive measure, such that componentwise the measurement error moment is such that $\sup_{e\in E'}g_M(\rand{x},e)<0<\inf_{e\in E''}g_M(\rand{x},e)$ with positive probability.  
\end{assumption}
\begin{assumption}\label{assu:bounded support} (Bounded support)
	The random array  $\rand{x}=(\rands{\rho}_t,\rand{c}_t)_{t\in\mathcal{T}}$ has a bounded support.
\end{assumption}
Assumption~\ref{assu:nondegenerate w} rules out cases in which there is no measurement error and allows us to a have a unique minimizer of the objective function. It can be relaxed since our methodology still works for cases without measurement error, but in those cases, it is preferable to use the equivalent deterministic RP benchmark. Assumption~\ref{assu:bounded support} is made to simplify the analysis and can be replaced by tail restrictions on the distribution of $\rand{x}$.
\par
Note that the user-specified distribution $\eta$ should obey the same restrictions as the unknown distribution of latent $\rand{e}$. Thus, we impose the following restrictions on $\eta$:
\begin{defn}[User-specified distribution]\label{assu:eta}
Almost surely in $\rand{x}$, the user-specified distribution $\eta(\cdot|\rand{x})$ satisfies all of the following:
\begin{enumerate}
	\item The set $\tilde E|X=\{e\in E|X\::\:g_I(\rand{x},e)=0\}$ has a positive measure under $\eta(\cdot|\rand{x})$.
	\item There exist two subsets of $\tilde E|X$, $E'$ and $E''$, with a positive measure under $\eta(\cdot|\rand{x})$, such that componentwise $\sup_{e\in E'}g_M(\rand{x},e)<0<\inf_{e\in E''}g_M(\rand{x},e)$.
	\item For every finite $\gamma_M\in \Real^{q}$,
	\[
	{\int_{E|X}\norm{g_M(\rand{x},e)}^2\exp(\gamma_M\tr g_M(\rand{x},e))d\eta(e|\rand{x})}<\infty.
	\]
\end{enumerate}
\end{defn}
The first condition in Definition~\ref{assu:eta} requires that the support of $\eta$ allows the inequalities to be satisfied. The second and third conditions are regularity conditions. This is a definition and not an assumption, as we can always construct an allowable $\eta$.\footnote{\citet{schennach2014entropic} provides a generic construction that can be used here.} We are ready to present our main result.

\begin{thm}\label{thm:semianal} Given a user-specified measure $\eta$ that satisfies the three conditions in Definition~\ref{assu:eta}, the following are equivalent:
	\begin{enumerate}
		\item A random array $\rand x=(\rands{\rho}_{t},\rand c_{t})_{t\in\mathcal{T}}$ is approximately consistent with $\s/\m$-rationalizability such that Assumptions~\ref{assu: general measurement error},~\ref{assu:nondegenerate w}, and~\ref{assu:bounded support} hold.
		\item For any sequence $\{\gamma_{I,l}\}_{l=1}^{+\infty}$ that componentwise diverges to $+\infty$,
		\begin{equation}\label{eq:full}
		\lim_{l\to+\infty}\min_{\gamma_M\in\mathbb{R}^{q}}\norm{\Expt[\pi_0]{h(\rand x;(\gamma_{I,l}\tr,\gamma_M\tr)\tr)}}=0.
		\end{equation}
		The sequence of minimizers of (\ref{eq:full}), $\{\gamma_{M,l}\}$, converges to some finite $\gamma_{0,M}$ that does not depend on $\{\gamma_{I,l}\}_{l=1}^{+\infty}$.
		\item 
		\begin{equation}\label{eq:partial}
		\min_{\gamma_M\in\mathbb{R}^{q}}\norm{\Expt[\pi_0]{\tilde h_M(\rand x;\gamma_M)}}=0,
		\end{equation}
		where
		\[
		\tilde h_M(x;\gamma)=\dfrac{\int_{e\in{E}|{X}}g_M(x,e)\exp(\gamma\tr g_M(x,e))\Char{g_{I}(x,e)=0}d\eta(e|x)}{\int_{e\in{E}|{X}}\exp(\gamma\tr g_M(x,e))\Char{g_{I}(x,e)=0}d\eta(e|x)}.
		\]
		Moreover, the minimizer of (\ref{eq:partial}) is finite, and is equal to $\gamma_{0,M}$.
	\end{enumerate}
\end{thm}
The intuition behind Theorem~\ref{thm:semianal} is that the RP inequalities presented here restrict only the conditional support of the latent variables (including the measurement error). Hence, given the support restrictions captured by the RP inequalities, only the centering condition comes in the form of moments.\footnote{If Assumption~\ref{assu:nondegenerate w} is violated for some component $j$ of $g_M$ (e.g., $\sup_{E'}g_{M,j}(x,e)\geq0\:\as$ for all $E'\subseteq E|X$), then since $\Exp{g_{M,j}(\rand{x},\rand{e})}=0$ it has to be the case that $g_{M,j}(\rand{e},\rand{x})=0\:\as$. Thus, $g_{M,j}$ becomes a condition similar to $g_{I}$ that can be replaced by a support restriction.}
\par
Theorem~\ref{thm:semianal} substantially simplifies the conclusion of Theorem~\ref{thm:ELVISexponentialdiscounting}. First, we need to minimize the objective function over a much smaller parameter space ($\Real^q$ instead of $\Real^{k+q}$). Thus, the problem becomes computationally tractable. Second, if the data is consistent with $\s/\m$-rationalizability, then the minimizer, $\gamma_{0,M}$, has to be \emph{finite and unique}. Finally, to compute $\tilde{h}_{M}(x,\gamma_M)$ in applications, one may need to use Markov Chain Monte Carlo (MCMC) methods by sampling from $\eta$. Theorem~\ref{thm:semianal} implies that it suffices to sample from $d\tilde{\eta}(\cdot|x)=\Char{g_{I}(x,\cdot)=0}d\eta(\cdot|x)$. The straightforward way to sample from $\tilde\eta$ is to sample from $\eta$ and then reject the draw if it does not satisfy the RP inequalities captured by $\Char{g_{I}(x,\cdot)=0}$. The last part usually amounts to solving a linear program for $\rat$-rationalizability. For the case where this rejection sampling is not efficient, in Appendix~\ref{appen: coputational aspects}, we describe a new hit-and-run algorithm to sample from $\tilde{\eta}(\cdot|x)$ directly.\footnote{The classical hit-and-run algorithm is an efficient MCMC method that generates uniform draws from a convex polytope. It initiates inside the polytope and proceeds by randomizing directions and then advancing a random distance while remaining inside the set.} This approach is particularly useful for testing $\s/\ed$-rationalizability in one of our applications.
\subsection{Testing}\label{subsec:elvis,testing}
Theorem~\ref{thm:semianal} provides moment conditions that are necessary and sufficient for the data $\{\rand x_{i}\}_{i=1}^n = \{(\rands{\rho}_{t,i},\rand c_{t,i})_{t\in\mathcal{T}}\}_{i=1}^n$ (where $n$ is the sample size), to be approximately consistent with $\s/\m$-rationalizability. Now, define the following sample analogues of the maximum-entropy moment and its variance:
\begin{align*}
\hat{\tilde h}_M(\gamma) & =\frac{1}{n}\sum_{i=1}^{n}\tilde h_M(\rand x_{i},\gamma);\\
\hat{\tilde\Omega}(\gamma) & =\frac{1}{n}\sum_{i=1}^{n}\tilde h_M(\rand x_{i},\gamma)\tilde h_M(\rand x_{i},\gamma)\tr-\hat{\tilde h}_M(\gamma)\hat{\tilde h}_M(\gamma)\tr.
\end{align*}
Let $\Omega^{-}$ denote the generalized inverse of the matrix $\Omega$. The testing procedure we propose is due to \citet{schennach2014entropic}
and is based on this test statistic:
\[
\mathrm{TS}_{n}=n\inf_{\gamma\in{\Real}^{q}}\hat{\tilde h}_M(\gamma)\tr\hat{\tilde\Omega}^{-}(\gamma)\hat{\tilde h}_M(\gamma).
\]
\begin{assumption}
	\label{ass: iid}The data $\{\rand{x}_{i}\}_{i=1}^{n}$ is i.i.d.
\end{assumption}

\begin{thm}\label{thm:test} Suppose Assumptions~\ref{assu: general measurement error}, \ref{assu:nondegenerate w}, \ref{assu:bounded support}, and~\ref{ass: iid} hold. Then under the null hypothesis that the data is approximately consistent with $\s/\m$-rationalizability, it follows that
	\begin{align*}
	\lim_{n\to\infty}\Prob{\mathrm{TS}_{n}>\chi_{q,1-\alpha}^{2}}\leq\alpha,
	\end{align*}
	for every $\alpha\in(0,1)$.
	\par
	If, moreover, the minimal eigenvalue of the variance matrix $\mathbb{V}[\tilde h_M(\rand x,\gamma)]$ is uniformly, in $\gamma$, bounded away from zero and the maximal eigenvalue of $\mathbb{V}[\tilde h_M(\rand x,\gamma)]$ is uniformly, in $\gamma$, bounded from above, then, 
	under the alternative hypothesis that the data is not approximately consistent with $\s/\m$-rationalizability, it follows that
	\[
	\lim_{n\to\infty}\Prob{\mathrm{TS}_{n}>\chi_{q,1-\alpha}^{2}}=1.
	\]
\end{thm}
We conclude this section by noting that we can speed up computations of the test statistic by obtaining an initial guess for the minimizer of the objective function efficiently. We make use of the fact that, although the objective function $\hat{\tilde h}_M(\cdot)\tr\hat{\tilde\Omega}^{-}(\cdot)\hat{\tilde h}_M(\cdot)$ may have several local minima, the quadratic form $\hat{\tilde h}_M(\cdot)\tr B \hat{\tilde h}_M(\cdot)$, where $B$ is any conformable positive definite matrix, has a unique global minimum and has no other local minima.\footnote{Additional details about the computational aspects of our methodology can be found in Appendix~\ref{appen: coputational aspects}.}

\subsection{Confidence Sets for Parameters of Interest}
The above testing procedure can be modified for construction of confidence sets for parameter $\theta_0$ from Section~\ref{subsec:Counterfactuals}. In particular, recall that one just needs to extend the set of the original moment conditions (the centering condition $g_M$ and the RP inequalities $g_I$) by $g_R$ or $g_C$ and add (if needed) extra RP inequalities that correspond to $(\rands{\rho}^*_{T+1}, \rand{c}^*_{T+1})$.  
As in Section~\ref{subsec:elvis,testing} we can then define $\mathrm{TS}_{n}(\theta)$ as the value of the test statistic computed for a fixed value of $\theta$. 
Under assumptions similar to those of Theorem~\ref{thm:test}, the confidence set for $\theta_0$ can be obtained by inverting $\mathrm{TS}_{n}(\theta_{0})$.
That is, the $(1-\alpha)$-confidence set for $\theta_0$ is 
\[
\{\theta_{0}\in\Theta:\mathrm{TS}_{n}(\theta_{0})\leq\chi_{q_{\mathrm{ext}},1-\alpha}^{2}\},
\]
where $\chi_{q_{\mathrm{ext}},1-\alpha}^{2}$ denotes the $(1-\alpha)$ quantile
of the $\chi^{2}$ distribution with $q_{\mathrm{ext}}$ degrees of freedom ($\chi_{q_{\mathrm{ext}}}^{2}$). The number $q_{\mathrm{ext}}$ is determined by the number of the non RP moment conditions in the extended system that are not the support restrictions.
Note that we do not pretest for $\s/\m$-rationalizability in order to construct the confidence set for $\theta_0$. If the data set is not $\s/\m$-rationalizable, then the confidence set will be empty asymptotically.

\section{Measurement Error in Different Data Sets}\label{sec:error in data}
Recall that to apply our methodology we need measurement-error restrictions that come in the form of moments $g_{M}$ (Assumption~\ref{assu: general measurement error}). These moments capture the knowledge the user has about measurement error in a particular data set. In this section we provide examples of such centering conditions in different data sets relevant for the RP literature.
\subsection{Measurement Error in Survey Data Sets}
Measurement error in surveys may arise because of errors due to the respondent, the interviewer, or the survey design itself \citep{carroll2014introduction}. Recall mistakes, social desirability, and recording errors, among other factors, may cause measurement error in surveys (see \citealp{meyer2015household}). Household surveys usually measure expenditure across different goods, but do not measure consumption or prices directly. Prices are sometimes price indexes constructed by a national statistics agency.\footnote{For instance, in data sets like the ones used in \citet{blundell2003nonparametric,adams2014consume}, and \citet{kitamura2018nonparametric}, price indexes are collected through direct observation in a given market and then merged with the expenditure survey data sets.} Consumption is generated by dividing reported consumption expenditures in surveys by the observed prices aggregated to a category of goods (e.g., \citealp{blundell2003nonparametric,adams2014consume}, and \citealp{kitamura2018nonparametric}). There is evidence that individuals, in either recall or self-reported surveys, can systematically overreport or underreport their expenditure on different categories of goods \citep{mathiowetz2002measurement, carroll2014introduction, pistaferri2015household}.
In sum, both prices and consumption may have measurement error.
\par
Consumption quantities are generated by dividing expenditures for a given good by its price index. This means that in the survey environment, measurement error in consumption is often \emph{nonclassical} because it may combine measurement error of expenditures and prices in a nonlinear way \citep{attanasio2014nonclassical}. Nonetheless, we will argue that average total expenditures is well-measured in some surveys. 
\par
\citet{abildgren2018consistency} find that mean total expenditures are usually well-measured in the Danish interview-based household budgets survey (DHS, 2015). They use administrative registry information of the same households interviewed in the survey of interest to conclude that there are no statistically significant differences between the mean total expenditure in Danish households as reported in the registry data set and in the household survey. In addition, \citet{kolsrud2017studying}, using administrative registry consumption data from Sweden, also finds that mean household expenditures match those computed on the basis of household surveys. One potential caveat to this evidence is that registry data itself may be mismeasured. However, direct evidence on household expenditures from retail data shows that total household expenditure error in registry data has zero mean \citep{baker2018measurement}.\footnote{This pattern of measurement error in expenditure is also indirectly supported by evidence on measurement error in income. For instance, \citet{angel2018differences}, using data from Austria, show that the difference in household income reported from surveys and that from registry data displays mean-reverting behavior, as high income levels are underreported and low income levels are overreported. This evidence is consistent with the literature on measurement error in earnings \citep{bound_chapter_2001}.}
\par
Based on this evidence for measurement error, the analyst can impose the following condition. For all $t\in\mathcal{T}$ it must be the case that
\begin{equation}\label{eq:centeredexpenditure}
    \Exp{\rands{\rho}_{t}\tr\rand{c}_{t}}=\Exp{\rands{\rho^*}_{t}\tr\rand{c}^*_{t}}.
\end{equation}

Equation~(\ref{eq:centeredexpenditure}) provides a measurement error moment that allows for nonclassical measurement error in consumption in the survey environment. In particular, it requires only that total expenditure measurement error be nonsystematic.    
Equation~(\ref{eq:centeredexpenditure}) implies that measurement error in consumption does not alter the mean value of total expenditures. In other words, it captures the idea that consumers, on average, may remember the total expenditure level better than the actual details. 
\par
Other surveys, such as the Spanish Continuous Expenditure Survey (1985-1997), used in \citet{beatty2011howdemanding}, \citet{adams2014consume}, and in one of our applications, have documented both over- and underreporting of income. Unfortunately, there is no validation data on expenditure itself (the object of our interest).  However, \citet{gradin2008inequality} report that for the time period of interest the evolution of the mean household income and the mean household expenditure for Spain are very similar. 
\par
Spanish, Swedish, and Danish households self-report expenditures for a number of good categories in their respective household surveys. We see no reason to believe that Spanish subjects will behave differently from Danish ones, given the similar format of the surveys. For that reason, we believe it is reasonable to assume that the mean expenditure is usually well measured in the Spanish household survey.   
\par
However, in practice, the condition in Equation~(\ref{eq:centeredexpenditure}) may not be restrictive enough (i.e., the joint hypothesis of this assumption and  $\s/\m$-rationality may not be rejected in general). We remind the reader that the measurement error moment is not the result of a modeling decision but should capture additional information that is usually obtained from validation exercises. Consequently, if we know too little about measurement error, then our data will not suffice to ensure a reasonable statistical power. A natural solution to this problem is to provide additional restrictions on measurement error by collecting new information about it. These additional restrictions will necessarily apply to specific settings or data sets. 
\par
A particular case of the condition captured by Equation (\ref{eq:centeredexpenditure}), with the additional support constraint that prices are perfectly measured (i.e., $\rands{\rho}^*_t=\rands{\rho}_t\:\as$ for all $t\in \mathcal{T}$), is:

\setcounter{assumption}{1}
\begin{subassumption}(Mean-Budget Neutrality for Survey Data Sets) \label{assu:(Measurement-error)Meanbudgetneutrality} 
For every $t\in \mathcal{T}$,
\[
\Exp{\rands{\rho}_{t}\tr\rand{w}^{c}_{t}}=0.
\]
\end{subassumption}

Assumption~\ref{assu: general measurement error}.\ref{assu:(Measurement-error)Meanbudgetneutrality} requires that measurement error in consumption be orthogonal to the effective prices. We maintain the assumption that prices are well measured in the rest of this subsection. We have three reasons for doing this. (i) This assumption is a common requirement in most of the work dedicated to the study of consumption decisions in households using survey data. In fact, we follow several works interested in modeling consumption in assuming that total expenditures and prices are measured better (i.e., with a higher signal-to-noise ratio) than the \textit{generated} consumption variable.\footnote{For examples of papers assuming that prices are perfectly measured and that measurement error is present only in consumption, see \citet{cochrane1991simple, ventura1994note, carroll2001death, ludvigson2001approximation, alanbrowning2009estimating, adams2014consume, carroll2014introduction, toda2015double}, and \citet{alan2018euler}.}  (ii) It is empirically plausible that prices are better measured than consumption.  In fact, there is evidence \citep{castillo1999medicion,guerrero2008sesgos} that the price index in Spain (which we use in one of our applications) is usually well measured with an estimated upward bias of the total price index of $0.07$ percent for the time window $1981-1991$ \citep{castillo1999medicion}, which is arguably small.\footnote{The bias is defined as the difference between the Laspeyres price index as used in Spain and the ideal Fisher price index \citep{diewert1998index}. This is usually known as the substitution or aggregation bias.} For context, the estimated bias for the US in the well-known Boskin report was $0.40$ percent for a similar time period \citep{boskin1998consumer}. In contrast, there is agreement that consumption measurement in surveys is very noisy \citep{alan2018euler}.\footnote{For instance, \citet{alan2010estimating} estimate that roughly $80$ percent of the variation in consumption growth rate is due to measurement error (this is for the PSID in the US).} (iii) We show that $\s/\m-$rationalizability is robust to measurement error when it is local (Appendix~\ref{appen: local perturb}). Hence, we will focus on the main source of measurement error that is consumption in the case of surveys. In sum, in the self-reported expenditures survey panel data sets, we argue that the \textit{main} source of measurement error is in the consumption variable. 
\par
The mean-budget neutrality assumption is still compatible with nonclassical measurement error in consumption, such as $\rand{w}^{c}_{l,t}\leq0\:\as$ ($\rand{w}^c_{l,t}\geq0\:\as$) for some good category $l$ and all decision tasks $t$. That is, this condition allows agents on average to underreport (overreport) consumption on some categories of goods as long as they on average equally overreport (underreport) consumption on other categories of goods. Mean-budget neutrality will fail if measurement error in expenditures is systematic (i.e., every consumer overreports expenditures or underreports expenditures simultaneously). 
\par
We want to highlight the fact that despite its simplicity, Assumption~\ref{assu: general measurement error}.\ref{assu:(Measurement-error)Meanbudgetneutrality} is a generalization of commonly held parametric assumptions about measurement error in the study of consumption. Here we provide some examples. 

\begin{example}[Multiplicative Measurement Error]
When one studies consumption of a single good across time (estimation of the Euler equation), one usually assumes that consumption measurement error is multiplicative \citep{alan2018euler}. Formally, $\rand{c}_t=\rand{c}_t^*\rands{\epsilon}_t$, where $\rands{\epsilon_t}$ can be assumed to be independent of, or at least orthogonal to, a set of instruments. For instance, we can assume that $\rands{\epsilon_t}$ is independent of true consumption conditional on effective prices and with mean $1$. In that case, Assumption~\ref{assu: general measurement error}.\ref{assu:(Measurement-error)Meanbudgetneutrality} holds. Alternatively, we can assume $\rands{\epsilon}_t$ to be independent of effective prices and true consumption, which implies Assumption~\ref{assu: general measurement error}.\ref{assu:(Measurement-error)Meanbudgetneutrality}. 
\end{example}
\par
\begin{example}[Additive Measurement Error] 
    Consider a case in which $\rand c_{t,l}=\rand c_{t,l}^{*}+\rands{\epsilon}_{t,l}$ for $l=1,\dots,L$; and where $\rands{\epsilon}_{t,l}\sim TN_{[-a,a]}(0,\sigma)$ for $l=1,\dots,L-1$ (from a truncated normal with variance $\sigma$ and bounds $[-a,a]$ for some positive $a>0$) such that $\rand c_{t,l}\geq0\:\as$. Note that this example is similar to the one used in \citet{varian1985non} with constraints to impose the nonnegativity of consumption. Note that  Assumption~\ref{assu: general measurement error}.\ref{assu:(Measurement-error)Meanbudgetneutrality}
    holds because measurement error is independent of prices and mean zero. 
\end{example}
\par
We conclude this section by noting that even if prices are measured correctly, there is a potential source of error coming from the fact that price indexes are the result of aggregation of commodities into categories. This aggregation implies that price indexes do not exactly reflect actual prices faced by consumers. Thus, the budget sets faced by the consumers may be different from the ones implied by price indexes.\footnote{We thank a referee for pointing this out.}
This problem is common to all demand analysis using survey data sets. Nonetheless, separability of the utility function or homotheticity within a category are examples of conditions under which commodity aggregation produces consumption and price indexes that are consistent with utility maximization \citep{jerison1994commodity, lewbel1996aggregation}. If one is willing to impose such conditions, then a rejection of utility maximization, even in the presence of commodity aggregation error, means that rationality can be rejected at the disaggregated level.\footnote{Recent work by \citet{sato2020large} provides nonparametric RP evidence in favor of weak separability and validity of several widely used price indexes for food and beverages categories using data from Japan. In addition, a survey by \citet{shumway2001does} finds evidence in favor of the conditions in \citet{lewbel1996aggregation}.} In addition, a nonrejection of the UMT means that price indexes and consumption behave  \textit{as if} they are rational. That is, the substitution patterns observed in the data after commodity aggregation and their relation with relative changes in price indexes can be summarized by the first-order condition approach.\footnote{Another possible interpretation of our analysis in survey environments is that we are testing for utility-maximizing behavior in a budgeting problem. In the budgeting problem, the consumer allocates her income into different good categories to maximize her utility at that level of aggregation. Then, the consumer maximizes utility within different categories.  This nested maximization problem is similar to mental accounting  \citep{thaler1985mental} in its modern interpretation by \citet{montgomery2019structural}.
}
\par
The aggregation error is not the only additional source of error in prices in survey data. For instance, since prices and expenditures are measured in different surveys, the observed prices in the first survey most likely are not those faced by some of the consumers in the second survey.\footnote{For instance, \citet{gaddis2016} documents that prices are sometimes collected by statistical offices in urban areas only, and that food prices can differ between urban and rural areas. However, price indexes constructed from these prices are matched with household expenditures surveys from both urban and rural areas.} Further investigation of these additional sources of error are left for future research.

\subsection{Measurement Error in Consumption or Prices in Experimental Data Sets}
Measurement error in consumption or prices in experimental data sets arise due to difficulties in eliciting the intended choices of the experimental subjects. Indeed, the experimental elicitation of choice may be subject to random variation due to (i) the subject's misperception of some elements of the task,  (ii) the level of understanding of the experimental design, and (iii)  nonsystematic mistakes in implementing intended choice. In general, there is an imperfect relation between the elicited choice and the intended choice behavior that the experiment tries to measure \citep{gillen2015experimenting}. 

We consider two sources of measurement error in the experimental environment in the context of the budget allocation task due to \citet{choi2014more} and \citet{ahn2014estimating} (which we use in our second application). The first one is the possibility of measurement error in consumption due to trembling-hand errors and the second one is measurement error in prices due to misperception.
\par
We capture the relation between the true intended choices and the measured choices with the following assumption that allows for trembling-hand errors. 
\begin{subassumption}(Trembling-Hand Errors for Experimental Data Sets) \label{assu:tremblinghand}
For all $t\in \mathcal{T}$, it must be the case that 
\[
\Exp{\rand{c}_t}=\Exp{\rand{c}^*_t}.
\]
\end{subassumption}

Assumption~\ref{assu: general measurement error}.\ref{assu:tremblinghand} requires that measurement error in consumption be nonsystematic or equivalently centered around zero. Formally, for all $t\in\mathcal{T}$ it must be that 
\[
\Exp{\rand{w}^c_t}=0.
\]
We use data from \citet{kurtz2019neural} to provide direct empirical evidence supporting Assumption~\ref{assu: general measurement error}.\ref{assu:tremblinghand}. In a setup with two goods \citet{kurtz2019neural} implement a motor task by asking their subjects to click on a visual target located on a budget line. Their experimental interface is exactly the same as in \citet{choi2014more}. \citet{kurtz2019neural} record the target coordinates and the actual coordinates that subjects choose on a screen using a mouse. The difference between the target coordinates and the actual choice is the trembling-hand error. Using this data set, we verified that the average trembling-hand error is statistically not different from zero.\footnote{Formally, we did a t-test of Assumption~\ref{assu: general measurement error}.\ref{assu:tremblinghand} across $23$ subjects for $27$ trials. The null hypothesis is not rejected at the $5$ percent significance level for all except for one good at one trial, but the null hypothesis is not rejected in any instance at the $1$ percent significance level. We did not perform a joint test because there are missing values.} We believe this is reasonable evidence in favor of Assumption~\ref{assu: general measurement error}.\ref{assu:tremblinghand} for budget allocation experiments. 
\par
Failure to account for the possibility of subjects' misperception of the experimental task may affect the elicitation of true consumer behavior. In particular, the experimental design of \citet{choi2014more} and \citet{ahn2014estimating} relies on a graphical representation of the budget hyperplane to elicit consumption choices. It is therefore a visual task and an economic task at the same time. We now consider the possibility of misperception of prices that can be thought of as measurement error in prices due to experimental design. 

\begin{subassumption}\label{assu:limitedattn} (Misperception of Prices for Experimental Data Sets)
For all $t\in \mathcal{T}$, it must be the case that 
\[
\Exp{\rands{\rho}_t}=\Exp{\rands{\rho}^*_t}.
\]
\end{subassumption}

Assumption~\ref{assu: general measurement error}.\ref{assu:limitedattn} relaxes the implicit requirement in the deterministic RP framework that subjects perceive the budget constraints without any distortion. Note that the graphical experimental device used by \citet{choi2014more} and \citet{ahn2014estimating} may make it difficult for consumers to  correctly understand the true prices. We believe that it is desirable to have an RP test of the validity of the UMT that is robust to the misperception of prices when the distortion is (i) attributable to the experimental design and (ii) nonsystematic. 
\par
We use data from \citet{kurtz2019neural} to provide empirical evidence supporting Assumption~\ref{assu: general measurement error}.\ref{assu:limitedattn}. In contrast to the case of trembling-hand errors, there are no direct measurements on price misperception. However, \citet{kurtz2019neural} collected data on visual misperception of coordinates that we can use to test Assumption~\ref{assu: general measurement error}.\ref{assu:limitedattn} indirectly. In particular, we find evidence for a special case of this assumption that is compatible with the misperceived price being $\rands{\rho}_{l,t}=\rands{\rho}_{l,t}^*/\rands{\epsilon}_{l,t}$, for all goods $l\in \{1,\cdots,L\}$, where $\rands{\epsilon}_{l,t}$ captures misperception. Assumption~\ref{assu: general measurement error}.\ref{assu:limitedattn} holds if the misperception error has mean 1 and is independent of true prices.
\par
In a setup with two goods, \citet{kurtz2019neural} implement a visual task by giving their subjects a particular numerical target $z_t^*=(z_{1,t}^*,z_{2,t}^*)$. Next, the experimenters ask their subjects to locate  this coordinate on a budget line. As a result, the point $z_t=(z_{1,t}^*\epsilon_{1,t},z_{2,t}^*\epsilon_{2,t})$ is observed, where we treat $\epsilon_{t}$ as the misperception error. Since prices can be computed from two observed points on the same budget line, misperceived prices will be $\rho_{l,t}=\rho_{l,t}^*/\epsilon_{l,t}$ for $l\in\{1,2\}$.\footnote{Formally, we assume that the perception error realization is the same when the subject observes a second point given the same budget line.} A sufficient condition for  Assumption~\ref{assu: general measurement error}.\ref{assu:limitedattn} is that the multiplicative error is independent of true prices and that the mean perception error across individuals is $1$. We cannot check the independence condition from this validation data. Nonetheless, we cannot reject the null hypothesis that the mean (multiplicative) perception error is equal to $1$ across the $25$ subjects in this experiment, for each of the $27$ trials (at the $5$ percent significance level).\footnote{We performed a t-test analogous to the case of the trembling-hand error.}  
\par
In the budget allocation tasks implemented by \citet{choi2014more} and \citet{ahn2014estimating}, the subject is forced to choose a point on the budget line. Given that in these experimental environments the total income or wealth is known, the support of measurement error $E|X$ must be such that $\rands{\rho}_t^{*\prime}\rand{c}_t^*=\rands{\rho}_t\tr\rand{c}_t\:\as$.\footnote{In our application, we want to test separately for the validity of the static UMT together with (i) trembling-hand error in consumption or (ii) misperception of prices. Hence, we will assume that prices are perfectly measured in the former case (i.e., $\rand{w}_t^p=0\:\as$), while assuming that consumption is perfectly measured in the latter case (i.e, $\rand{w}_t^c=0\:\as$).} Note that in both the trembling-hand errors and in the misperception case we have $d_M=|\mathcal{T}|\cdot L$. The number of centering conditions grows with the number of commodities and the number of decision tasks. 
\subsection{Measurement Error in Scanner Data Sets}
Even though our main focus is on survey and experimental data sets, our methodology can be used in other types of data sets with their corresponding measurement error constraints. Here we provide a quick overview of the case of scanner data sets because of its relevance for RP practitioners. 
\par
Scanner consumer panel data sets are usually of high quality; thus, measurement error concerns may be less important. However, in some cases, like in the well-known Nielsen Homescan Scanner Data Set (NHS), there is evidence of measurement error in prices due to classical misreporting but also due to imputation done in the data collection stage \citep{einav2010recording}.\footnote{Similarly, in the older Stanford Basket Scanner Data Set \citep{echenique2014testable}, there could be measurement error in prices due to unobserved coupons or discounts.}  
\par 
\citet{einav2010recording} use validation data from a retailer and compare it to a subsample of the self-reported NHS (for $2004$); they conclude that consumption is measured rather precisely with roughly $90$ percent of all records being exactly reported. On the other hand, prices are likely to be recorded with an upward bias. In fact, prices are measured precisely in only $50$ percent of records in the sample of interest. The sample mean of the logarithm of prices\footnote{The mean is taken across all records in the time window of interest.} in the NHS seems to be slightly above the same quantity for the validation data.\footnote{It seems that Nielsen generates the price in cases in which they have access to retailer-level price data. The reason for the overreporting is that this imputation process ignores the discounts that consumers may get \citep{einav2010recording}.} Statistically, the difference between the logarithm of prices in the total sample of interest in the NHS and in the validation data is not different from zero (at the $5$ percent significance level), as reported in \citet{einav2010recording}.\footnote{In the subsample of records from the NHS in which the consumers did not get a sales discount, the measurement error in the logarithm of prices seems classical (i.e., centered around zero and symmetric). In contrast, in the subsample of records in which the consumers got a sales discount, the distribution of measurement error in the logarithm of prices has a fat right tail \citep{einav2010recording}. As a result, the total measurement error in the logarithm of prices is not symmetric.} 
\par
We believe that the findings of \citet{einav2010recording} support the conclusion that consumption measurement error in the NHS can be treated as local perturbations (see Appendix~\ref{appen: local perturb}). Hence, we assume that consumption is measured precisely (i.e., $\rand{w}_t^c=0\:\as$) and impose the following centering condition for the NHS:
\begin{subassumption}
    (Centered Differences in Price Measurement Error)\label{assu:scannermeasurement} For all $t,s\in\mathcal{T}$ and all $l=1,2,\dots,L$, it must be the case that
    \[
    \Exp{\log{\rand{p}}_{l,t}^*}=\Exp{\log{\rand{p}}_{l,t}}.
    \]
\end{subassumption}

The above centering condition for measurement error allows for nonsystematic overreporting or underreporting of the logarithm of prices on average. Assumption~\ref{assu: general measurement error}.\ref{assu:scannermeasurement} implies that the number of measurement error conditions is $d_M=\abs{\mathcal{T}}\cdot L$.
\par 
The source of measurement error in the NHS is very different from the one in survey data. In the NHS, households report quantities and prices, while in the survey data set, households report expenditures. For this reason survey data sets are usually used for aggregate analysis (at the level of the category of goods), in cases where price indexes computed by the national statistical agencies are available. As a result, the main mismeasured object is consumption. In contrast, scanner data sets have rich disaggregated information on prices and quantities. But misreported or imputed prices lead to measurement error. For that reason, we impose our centering condition on the main source of measurement error in each of these cases.

\subsection{Other Forms of Measurement Error: Moment Inequalities and Instruments}
Our methodology also allows for imposing moment inequality restrictions on measurement error. Following \citet{schennach2014entropic}, we can handle conditions of the type
\[
 \Exp{g_M(\rand{x},\rand{e})}\geq 0
\]
by introducing an additional slack positive random vector $\rand{s}=(\rand{s}_j)_{j\in \{1,\cdots,d_M\}}$  such that 
\[
\Exp{g_M(\rand{x},\rand{e})-\rand{s}}= 0.
\]
Moment inequalities may be particularly useful for taking into account bounds on measurement-error averages (e.g., $\Exp{\rand{w}_t^p}\leq 0$ for all $t\in \mathcal{T}$). Imposing support constraints on measurement error (e.g., by rounding: $\abs{\rand{w}_{l,t}}\leq 1/2\:\as$) can be handled automatically by setting the support $E|X$ appropriately. 
\par
Measurement error moments can also capture exclusion/orthogonality restrictions. In other words, the analyst may have information or be willing to assume that a particular observed variable is orthogonal to measurement error in consumption, prices or expenditures. The literature of demand estimation in both the static and dynamic setups, which use moments, usually handles measurement error through exclusion restrictions.\footnote{See \citet{lewbel2009tricks} and \citet{alan2018euler}.} 
\par 
Note that Assumption~~\ref{assu: general measurement error}.\ref{assu:(Measurement-error)Meanbudgetneutrality} can be understood as an orthogonality restriction between prices and consumption measurement error (when prices are observed without error). Consider an instrumental variable $\rand{z}_t$ supported on $\Real^{L}$ that is orthogonal to the consumption measurement error; this can be expressed as:
\[
\Exp{\rand{z}_t\tr\rand{w}^c_t}=0,
\]
for all $t\in \mathcal{T}$. For Assumption~\ref{assu: general measurement error}.\ref{assu:(Measurement-error)Meanbudgetneutrality}, the variable $\rand{z}_t=\rand{p}_t\:\as$. These additional restrictions will increase the power of the test simply because they contain more information about measurement error.

\subsection{Asymptotic Power for s/ED-Rationalizability: Illustrative Example}
In previous sections we discussed different centering conditions that can be imposed on measurement error in different data sets. Since our characterization of $\s/\m$-rationality is necessary and sufficient, we have asymptotic power of one against the alternate hypothesis of inconsistency with $\s/\m$-rationalizability. Nonetheless, we still have to show that the alternative hypothesis space is nonempty. That is, we need to make sure that the restrictions on measurement error we provide are falsifiable. Here we build an illustrative example for $\s/\ed$-rationality and the centering condition captured by Assumption~\ref{assu: general measurement error}.\ref{assu:(Measurement-error)Meanbudgetneutrality} with no price measurement error. Similar examples can be built for other moments used in our applications (see Appendix~\ref{appen: counter and robustness}). We also provide simulation evidence in Appendix~\ref{appen:monte carlo,power}.
\par
Consider the random array $(\rands{\rho}_t,\rand{c}_t)_{t\in \mathcal{T}}$ such that $\mathcal{T}=\{0,1\}$, $\rands{\rho}_{0}=(1,1)\tr$, $\rands{\rho}_1=(2,2)\tr$, $\rand{c}_0=(1,1)\tr$, and $\rand{c}_1=(2,2)\tr \:\as$. This data set requires that prices and consumption be the same for all consumers in the population. Moreover, it is easy to see that in deterministic terms $\ed$-rationality fails. We will now show that $\s/\ed$-rationality also fails with appropriate centering conditions. Assume towards a contradiction that this random array is $\s/\ed$-rationalizable. Then, there have to be random variables $\rand{d}\in(0,1]$,  $\{\rand{w}^{c}_t\}_{t=0,1}$, and $\{\rand{v}_t\}_{t=0,1}$ such that
\begin{align*}
    \rand{v}_1-\rand v_0\geq \dfrac{\rho_1\tr}{\rand{d}}(&c_1-c_0)-\dfrac{\rho_1\tr}{\rand{d}}(\rand{w}_1-\rand{w}_0)\:\as,\\
    \rand{v}_0-\rand v_1\geq \rho_0\tr(&c_0-c_1)-\rho_0\tr(\rand{w}_0-\rand{w}_1)\:\as.
\end{align*}
Combining these two inequalities, we get
\[
\dfrac{\rand{d}}{2}(\rand{v}_1-\rand{v}_0)\geq \rand{v}_1-\rand{v}_0\:\as.
\]
Thus, since $1\geq\rand{d}>0\:\as$, it follows that $\rand{d}(\rand{v}_1-\rand{v}_0)\leq0\:\as$. However, this implies that $\rand{v}_0\geq\rand{v}_1\:\as$.
If Assumption~\ref{assu: general measurement error}.\ref{assu:(Measurement-error)Meanbudgetneutrality} holds, then it must be the case that:
\[
\Exp{\rho_0\rand{w}_0}=\Exp{\rho_1\rand{w}_1}=0.
\]
As a result, applying the centering condition, we obtain a contradiction:
\[
0\geq\Exp{\rand{d}(\rand{v}_1-\rand{v}_0)}\geq 4 +\Exp{\rho_1\tr(\rand{w}_1-\rand{w}_0)}=4+\Exp{\rho_1\tr\rand{w}_1}-2\Exp{\rho_0\tr\rand{w}_0}=4. 
\]
This contradiction means that the random array $(\rands{\rho}_t,\rand{c}_t)_{t\in \mathcal{T}}$ is not $\s/\ed$-rationalizable under the centering conditions we described above. 
\par
We highlight that potential lack of statistical power of some centering conditions is not a defect of our methodology. The reason is that if the quality of the data set at hand is not good enough to credibly discern whether observed behavior is consistent with a given model, then no methodology can address this issue.
\section{Empirical Application (I): Testing the Dynamic UMT with Exponential Discounting in Survey Data}\label{sec:empirical application}
In our first application, we apply our methodology to a consumer panel data set gathered from single-individual and couples' households in Spain to test for $\s/\ed$-rationalizability. This important model is under increasing scrutiny because experimental evidence tends to find that the behavior of experimental subjects is time-inconsistent.\footnote{See, for instance, \citet{andreoni2012estimating}, \citet{montielolea2014axiomatization}, and \citet{echenique2014testable}.} Nonetheless, it is important to explore to what extent this finding has external validity. 
\par
To address this issue, some researchers have turned to survey data in the form of household consumption panels. 
Most of this work has found evidence against
exponential discounting \citep{blownonparametric2017}. However, the existing literature has not yet addressed the issue of measurement error in the consumption reported by households in a way that allows us to perform traditional hypothesis testing. (Some additional problems with the existing evidence are (i) the strong parametric assumption on preferences, and (ii) homogeneity restrictions on the discount factor and preferences.)
\par 
One solution to some of the problems in the literature can be found in the work on deterministic RP by \citet{browning1989anonparametric}. In particular, Browning's work avoids making parametric assumptions about the functional form of instantaneous utility. However, this work does not take into consideration the fact that consumption can be mismeasured. In our Monte Carlo experiment, the deterministic test in \citet{browning1989anonparametric} rejects the correct null hypothesis of exponential discounting behavior in $79.4$ percent of the cases, while our methodology correctly fails to reject the null hypothesis that all households are consistent with exponential discounting at the correct $5$ percent significance level (Appendix~\ref{appen:monte carlo,browning}). 
In addition, when we applied \cite{browning1989anonparametric} deterministic methodology to our single-individual households data set, we also obtained a very low success rate.  However, this low success rate of the deterministic test for exponential discounting may be due to measurement error. In our empirical application, we found support for exponential discounting behavior in single-individual households, while at the same time, support for the negative finding in \citet{blownonparametric2017} in the case of couples' households. This fact indicates that deterministic tests may not be very informative about the behavior in a population due to measurement error. Small violations of the deterministic RP inequalities will lead to big rejection rates. Introducing measurement error into the analysis takes these small violations into account.\footnote{In Appendix~\ref{appen:extensions,collective}, we also establish that our test fails to rejects implications of the collective household consumption problem presented in \citet{adams2014consume}.}

Our empirical application also contributes to the literature on estimating the discount factor distribution in survey data sets and in a classical consumer theory environment. This has been the topic of a large body of work which, however, has reached little or no consensus.\footnote{We refer the reader to the survey by \citet{frederick2002timediscounting} for its extensive references.} This lack of consensus can be attributed in some degree to a failure to identify the parameters of interest. Here, we show that the discount factor distribution cannot be identified solely from prices, interest rates, and consumption observations in a data set that suffers from measurement error. (For details see Section~\ref{sec:recoverabilitycounterfactual}.) However, our methodology allows us to test for exponential discounting behavior even in this setting (i.e., without identifying the discount factor distribution).\footnote{In order to learn more information about the discount factor distribution, one needs additional data. One notable example is \citet{mastrobuoni2016criminal}, which uses a quasi-experiment to pin down criminals' time preferences.}

If one ignores the issues of measurement error, the Euler equation allows one to estimate the discount factor and the marginal utility either parametrically or semi-parametrically.\footnote{Examples of estimators of the Euler equation and similar models include \citet{hall1978stochastic}, \citet{hansen1982generalized}, \citet{dunn1986modeling}, \citet{gallant1989seminonparametric}, \citet{chapman1997approximating}, \citet{campbell1999force}, \citet{ai2003efficient}, \citet{chen2009land}, \citet{darolles2011nonparametric}, \citet{chen2014local}, and \citet{escanciano2016}.}
Since our objective is not to estimate but to test the exponential discounting model, we follow a different path. 

In particular, we work with the data set used in \citet{adams2014consume}: the Spanish Continuous Family Expenditure Survey (\textit{Encuesta Continua de Presupuestos Familiares}). The data set consists of the expenditures for 185 individuals and 2004 couples, as well as prices for 17 commodities (categories of goods) recorded over four time periods. Each household was interviewed for four consecutive quarters between 1985-1997. We construct a panel data set of consumption and prices by pooling household's quarterly data points.
\par
The categories of goods are: all food and nonalcoholic drinks, all clothing, cleaning, nondurable articles, household services, domestic services, public transport, long-distance travel, other transport,  petrol, leisure (four categories), other services (two categories), and food consumed outside the home. The data set also contains information on the nominal interest rate on consumer loans faced by the household in any particular quarter.\footnote{We spare the reader more details and refer them instead to \citet{adams2014consume} for further information on the data set.}
\par
Formally, we test for $\s/\ed$-rationalizability with (i) effective prices equal to the discounted spot prices, $\rands{\rho}_t=\rands{\rho}_t^{\ed}$ (defined in Table~\ref{table:Lambdarhodef}), (ii) random marginal utility of income equal to the discounted value of one unit of wealth, $\rands{\lambda}_t=1\:\as$, and (iii) $\rands{\delta}_t=\rand{d}^{t}$ where $\rand{d}$  is interpreted as the (time-invariant) random discount factor supported on or inside $(0,1]$.\footnote{When the discount factor is $0$, it is easy to see that $\ed$-rationality becomes equivalent to $\rat$-rationality. In practice, we pick the interval $[0.1,1]$ as the largest possible support for discount factor. This interval contains most reasonable values in the literature \citep{frederick2002timediscounting,montielolea2014axiomatization}.} Imposing the additional Assumptions~\ref{assu: general measurement error}.\ref{assu:(Measurement-error)Meanbudgetneutrality}, \ref{assu:nondegenerate w}, and~\ref{ass: iid}, we can apply our testing methodology developed in Section~\ref{subsec:elvis,testing}. 
Recall that Assumption~\ref{assu: general measurement error}.\ref{assu:(Measurement-error)Meanbudgetneutrality} indicates that measurement error does not alter the mean value of total expenditure, $\Exp{\rands{\rho}_{t}\tr\rand{c}_{t}}=\Exp{\rands{\rho}_{t}\tr\rand{c}_{t}^{*}}$.

\subsection{The Results}\label{subsec:emp apl,results}

\subsection*{Single-Individual Households}
We apply the deterministic methodology of \citet{browning1989anonparametric} to single-individual households. Our initial conclusion is that $84.3$ percent of the single-individual households behave inconsistently with exponential discounting (even when allowing for substantially  more heterogeneity than previous works).\footnote{We search for each individual household discount factor $d$ in the grid $\{0.1,0.15,\cdots,1\}$. See \citet{crawford2010habits,adams2014consume}, and \citet{blownonparametric2017} for discount factor ranges close to $[0.9,1]$.} Next, we revisit this conclusion using our methodology, which addresses measurement error, while allowing a heterogeneous discount factor. We find that we cannot reject exponential discounting. Formally, we find at the $5$ percent significance level that the random array $\rand{x}=(\rands{\rho}_t,\rand{c}_t)_{t\in\mathcal{T}}$ is $\s/\ed$-rationalizable with a random discount factor $\rand{d}$ \textit{supported on or inside} $[0.975,1]$  
($\mathrm{TS}_{n}=6.480$, $\text{p-value}=0.166$). We believe that the lower bound value of $0.975$, for the quarterly discount rate, is reasonable (it corresponds to a annualized discount rate of $0.9$). 
\par
Although there is no agreement in the literature about what appropriate values for the discount factor are \citep{frederick2002timediscounting}, a common benchmark in applied work is to set the discount factor according to the real interest rate in the economy.\footnote{\citet{dejong2011structural} suggests setting the discount factor value to $d=1/(1+\overline{r})$, where $\overline{r}$ is the average (across individuals) annual real interest rate.} In our case, the lower bound of the quarterly discount factor is $0.975$. This discount factor bound corresponds to an annual real interest rate of $10.7$ percent. This is roughly consistent with this benchmark for the average real interest rates observed in our sample.\footnote{Using the results from Section~\ref{sunsec:Recoverability} we also tested several candidates for the average quarterly discount factor $\theta_0=\Exp{\rand{d}}$: 
$0.995$ ($\mathrm{TS}_{n}=14.071$, $\text{p-value}=0.015$), 
$0.996$ ($\mathrm{TS}_{n}=5.105$, $\text{p-value}=0.403$), and 
$0.997$ ($\mathrm{TS}_{n}=2.967$, $\text{p-value}=0.705$). 
The smallest $\theta_0$ that is not rejected at the $5$ percent significance level, $0.996$, corresponds to the annual average discount factor of $0.984$ and the average annual real interest rate of $1.6$ percent.} 
\par
However, at the $5$ percent significance level, we cannot reject exponential discounting, when discounting is set at $1\:\as$ ($\mathrm{TS}_{n}=6.140$, $\text{p-value}=0.189$). Of course, this does not mean that this is the value of the discount factor. In fact, our sample and our knowledge about measurement error are not enough to elicit the support of the distribution of discount factors. Nonetheless, we can differentiate between behavior consistent with exponential discounting and systematic departures from it (as seen in the power analysis of our method). In sum, our findings provide evidence supporting  exponential discounting for singles under reasonable discount factors. 
\par 
Using our methodology in data sets with more time periods or with additional information about measurement error will improve the information we can obtain about discount factors. However, we will see that identifying the support of the discount factor is not essential to provide informative bounds on average demand (see Appendix~\ref{appen:averagevarianempirical}).\footnote{As robustness check we conducted the test for the discount factor supported on or inside $[0.1,1]$, $[0.5,1]$, and $[0.9,1]$. As expected, since all three intervals contain $[0.975,1]$, the null hypothesis is not rejected. The test statistic $\mathrm{TS}_{n}$ is equal to $6.476$ with $\text{p-value}=0.166$ for all three intervals.}

\subsection*{Couples' Households}
For the couples' households, the deterministic test of \citet{browning1989anonparametric} rejects the exponential discounting model for $89.8$ percent of the observations. Although this number seems large, one should keep in mind that for single-individual households the same deterministic test rejects the model in $84.3$ percent of the cases. At the same time, for $\rand{d}\in[0.1,1]$, our method does not reject the exponential discounting model for single-individual households. But we do reject the model for couples' households. In the case of couples' households, the test statistic is $\mathrm{TS}_{n}=71.015$ ($\text{p-value}<10^{-12}$).\footnote{As a robustness check, we also tested the model assuming that $\rand{d}$ is supported on or inside $[0.5,1]$, $[0.9,1]$, and $\rand{d}=1\:\as$. As expected, since these intervals are contained in $[0.1,1]$, the null hypothesis is rejected. The test statistic $\mathrm{TS}_{n}$ ($\text{p-value}$) is equal to $71.015$ ($<10^{-12}$), $70.964$ ($<10^{-12}$), and $101.579$ ($<10^{-12}$), respectively.}

\subsection{Discussion and Related Work on Testing the Exponential Discounting Model}
One possible concern about our methodology is that its power is low in survey data sets (due to a small $T$ dimension of the data) given its nonparametric nature. However, in Appendix~\ref{appen:monte carlo,power} we report, for $1000$ trials with a sample size of $2000$, a rejection rate  greater than or equal to $72$ percent (with a data generating process consistent with the collective model as in \citealp{adams2014consume}).

Our results for couples' households provide the first nonparametric evidence which is robust to measurement error and which demonstrates that not all couples' households manifests behavior consistent with exponential discounting. In Appendix~\ref{appen:extensions,collective}, we establish that a suitable extension of our methodology fails to reject the collective household consumption  problem presented in  \citet{adams2014consume}.\footnote{\citet{mazzocco07} also provides evidence in favor of the collective model using a parametric methodology.} This should convince practitioners about the importance of  modeling intrahousehold decision-making when dealing with intertemporal choice. The rejection of exponential discounting behavior for couples' households can be better understood given theoretical results that show that aggregating time-consistent preferences may lead to time-inconsistent behavior \citep{jackson2015collective}. 

The deterministic methodology of \citet{browning1989anonparametric} concludes that the fraction of households that is inconsistent with exponential discounting under the deterministic test is similar for both single individuals and couples, but our statistical test rejects in the latter case while reaching the opposite conclusion in the former case. The difference in conclusions is due to the fact that our test implicitly takes into account the severity of the violations of exponential discounting, and imposes the mean budget-neutrality assumption on the measurement error corrections.\footnote{We have also tested $\s/\ed$-rationality under Assumption~\ref{assu: general measurement error}.\ref{assu:tremblinghand}, which requires that consumption quantities measurement error is centered around zero as in \citet{varian1985non}. We strongly reject the null hypothesis in this case, providing an illustration of the importance of using empirically-backed centering conditions.} 
\par
Our main empirical finding  is robust to price measurement error. Adding an additional source of measurement error would make the rationalizability notion less demanding.\footnote{See Appendix~\ref{appen: local perturb} for results on robustness of our methodology to local perturbations in prices.} 
\section{Empirical Application (II): Static UMT in Experimental Data Sets with Trembling-Hand Errors and Misperception of Prices}\label{sec:ApplicationStaticUMT}
In this section, we use our methodology to test the static UMT in the widely known experimental data set by \citet{ahn2014estimating}. The experimental task consists of $T=50$ independent decision trials with $n=154$ subjects. Each decision is a portfolio problem. The subjects face three states of the world $\sigma\in\{1,2,3\}$. The subjects are given $100$ tokens per task and they have to choose a bundle of Arrow securities, $c_t\in \Real^{3}_{+}$, for a randomly drawn price vector $p_{t}\in \Real^{3}_{+}\setminus{\{0\}}$. The subjects are forced to choose a bundle that satisfies Walras' law such that for every decision task it must be that $p_{t}\tr c_{t}=100$. The subjects receive a payment in tokens according to the probability of each state of the world at the end of each round. At the end of the experimental task one of the rounds was selected using a uniform distribution and the tokens payment corresponding to that round is paid in dollars.\footnote{The subjects are informed that the probability of state $\sigma=2$ is $1/3$, and the joint probability of the states $\sigma\in \{1,3\}$ is $2/3$.} The exchange rate is $0.05$ dollars per token and the participation fee is $5$ dollars.
\par
This ingenious experimental device (due to \citealp{choi2014more}) has allowed the RP practitioners to collect a large number of observations per individual with high price variation. \citet{beatty2011howdemanding} highlighted the importance of rich price variation to have enough power in the experimental design to detect violations of the UMT. 
\par
The deterministic RP test for the static UMT in this data set concludes that only $12.99$ percent of the experimental subjects pass the test. At first sight, this is a striking result, because the majority of subjects seems to be inconsistent with the core consumer model in economics. We reexamine the robustness of this result to measurement error in consumption due to errors in the elicitation of the intended behavior of consumers. In our application, we found support for the static UMT in this experiment in stark contrast with the findings from the deterministic RP test.  
\par
Measurement error in the experimental environment may arise due to the nature of the design. Subjects are presented with a graphical representation of a $3$ dimensional budget hyperplane, and they must choose the consumption bundle by pointing to a point in this hyperplane using a computer mouse or the arrow keys in a keyboard.  We must note that there is a mechanical measurement error due to the resolution of the budget hyperplane which is $0.2$ tokens. More important factors, such as a lack of expertise in the decision task, could lead the consumers to make implementation mistakes when trying to choose their preferred alternative. \citet{kurtz2019neural} provide direct evidence of trembling-hand error in a budget allocation task similar to \citet{ahn2014estimating}, as well as indirect evidence of visual misperception, as we have previously discussed. Nonetheless, the actual reason why the experimental design fails to elicit the intended decision task is not our main concern. We take the stand that a desirable test of the static UMT has to be robust to possible nonsystematic mistakes arising from any experimental design. We consider both trembling-hand errors in consumption and nonsystematic misperception of prices.  
Formally, we test for $\s/\rat$-rationalizability with (i) effective prices equal to the prices at each trial $\rands{\rho}_t=\rands{\rho}_t^{\rat}$ (defined in Table~\ref{table:Lambdarhodef}), (ii) marginal utility of wealth $(\rands{\lambda}_t)_{t\in \mathcal{T}}$ supported on $\Lambda=\Real^{\abs{\mathcal{T}}}_{++}$, and (iii) random discount factor equal to $1$ ($\rands{\delta}_t=1\:\as$ for all $t\in\mathcal{T}$).
\par 
We note that there is evidence for trembling-hand errors and misperception errors in the budget allocation task we consider here (see Section~\ref{sec:error in data}). However, we fail to reject the null hypothesis of the static UMT when allowing only for price misperception. Evidently, it follows that we fail to reject the null hypothesis of the static UMT when allowing for both price misperception and trembling hand errors. In addition, we reject the null hypothesis of the static UMT when allowing only for trembling hand errors. Thus, we can conclude that the main source of error in this type of experimental tasks is most likely price misperception.

\subsection*{Misperception of Prices}
We consider the possibility of measurement error in prices arising from misperception. We investigate the case where consumers behave as if they are trying to maximize a utility function subject to a misperceived vector of prices. In this regard, we take the point-of-view of \citet{gillen2015experimenting} that points out that misperception in a low stakes experimental environments may affect the external validity of the conclusions drawn from an experimental data set. \citet{kurtz2019neural} provide indirect evidence of errors induced by the visual task in the budget allocation problem in the very closely related interface of \citet{choi2014more}. We require that consumers' average perception of prices is unbiased, namely, for all $t\in \mathcal{T}$ it must be that 
\[
\Exp{\rands{\rho}_t}=\Exp{\rands{\rho}_t^*}. 
\] 
This is captured in Assumption~\ref{assu: general measurement error}.\ref{assu:limitedattn}.
In order to isolate the effect of misperception on the observed violations of the static UMT, we assume that consumer behavior is measured perfectly ($\rand{w}_t^c=0\:\as$). In addition, due to the experimental design in \citet{ahn2014estimating}, it must be that true prices are such that $\rand{p}^{*\prime}_{t} \rand{c}^*_t=100\:\as$ (i.e., Walras' law holds). The value of the test statistic is $\mathrm{TS}_{n}=17.879$ ($\text{p-value}> 1-10^{-10}$). This is below the conservative critical value $\chi_{150,0.95}^{2}=179.581$.\footnote{We used $900$ draws in the Monte Carlo computation of the maximum-entropy moment of this problem. We chose this number on the basis of the trembling-hand error exercise in this experimental data set in the next section.}
We do not reject the null hypothesis of the static UMT in the presence of misperception of prices, when the average vector of prices is equivalent to the true vector of prices.\footnote{Assumption~\ref{assu: general measurement error}.\ref{assu:limitedattn} has empirical bite. See Appendix~\ref{appen: counterexamples GARP}.}  More importantly, this finding puts in perspective the rejection of the static UMT in experimental data sets that use the graphical representation of budget hyperplanes \citep{choi2014more,ahn2014estimating}. In particular, we find evidence that prices misperception matters. When we account for this possibility, we no longer reject the null hypothesis of the static UMT.

\subsection*{Trembling-Hand Errors in Consumption}
We say that measurement error in consumption is the result of a trembling-hand when it is centered at zero. This idea is captured in Assumption~\ref{assu: general measurement error}.\ref{assu:tremblinghand} that requires that for all $t\in \mathcal{T}$ it must be that $\Exp{\rand{w}^c_t}=0$. Also, we keep the restriction that the true prices and consumption satisfy Walras' law $\rand{p}^{*\prime}_{t} \rand{c}^*_t=100\:\as$, and that prices are measured perfectly $\rand{w}_t^p=0\:\as$. As we discussed before, we have obtained direct evidence for this centering condition using the replication data of \citet{kurtz2019neural}.    
We strongly reject the null hypothesis of the static UMT when allowing for measurement error in the elicitation of the true consumer behavior due to trembling-hand errors. The test statistic is $\mathrm{TS}_{n}=299.137$ ($\text{p-value}<10^{-11}$). This is above the conservative critical value $\chi_{150,0.95}^{2}=179.581$.\footnote{We used $2970$ draws in the Monte Carlo computation of the maximum-entropy moment of this problem. We also tried $580$ and $900$ with test statistics with values of $301.654$ and $306.882$. Which is evidence that moderate size of draws for the Monte Carlo integration step do relatively well in this setup.}

\section{Relation to the Literature}\label{sec:literature}
\citet{afriat1967construction} shows that it is sufficient to know shape constraints on the utility function (i.e., concavity) to bypass the need to know the utility function when testing for the UMT. We generalize this insight by allowing measurement error in consumption and prices. Among authors using the deterministic RP approach, the immediate antecedents to our work using the first-order approach are \citet{browning1989anonparametric}, \citet{blownonparametric2017}, and  \citet{brownquasilinear}. Important advances have been made on testing and doing counterfactual analysis under rationalizability or random utility.\footnote{Relevant examples are \citet{blundell2014bounding}, \citet{dette2016testing}, \citet{lewbel2017unobserved}, and \citet{kitamura2018nonparametric}.} However, the majority of these results assume that observed quantities are measured accurately. \citet{varian1985non} is possibly the first work to introduce the subject of measurement error into the RP approach. Varian's methodology is the closest to that of our own work; he considers precisely measured (albeit random) prices to study measurement error in consumption. Varian's work is compatible with standard statistical hypothesis testing under the strong assumptions of normality (with known variance) and additivity of consumption measurement error. In contrast, our methodology is fully nonparametric. We are able to improve upon Varian's methodology and relax its core assumptions by using a moments approach to measurement error in the RP framework. 
\par
Other papers have dealt with measurement error under different parametric assumptions about measurement error or about the heterogeneity of preferences.\footnote{\citet{gross1995testing} assumes that random consumption is generated by consumers with similar preferences. \citet{tsur1989testing} imposes a log-normal multiplicative measurement-error structure in expenditures. \citet{hjertstrand2013simple} proposes a generalization of  \citet{varian1985non}, but requires knowing the distribution of measurement error. \citet{echenique2011money} assume that measurement error in prices is a normal random variable independent across households and prices with a fixed mean and known variance.}  
\citet{deb2017revealed} consider a nonparametric model of ``price preference.'' They propose an RP test of their model that is robust to small measurement error in prices. \citet{boccardi2016predictive} considers a case of demand with error and establishes a way to account for the trade-off between the fit of the model and its predictive ability (which is a generalization of \citealp{beatty2011howdemanding}).\footnote{The fit of an RP model indicates whether the a data set is consistent with the model. The predictive ability of an RP model is a measure of how easy it is for a data set generated at random to be consistent with the model. See \citet{beatty2011howdemanding} for further details.} 

In practice, the RP theorists (e.g., \citealp{adams2014consume} and \citealp{cherchye2017household}) have dealt with measurement error by perturbing (minimally) the observed individual consumption in order to satisfy the conditions of an RP test. For instance, \citet{adams2014consume} find the additive perturbation with a minimal norm that renders the individual consumption streams compatible with the RP restrictions. Then a subjective threshold is imposed on the maximum admissible norm of the measurement error vector. If the computed norm is above the threshold, then the model is rejected. However, their methodology has one important drawback: every data set can be made to satisfy their test or, equivalently, the test has no power. In addition, the test in \citet{adams2014consume} has no size control. 

Among researchers using the RP approach, \citet{blundell2003nonparametric,blundell2014bounding} are the first to provide consumer demand bounds, under the assumption of static utility maximization in a semiparametric environment (with additive heterogeneity) in which income changes continuously. Our work differs from theirs mainly in that we allow for unrestricted heterogeneity in preferences, do not require that income be observable, and do not impose semiparametric assumptions on wealth effects to provide bounds for demand, given new prices. In addition, nonclassical measurement error is not compatible with their approach.

\section{Conclusion}\label{sec:conclusion}
We propose a new stochastic and nonparametric RP approach (suitable for an environment with measurement error in consumption or prices) that is useful to test for several consumer models that can be characterized by their first-order conditions. In particular, our work can be used (but is not limited) to test for static utility maximization \citep{afriat1967construction}, and for dynamic rationalizability with exponential discounting \citep{browning1989anonparametric}. 

\bibliographystyle{apalike2}
\phantomsection\addcontentsline{toc}{section}{\refname}\bibliography{references}

\begin{thebibliography}{}

\bibitem[Abildgren et~al., 2018]{abildgren2018consistency}
Abildgren, K., Kuchler, A., Rasmussen, A. S.~L., \& Sorensen, H.~S. (2018).
\newblock Consistency between household-level consumption data from registers
  and surveys.
\newblock In {\em 35th IARIW General Conference, Copenhagen, Danmarks
  Nationalbank Working Paper, no. 131}.

\bibitem[Adams et~al., 2014]{adams2014consume}
Adams, A., Cherchye, L., De~Rock, B., \& Verriest, E. (2014).
\newblock Consume {Now} or {Later}? {Time} {Inconsistency}, {Collective}
  {Choice}, and {Revealed} {Preference}.
\newblock {\em American Economic Review}, 104(12), 4147--4183.

\bibitem[Afriat, 1967]{afriat1967construction}
Afriat, S.~N. (1967).
\newblock The construction of utility functions from expenditure data.
\newblock {\em International economic review}, 8(1), 67--77.

\bibitem[Ahn et~al., 2014]{ahn2014estimating}
Ahn, D., Choi, S., Gale, D., \& Kariv, S. (2014).
\newblock Estimating ambiguity aversion in a portfolio choice experiment.
\newblock {\em Quantitative Economics}, 5(2), 195--223.

\bibitem[Ai \& Chen, 2003]{ai2003efficient}
Ai, C. \& Chen, X. (2003).
\newblock Efficient estimation of models with conditional moment restrictions
  containing unknown functions.
\newblock {\em Econometrica}, 71(6), 1795--1843.

\bibitem[Alan et~al., 2018]{alan2018euler}
Alan, S., Atalay, K., \& Crossley, T.~F. (2018).
\newblock Euler equation estimation on micro data.
\newblock {\em Macroeconomic Dynamics}, (pp.\ 1--26).

\bibitem[Alan et~al., 2009]{alanbrowning2009estimating}
Alan, S., Attanasio, O., \& Browning, M. (2009).
\newblock Estimating euler equations with noisy data: two exact gmm estimators.
\newblock {\em Journal of Applied Econometrics}, 24(2), 309--324.

\bibitem[Alan \& Browning, 2010]{alan2010estimating}
Alan, S. \& Browning, M. (2010).
\newblock Estimating intertemporal allocation parameters using synthetic
  residual estimation.
\newblock {\em The Review of Economic Studies}, 77(4), 1231--1261.

\bibitem[Allen \& Rehbeck, 2019]{allen2019identification}
Allen, R. \& Rehbeck, J. (2019).
\newblock Identification with additively separable heterogeneity.
\newblock {\em Econometrica}, 87(3), 1021--1054.

\bibitem[Andreoni \& Sprenger, 2012]{andreoni2012estimating}
Andreoni, J. \& Sprenger, C. (2012).
\newblock Estimating time preferences from convex budgets.
\newblock {\em The American Economic Review}, 102(7), 3333--3356.

\bibitem[Angel et~al., 2018]{angel2018differences}
Angel, S., Heuberger, R., \& Lamei, N. (2018).
\newblock Differences between household income from surveys and registers and
  how these affect the poverty headcount: evidence from the austrian silc.
\newblock {\em Social indicators research}, 138(2), 575--603.

\bibitem[Attanasio et~al., 2014]{attanasio2014nonclassical}
Attanasio, O., Hurst, E., \& Pistaferri, L. (2014).
\newblock The evolution of income, consumption, and leisure inequality in the
  united states, 1980--2010.
\newblock In {\em Improving the measurement of consumer expenditures}  (pp.\
  100--140). University of Chicago Press.

\bibitem[Baker et~al., 2018]{baker2018measurement}
Baker, S.~R., Kueng, L., Meyer, S., \& Pagel, M. (2018).
\newblock {\em Measurement error in imputed consumption}.
\newblock Technical report, National Bureau of Economic Research.

\bibitem[Beatty \& Crawford, 2011]{beatty2011howdemanding}
Beatty, T.~K. \& Crawford, I.~A. (2011).
\newblock How demanding is the revealed preference approach to demand?
\newblock {\em The American Economic Review}, 101(6), 2782--2795.

\bibitem[Blow et~al., 2013]{blow2013nevermind}
Blow, L., Browning, M., \& Crawford, I. (2013).
\newblock Never {Mind} the {Hyperbolics}: {Nonparametric} {Analysis} of
  {Time}-{Inconsistent} {Preferences}.
\newblock {\em Unpublished manuscript}.

\bibitem[Blow et~al., 2017]{blownonparametric2017}
Blow, L., Browning, M., \& Crawford, I. (2017).
\newblock Nonparametric analysis of time-inconsistent preferences.
\newblock {\em Mimeo.}

\bibitem[Blundell et~al., 2014]{blundell2014bounding}
Blundell, R., Kristensen, D., \& Matzkin, R. (2014).
\newblock Bounding quantile demand functions using revealed preference
  inequalities.
\newblock {\em Journal of Econometrics}, 179(2), 112--127.

\bibitem[Blundell et~al., 2003]{blundell2003nonparametric}
Blundell, R.~W., Browning, M., \& Crawford, I.~A. (2003).
\newblock Nonparametric engel curves and revealed preference.
\newblock {\em Econometrica}, 71(1), 205--240.

\bibitem[Boccardi, 2016]{boccardi2016predictive}
Boccardi, M.~J. (2016).
\newblock Predictive ability and the fit-power trade-off in theories of
  consumer behavior.
\newblock {\em Mimeo}.

\bibitem[Boelaert, 2014]{Rboelaert2014}
Boelaert, J. (2014).
\newblock {\em revealedPrefs: Revealed Preferences and Microeconomic
  Rationality}.
\newblock R package version 0.2.

\bibitem[Boskin et~al., 1998]{boskin1998consumer}
Boskin, M.~J., Dulberger, E.~L., Gordon, R.~J., Griliches, Z., \& Jorgenson,
  D.~W. (1998).
\newblock Consumer prices, the consumer price index, and the cost of living.
\newblock {\em Journal of economic perspectives}, 12(1), 3--26.

\bibitem[Bound et~al., 2001]{bound_chapter_2001}
Bound, J., Brown, C., \& Mathiowetz, N. (2001).
\newblock Chapter 59 - {Measurement} {Error} in {Survey} {Data}.
\newblock In J.~J. Heckman \& E. Leamer (Eds.), {\em Handbook of
  {Econometrics}}, volume~5  (pp.\ 3705--3843). Elsevier.

\bibitem[Brown \& Calsamiglia, 2007]{brownquasilinear}
Brown, D.~J. \& Calsamiglia, C. (2007).
\newblock The nonparametric approach to applied welfare analysis.
\newblock {\em Economic Theory}, 31(1), 183--188.

\bibitem[Browning, 1989]{browning1989anonparametric}
Browning, M. (1989).
\newblock A nonparametric test of the life-cycle rational expections
  hypothesis.
\newblock {\em International Economic Review}, (pp.\ 979--992).

\bibitem[Browning et~al., 2010]{browning2010uncertainty}
Browning, M., Chiappori, P.-A., \& Weiss, Y. (2010).
\newblock Uncertainty and dynamics in the collective model.

\bibitem[Campbell \& Cochrane, 1999]{campbell1999force}
Campbell, J.~Y. \& Cochrane, J.~H. (1999).
\newblock By force of habit: A consumption-based explanation of aggregate stock
  market behavior.
\newblock {\em Journal of political Economy}, 107(2), 205--251.

\bibitem[Carroll, 2001]{carroll2001death}
Carroll, C.~D. (2001).
\newblock Death to the log-linearized consumption euler equation!(and very poor
  health to the second-order approximation).
\newblock {\em Advances in Macroeconomics}, 1(1).

\bibitem[Carroll et~al., 2014]{carroll2014introduction}
Carroll, C.~D., Crossley, T.~F., \& Sabelhaus, J. (2014).
\newblock Introduction to" improving the measurement of consumer expenditures".
\newblock In {\em Improving the Measurement of Consumer Expenditures}  (pp.\
  1--20). University of Chicago Press.

\bibitem[Castillo et~al., 1999]{castillo1999medicion}
Castillo, J.~R., Ley, E., \& Izquierdo, M. (1999).
\newblock {\em La medici{\'o}n de la inflaci{\'o}n en Espa{\~n}a}.
\newblock Number~17. " la Caixa: Savings and Pensions Bank of Barcelona".

\bibitem[Chapman, 1997]{chapman1997approximating}
Chapman, D.~A. (1997).
\newblock Approximating the asset pricing kernel.
\newblock {\em The Journal of Finance}, 52(4), 1383--1410.

\bibitem[Chen et~al., 2014]{chen2014local}
Chen, X., Chernozhukov, V., Lee, S., \& Newey, W.~K. (2014).
\newblock Local identification of nonparametric and semiparametric models.
\newblock {\em Econometrica}, 82(2), 785--809.

\bibitem[Chen \& Ludvigson, 2009]{chen2009land}
Chen, X. \& Ludvigson, S.~C. (2009).
\newblock Land of addicts? an empirical investigation of habit-based asset
  pricing models.
\newblock {\em Journal of Applied Econometrics}, 24(7), 1057--1093.

\bibitem[Cherchye et~al., 2017]{cherchye2017household}
Cherchye, L., Demuynck, T., De~Rock, B., Vermeulen, F., et~al. (2017).
\newblock Household consumption when the marriage is stable.
\newblock {\em American Economic Review}, 107(6), 1507--1534.

\bibitem[Choi et~al., 2014]{choi2014more}
Choi, S., Kariv, S., M{\"u}ller, W., \& Silverman, D. (2014).
\newblock Who is (more) rational?
\newblock {\em The American Economic Review}, 104(6), 1518--1550.

\bibitem[Cochrane, 1991]{cochrane1991simple}
Cochrane, J.~H. (1991).
\newblock A simple test of consumption insurance.
\newblock {\em Journal of political economy}, 99(5), 957--976.

\bibitem[Crawford, 2010]{crawford2010habits}
Crawford, I. (2010).
\newblock Habits revealed.
\newblock {\em The Review of Economic Studies}, 77(4), 1382--1402.

\bibitem[Darolles et~al., 2011]{darolles2011nonparametric}
Darolles, S., Fan, Y., Florens, J.-P., \& Renault, E. (2011).
\newblock Nonparametric instrumental regression.
\newblock {\em Econometrica}, 79(5), 1541--1565.

\bibitem[Deb et~al., 2017]{deb2017revealed}
Deb, R., Kitamura, Y., Quah, J. K.-H., \& Stoye, J. (2017).
\newblock Revealed price preference: Theory and stochastic testing.
\newblock {\em Working paper}.

\bibitem[DeJong \& Dave, 2011]{dejong2011structural}
DeJong, D.~N. \& Dave, C. (2011).
\newblock {\em Structural macroeconometrics}.
\newblock Princeton University Press.

\bibitem[DellaVigna \& Malmendier, 2006]{dellavigna2006payingnot}
DellaVigna, S. \& Malmendier, U. (2006).
\newblock Paying not to go to the gym.
\newblock {\em The American Economic Review}, (pp.\ 694--719).

\bibitem[Dette et~al., 2016]{dette2016testing}
Dette, H., Hoderlein, S., \& Neumeyer, N. (2016).
\newblock Testing multivariate economic restrictions using quantiles: the
  example of slutsky negative semidefiniteness.
\newblock {\em Journal of Econometrics}, 191(1), 129--144.

\bibitem[Diewert, 1998]{diewert1998index}
Diewert, W.~E. (1998).
\newblock Index number issues in the consumer price index.
\newblock {\em Journal of Economic Perspectives}, 12(1), 47--58.

\bibitem[Diewert, 2012]{diewert2012afriat}
Diewert, W.~E. (2012).
\newblock Afriat's theorem and some extensions to choice under uncertainty.
\newblock {\em The Economic Journal}, 122(560), 305--331.

\bibitem[Dunn \& Singleton, 1986]{dunn1986modeling}
Dunn, K.~B. \& Singleton, K.~J. (1986).
\newblock Modeling the term structure of interest rates under non-separable
  utility and durability of goods.
\newblock {\em Journal of Financial Economics}, 17(1), 27--55.

\bibitem[Echenique et~al., 2014]{echenique2014testable}
Echenique, F., Imai, T., \& Saito, K. (2014).
\newblock Testable implications of quasi-hyperbolic and exponential time
  discounting.

\bibitem[Echenique et~al., 2011]{echenique2011money}
Echenique, F., Lee, S., \& Shum, M. (2011).
\newblock The money pump as a measure of revealed preference violations.
\newblock {\em Journal of Political Economy}, 119(6), 1201--1223.

\bibitem[Einav et~al., 2010]{einav2010recording}
Einav, L., Leibtag, E., \& Nevo, A. (2010).
\newblock Recording discrepancies in nielsen homescan data: Are they present
  and do they matter?
\newblock {\em QME}, 8(2), 207--239.

\bibitem[Ekeland et~al., 2010]{ekeland2010optimal}
Ekeland, I., Galichon, A., \& Henry, M. (2010).
\newblock Optimal transportation and the falsifiability of incompletely
  specified economic models.
\newblock {\em Economic Theory}, 42(2), 355--374.

\bibitem[Escanciano et~al., 2016]{escanciano2016}
Escanciano, J.~C., Hoderlein, S., Lewbel, A., Linton, O., \& Srisuma, S.
  (2016).
\newblock Nonparametric euler equation identification and estimation.
\newblock {\em Working Paper}.

\bibitem[Forges \& Minelli, 2009]{forges2009afriats}
Forges, F. \& Minelli, E. (2009).
\newblock Afriat's theorem for general budget sets.
\newblock {\em Journal of Economic Theory}, 144(1), 135--145.

\bibitem[Frederick et~al., 2002]{frederick2002timediscounting}
Frederick, S., Loewenstein, G., \& O'Donoghue, T. (2002).
\newblock Time {Discounting} and {Time} {Preference}: {A} {Critical} {Review}.
\newblock {\em Journal of Economic Literature}, 40(2), 351--401.

\bibitem[Gaddis, 2016]{gaddis2016}
Gaddis, I. (2016).
\newblock {\em Prices for Poverty Analysis in Africa}.
\newblock The World Bank.

\bibitem[Galichon \& Henry, 2013]{Galichon2013db}
Galichon, A. \& Henry, M. (2013).
\newblock Dilation bootstrap.
\newblock {\em Journal of Econometrics}, 177(1), 109 -- 115.

\bibitem[Gallant \& Tauchen, 1989]{gallant1989seminonparametric}
Gallant, A.~R. \& Tauchen, G. (1989).
\newblock Seminonparametric estimation of conditionally constrained
  heterogeneous processes: Asset pricing applications.
\newblock {\em Econometrica: Journal of the Econometric Society}, (pp.\
  1091--1120).

\bibitem[Gauthier, 2018]{gauthier2018}
Gauthier, C. (2018).
\newblock Nonparametric identification of discount factors under partial
  efficiency.
\newblock {\em Working paper}.

\bibitem[Gillen et~al., 2017]{gillen2015experimenting}
Gillen, B., Snowberg, E., \& Yariv, L. (2017).
\newblock {\em Experimenting with measurement error: techniques with
  applications to the Caltech cohort study}.
\newblock Technical report, National Bureau of Economic Research.

\bibitem[Grad{\'i}n et~al., 2008]{gradin2008inequality}
Grad{\'i}n, C., Cant{\'o}, O., \& Del~R{\'i}o, C. (2008).
\newblock Inequality, poverty and mobility: Choosing income or consumption as
  welfare indicators.
\newblock {\em Investigaciones Econ{\'o}micas}, 32(2), 169--200.

\bibitem[Gross, 1995]{gross1995testing}
Gross, J. (1995).
\newblock Testing data for consistency with revealed preference.
\newblock {\em The Review of Economics and Statistics}, (pp.\ 701--710).

\bibitem[Guerrero~de Lizardi, 2008]{guerrero2008sesgos}
Guerrero~de Lizardi, C. (2008).
\newblock Sesgos de medici{\'o}n del {\'\i}ndice nacional de precios al
  consumidor, 2002-2007.
\newblock {\em Investigaci{\'o}n econ{\'o}mica}, 67(266), 37--65.

\bibitem[Hall, 1978]{hall1978stochastic}
Hall, R.~E. (1978).
\newblock Stochastic implications of the life cycle-permanent income
  hypothesis: theory and evidence.
\newblock {\em Journal of political economy}, 86(6), 971--987.

\bibitem[Hansen \& Singleton, 1982]{hansen1982generalized}
Hansen, L.~P. \& Singleton, K.~J. (1982).
\newblock Generalized instrumental variables estimation of nonlinear rational
  expectations models.
\newblock {\em Econometrica: Journal of the Econometric Society}, (pp.\
  1269--1286).

\bibitem[Hjertstrand, 2013]{hjertstrand2013simple}
Hjertstrand, P. (2013).
\newblock A simple method to account for measurement errors in revealed
  preference tests.
\newblock {\em IFN Working Paper}.

\bibitem[Jackson \& Yariv, 2015]{jackson2015collective}
Jackson, M.~O. \& Yariv, L. (2015).
\newblock Collective dynamic choice: the necessity of time inconsistency.
\newblock {\em American Economic Journal: Microeconomics}, 7(4), 150--178.

\bibitem[Jerison \& Jerison, 1994]{jerison1994commodity}
Jerison, D. \& Jerison, M. (1994).
\newblock Commodity aggregation and slutsky asymmetry.
\newblock In {\em Models and Measurement of Welfare and Inequality}  (pp.\
  752--764). Springer.

\bibitem[Kitamura \& Stoye, 2018]{kitamura2018nonparametric}
Kitamura, Y. \& Stoye, J. (2018).
\newblock Nonparametric analysis of random utility models.
\newblock {\em Econometrica}, 86(6), 1883--1909.

\bibitem[Kolsrud et~al., 2017]{kolsrud2017studying}
Kolsrud, J., Landais, C., \& Spinnewijn, J. (2017).
\newblock Studying consumption patterns using registry data: lessons from
  swedish administrative data.

\bibitem[Kurtz-David et~al., 2019]{kurtz2019neural}
Kurtz-David, V., Persitz, D., Webb, R., \& Levy, D.~J. (2019).
\newblock The neural computation of inconsistent choice behavior.
\newblock {\em Nature communications}, 10(1), 1583.

\bibitem[Lewbel, 1996]{lewbel1996aggregation}
Lewbel, A. (1996).
\newblock Aggregation without separability: A generalized composite commodity
  theorem.
\newblock {\em The American Economic Review}, 86(3), 524--543.

\bibitem[Lewbel \& Pendakur, 2009]{lewbel2009tricks}
Lewbel, A. \& Pendakur, K. (2009).
\newblock Tricks with hicks: The easi demand system.
\newblock {\em American Economic Review}, 99(3), 827--63.

\bibitem[Lewbel \& Pendakur, 2017]{lewbel2017unobserved}
Lewbel, A. \& Pendakur, K. (2017).
\newblock Unobserved preference heterogeneity in demand using generalized
  random coefficients.
\newblock {\em Journal of Political Economy}, 125(4), 1100--1148.

\bibitem[Ludvigson \& Paxson, 2001]{ludvigson2001approximation}
Ludvigson, S. \& Paxson, C.~H. (2001).
\newblock Approximation bias in linearized euler equations.
\newblock {\em Review of Economics and Statistics}, 83(2), 242--256.

\bibitem[Mastrobuoni \& Rivers, 2016]{mastrobuoni2016criminal}
Mastrobuoni, G. \& Rivers, D. (2016).
\newblock Criminal discount factors and deterrence.
\newblock {\em Available at SSRN 2730969}.

\bibitem[Mathiowetz et~al., 2002]{mathiowetz2002measurement}
Mathiowetz, N., Brown, C., \& Bound, J. (2002).
\newblock Measurement error in surveys of the low-income population.
\newblock {\em Studies of welfare populations: Data collection and research
  issues}, (pp.\ 157--194).

\bibitem[Mazzocco, 2007]{mazzocco07}
Mazzocco, M. (2007).
\newblock {Household Intertemporal Behaviour: A Collective Characterization and
  a Test of Commitment}.
\newblock {\em The Review of Economic Studies}, 74(3), 857--895.

\bibitem[Meyer et~al., 2015]{meyer2015household}
Meyer, B.~D., Mok, W.~K., \& Sullivan, J.~X. (2015).
\newblock Household surveys in crisis.
\newblock {\em Journal of Economic Perspectives}, 29(4), 199--226.

\bibitem[Montgomery et~al., 2019]{montgomery2019structural}
Montgomery, A., Olivola, C.~Y., \& Pretnar, N. (2019).
\newblock A structural model of mental accounting.
\newblock {\em Available at SSRN 3472156}.

\bibitem[Montiel~Olea \& Strzalecki, 2014]{montielolea2014axiomatization}
Montiel~Olea, J.~L. \& Strzalecki, T. (2014).
\newblock Axiomatization and {Measurement} of {Quasi}-hyperbolic {Discounting}.
\newblock {\em Quarterly Journal of Economics}, 129, 1449--1499.

\bibitem[Norets \& Tang, 2014]{norets2014semiparametric}
Norets, A. \& Tang, X. (2014).
\newblock Semiparametric inference in dynamic binary choice models.
\newblock {\em Review of Economic Studies}, 81(3), 1229--1262.

\bibitem[Pistaferri, 2015]{pistaferri2015household}
Pistaferri, L. (2015).
\newblock Household consumption: Research questions, measurement issues, and
  data collection strategies.
\newblock {\em Journal of Economic and Social Measurement}, 40(1-4), 123--149.

\bibitem[Polisson et~al., 2020]{polisson2020revealed}
Polisson, M., Quah, J. K.-H., \& Renou, L. (2020).
\newblock Revealed preferences over risk and uncertainty.
\newblock {\em American Economic Review}, 110(6), 1782--1820.

\bibitem[Powell, 2009]{powell2009bobyqa}
Powell, M.~J. (2009).
\newblock The bobyqa algorithm for bound constrained optimization without
  derivatives.
\newblock {\em Cambridge NA Report NA2009/06, University of Cambridge,
  Cambridge}, (pp.\ 26--46).

\bibitem[Rios \& Sahinidis, 2013]{rios2013derivative}
Rios, L.~M. \& Sahinidis, N.~V. (2013).
\newblock Derivative-free optimization: a review of algorithms and comparison
  of software implementations.
\newblock {\em Journal of Global Optimization}, 56(3), 1247--1293.

\bibitem[Rockafellar, 1970]{rockafellar1970convexanalysis}
Rockafellar, R.~T. (1970).
\newblock {\em Convex analysis}.
\newblock Princeton university press.

\bibitem[Sasaki, 2015]{sasaki2015contraction}
Sasaki, Y. (2015).
\newblock A contraction fixed point method for infinite mixture models and
  direct counterfactual analysis.

\bibitem[Sato, 2020]{sato2020large}
Sato, H. (2020).
\newblock Do large-scale point-of-sale data satisfy the generalized axiom of
  revealed preference in aggregation using representative price indexes?: A
  case involving processed food and beverages.
\newblock {\em Mimeo}.

\bibitem[Schennach, 2014]{schennach2014entropic}
Schennach, S.~M. (2014).
\newblock Entropic latent variable integration via simulation.
\newblock {\em Econometrica}, 82(1), 345--385.

\bibitem[Shumway \& Davis, 2001]{shumway2001does}
Shumway, C.~R. \& Davis, G.~C. (2001).
\newblock Does consistent aggregation really matter?
\newblock {\em Australian Journal of Agricultural and Resource Economics},
  45(2), 161--194.

\bibitem[Thaler, 1985]{thaler1985mental}
Thaler, R. (1985).
\newblock Mental accounting and consumer choice.
\newblock {\em Marketing science}, 4(3), 199--214.

\bibitem[Toda \& Walsh, 2015]{toda2015double}
Toda, A.~A. \& Walsh, K. (2015).
\newblock The double power law in consumption and implications for testing
  euler equations.
\newblock {\em Journal of Political Economy}, 123(5), 1177--1200.

\bibitem[Tsur, 1989]{tsur1989testing}
Tsur, Y. (1989).
\newblock On testing for revealed preference conditions.
\newblock {\em Economics Letters}, 31(4), 359--362.

\bibitem[Varian, 1982]{varian1982nonparametric}
Varian, H.~R. (1982).
\newblock The nonparametric approach to demand analysis.
\newblock {\em Econometrica: Journal of the Econometric Society}, (pp.\
  945--973).

\bibitem[Varian, 1984]{varian1984nonparametric}
Varian, H.~R. (1984).
\newblock The nonparametric approach to production analysis.
\newblock {\em Econometrica: Journal of the Econometric Society}, (pp.\
  579--597).

\bibitem[Varian, 1985]{varian1985non}
Varian, H.~R. (1985).
\newblock Non-parametric analysis of optimizing behavior with measurement
  error.
\newblock {\em Journal of Econometrics}, 30(1-2), 445--458.

\bibitem[Varian, 1990]{varian1990goodness}
Varian, H.~R. (1990).
\newblock Goodness-of-fit in optimizing models.
\newblock {\em Journal of Econometrics}, 46(1-2), 125--140.

\bibitem[Ventura, 1994]{ventura1994note}
Ventura, E. (1994).
\newblock A note on measurement error and euler equations: An alternative to
  log-linear approximations.
\newblock {\em Economics Letters}, 45(3), 305--308.

\end{thebibliography}
\appendix
\hypertarget{appendix}{\section{Appendix}\label{sec:appendix}}

\subsection{Proof of Lemma~\ref{lem:RandomExponentialdiscountingCstar}}\label{appen: proof Lemma2}
First we establish that ({i}) implies ({ii}). If the random array $(\rands{\rho}_{t}^*,\rand c_{t}^{*})_{t\in\mathcal{T}}$ is $\s/\m$-rationalizable, by concavity of $\rand u(\cdot)$, with probability 1 for any $s,t$ and some $\xi\in\nabla \rand{u}(\rand{c}_t^*)$
\begin{align*}
\rand{u}(\rand{c}_s^{*}) -\rand{u}(\rand{c}_t^{*})&\leq \xi\tr(\rand c_s^{*}-\rand c_t^{*}),\\
\xi&\leq\frac{\rands{\lambda}_{t}}{\rands{\delta}_{t}} \rands{\rho}_t^*.
\end{align*}
Let $\rand{N}$ be a random set of indices such that $\frac{\rands{\lambda}_{t}}{\rands{\delta}_{t}}\rands{\rho}_{ti}^*=\xi_i$ for every $i\in \rand{N}$. Hence, $\frac{\rands{\lambda}_{t}}{\rands{\delta}_{t}}\rands{\rho}_{ti}^*\geq\xi_i$ for every $i\not\in \rand{N}$ with probability $1$. As a result, $\rand{c}_{ti}^{*}$ has to be a corner solution for every $i\not\in \rand{N}$. That is, $\rand{c}_{ti}^{*}=0$. Thus, with probability $1$,
\begin{align*}
&\rand{u}(\rand{c}_s^{*}) -\rand{u}(\rand{c}_t^{*})\leq \xi\tr(\rand c_s^{*}-\rand c_t^{*})=\sum_{i\in \rand{N}}{\xi_i(\rand c_{si}^{*}-\rand c_{ti}^{*})}+\sum_{i\not\in \rand{N}}{\xi_i\rand c_{si}^{*}}=\\
&=\sum_{i\in \rand{N}}{\frac{\rands{\lambda}_{t}}{\rands{\delta}_{t}}\rands{\rho}_{ti}^*(\rand c_{si}^{*}-\rand c_{ti}^{*})}+\sum_{i\not\in \rand{N}}{\xi_i\rand c_{si}^{*}}\leq \sum_{i\in \rand{N}}{\frac{\rands{\lambda}_{t}}{\rands{\delta}_{t}}\rands{\rho}_{ti}^*(\rand c_{si}^{*}-\rand c_{ti}^{*})}+\sum_{i\not\in \rand{N}}{\frac{\rands{\lambda}_{t}}{\rands{\delta}_{t}}\rands{\rho}^*_{ti}\rand c_{si}^{*}},
\end{align*}
where the last inequality follows from $\rand c_s$ being nonnegative. As a result, with probability $1$,
\[
\forall s,t\in\mathcal{T}\::\:\rand u(\rand c_{t}^{*})-\rand u(\rand c_{s}^{*})\geq\frac{\rands{\lambda}_{t}}{\rands{\delta}_{t}}\rands{\rho}_{t}^{*\prime}[\rand c_{t}^{*}-\rand c_{s}^{*}].
\]
For any $\epsilon>0$, we let $\rand v_{t}=\rand u(\rand c^{*}_t)-\min_{s\in\mathcal{T}}\rand u(\rand{c}^*_s)+\epsilon\:\as$, for all $t\in\mathcal{T}$.  The well-defined positive random vector $(\rand v_{t})_{t\in\mathcal{T}}$ together with $(\rands{\lambda}_t,\rands{\delta}_t)_{t\in\mathcal{T}}$	satisfies the inequalities in ({ii}).

Now, we want to prove that ({ii}) implies ({i}). The result follows from Theorem 24.8 in \citet{rockafellar1970convexanalysis}. For completeness of the proof we repeat the arguments of Theorem 24.8 in \citet{rockafellar1970convexanalysis}. For a finite $m\in\Natural$, let $\mathbf{t}=\{t_i\}_{i=1}^{m}$, $t_i\in\mathcal{T}$, be a finite set of indices such that for a fixed $\hat{t}\in\mathcal{T}$, $c_{t_1}^*=c_{\hat{t}}^*$. Let $\mathcal{I}$ be the collection of all such indices (i.e., $\mathbf{t}\in \mathcal{I}$). Define
\[
\rand{u}(c^{*})=\inf_{\mathbf{t}\in \mathcal{I}}\left\{\frac{\rands{\lambda}_{t_1}}{\rands{\delta}_{t_1}}\rands{\rho}_{t_1}^{*\prime}(\rand{c}^{*}_{t_2}-\rand{c}^{*}_{t_1})+\cdots+
\frac{\rands{\lambda}_{t_m}}{\rands{\delta}_{t_m}}\rands{\rho}_{t_m}^{*\prime}(c^{*}-\rand{c}^{*}_{t_m})
\right\}.
\]
With probability $1$, the random function $\rand{u}(\cdot)$ is well-defined, concave, locally nonsatiated, and continuous, since it is a pointwise minimum of a finite set of affine functions for every $m$. Moreover, the infimum in $\mathcal{I}$ is attained and achieved at a set of indices without repetitions (this is a consequence of (ii)). Indeed, under (ii), for any finite $m$, $\{t_{i}\}_{i=1}^m$ and $\rand{c}_s^*$, $s\in\mathcal{T}$, with probability $1$,
\begin{align*}
&\frac{\rands{\lambda}_{t_1}}{\rands{\delta}_{t_1}}\rands{\rho}_{t_1}^{*\prime}(\rand{c}^{*}_{t_2}-\rand{c}^{*}_{t_1})+\cdots+
\frac{\rands{\lambda}_{t_m}}{\rands{\delta}_{t_m}}\rands{\rho}_{t_m}^{*\prime}(\rand{c}_s^{*}-\rand{c}^{*}_{t_m})+\frac{\rands{\lambda}_{s}}{\rands{\delta}_{s}}\rands{\rho}_{s}^{*\prime}(\rand{c}_{t_1}^{*}-\rand{c}^{*}_{s})\geq\\
&\rand{v}_{t_{2}}-\rand{v}_{t_{1}}+\rand{v}_{t_{3}}-\rand{v}_{t_{2}}+\cdots+\rand{v}_{s}-\rand{v}_{t_m}+\rand{v}_{t_1}-\rand{v}_{s}=0.
\end{align*}
Thus,
\[
\rand{u}(\rand{c}^{*}_s)\geq \frac{\rands{\lambda}_{s}}{\rands{\delta}_{s}}\rands{\rho}_{s}^{*\prime}(\rand{c}_{s}^{*}-\rand{c}^{*}_{t_1})>-\infty
\]
with probability $1$. (In particular, $u(c_{\hat{t}}^*)=0$.)
\par
It is left to show that, with probability $1$, $\frac{\rands{\lambda}_{t}}{\rands{\delta}_{t}}\rands{\rho}_{t}^*\in\nabla\rand{u}(\rand{c}^*_t)$ for all $t\in\mathcal{T}$.
Fix some $t\in\mathcal{T}$ and $\delta>0$. By the definition of $\rand{u}(\cdot)$, there exists some $\{t_i\}_{i=1}^{m}$ such that, with probability 1, $\rand{u}(\rand{c}^*_t)+\delta>\frac{\rands{\lambda}_{t_1}}{\rands{\delta}_{t_1}}\rands{\rho}_{t_1}^{*\prime}(\rand{c}^{*}_{t_2}-\rand{c}^{*}_{t_1})+\cdots+
\frac{\rands{\lambda}_{t_m}}{\rands{\delta}_{t_m}}\rands{\rho}_{t_m}^{*\prime}(\rand{c}^{*}_t-\rand{c}^{*}_{t_m})\geq \rand{u}(\rand{c}^*_t)$. Again, by the definition of $\rand{u}(\cdot)$, for any $c^*$
\[
\frac{\rands{\lambda}_{t_1}}{\rands{\delta}_{t_1}}\rands{\rho}_{t_1}^{*\prime}(\rand{c}^{*}_{t_2}-\rand{c}^{*}_{t_1})+\cdots+
\frac{\rands{\lambda}_{t_m}}{\rands{\delta}_{t_m}}\rands{\rho}_{t_m}^{*\prime}(\rand{c}^{*}_t-\rand{c}^{*}_{t_m})+\frac{\rands{\lambda}_{t}}{\rands{\delta}_{t}}\rands{\rho}_{t}^{*\prime}({c}^{*}-\rand{c}^{*}_{t})\geq \rand{u}(c^*).
\]
Hence,
\[
\rand{u}(\rand{c}^*_t)+\delta+\frac{\rands{\lambda}_{t}}{\rands{\delta}_{t}}\rands{\rho}_{t}^{*\prime}({c}^{*}-\rand{c}^{*}_{t})>\rand{u}(c^*).
\]
Since the choice of $\delta$, $t$ and $c^*$ was arbitrary, $\frac{\rands{\lambda}_{t}}{\rands{\delta}_{t}}\rands{\rho}_{t}^*\in\nabla\rand{u}(\rand{c}^*_t)$ for all $t\in\mathcal{T}$.

\subsection{Proof of Proposition~\ref{prop:convexity}}

Take any $\theta_1\in\Theta_0$, $\theta_2\in\Theta_0$, and $\lambda\in[0,1]$ (if $\Theta_0$ is empty, then the conclusion of the proposition follows trivially). Since $\theta_i\in\Theta_0$, $i=1,2$, by Theorems~\ref{thm:DistributionalRP} and~\ref{thm:ELVISexponentialdiscounting} there exist $\{\mu_{i,k}\}_{k=1}^{\infty}$, $i=1,2$, such that 
\[
\lim_{k\to\infty}\norm{\Expt[\mu_{i,k}\times\pi_0]{g(\rand{x},\rand{\tilde{e}};\rho_{T+1},\theta_{i})}}=0
\]
for both $i=1,2$. Consider $\theta_{\lambda}=\lambda\theta_1+(1-\lambda)\theta_2$ and $\mu_{\lambda,k}=\lambda\mu_{1,k}+(1-\lambda)\mu_{2,k}$. Note that
\begin{align*}
  &\Expt[\mu_{\lambda,k}\times\pi_0]{g_C(\rand{x},\rand{c}_{T+1};\rho_{T+1},\theta_{\lambda})}=\Expt[\mu_{\lambda,k}\times\pi_0]{\tilde g(\rand{x},\rand{c}_{T+1};\rho_{T+1})}-\Expt[\pi_0]{A(\rand{x};\rho_{T+1})\tr\theta_{\lambda}}=\\
  &\lambda\Expt[\mu_{1,k}\times\pi_0]{\tilde g(\rand{x},\rand{c}_{T+1};\rho_{T+1})}-\lambda\Expt[\pi_0]{A(\rand{x};\rho_{T+1})\tr\theta_{1}}+\\
  &(1-\lambda)\Expt[\mu_{2,k}\times\pi_0]{\tilde g(\rand{x},\rand{c}_{T+1};\rho_{T+1})}-(1-\lambda)\Expt[\pi_0]{A(\rand{x};\rho_{T+1})\tr\theta_{2}}\\
  &\lambda\Expt[\mu_{1,k}\times\pi_0]{g_C(\rand{x},\rand{c}_{T+1};\rho_{T+1},\theta_{1})}+(1-\lambda)\Expt[\mu_{2,k}\times\pi_0]{g_C(\rand{x},\rand{c}_{T+1};\rho_{T+1},\theta_{2})}.
\end{align*}
Hence, by the triangular inequality,
\begin{align*}
  &\lim_{k\to\infty}\norm{\Expt[\mu_{\lambda,k}\times\pi_0]{g_C(\rand{x},\rand{c}_{T+1};\rho_{T+1},\theta_{\lambda})}}\leq
  \lambda\lim_{k\to\infty}\norm{\Expt[\mu_{1,k}\times\pi_0]{g_C(\rand{x},\rand{c}_{T+1};\rho_{T+1},\theta_{1})}}+\\
  &(1-\lambda)\lim_{k\to\infty}\norm{\Expt[\mu_{2,k}\times\pi_0]{g_C(\rand{x},\rand{c}_{T+1};\rho_{T+1},\theta_{2})}}=\lambda\cdot 0 + (1-\lambda)\cdot 0=0.
\end{align*}
Thus, since $g_{I,O}$ and $g_M$ do not depend on $\theta$, by the triangular inequality,
\begin{align*}
  &\lim_{k\to\infty}\norm{\Expt[\mu_{\lambda,k}\times\pi_0]{g(\rand{x},\rand{\tilde{e}};\rho_{T+1},\theta_{\lambda})}}\leq
  \lambda\lim_{k\to\infty}\norm{\Expt[\mu_{1,k}\times\pi_0]{g(\rand{x},\rand{\tilde{e}};\rho_{T+1},\theta_1)}}+\\
  &(1-\lambda)\lim_{k\to\infty}\norm{\Expt[\mu_{2,k}\times\pi_0]{g(\rand{x},\rand{\tilde{e}};\rho_{T+1},\theta_2)}}=0.
\end{align*}
The later means that $\theta_{\lambda}\in\Theta_0$. The fact that the choice of $\theta_1$, $\theta_2$, and $\lambda$ was arbitrary implies that $\Theta_0$ is convex.

\subsection{Proof of Theorem~\ref{thm:ELVISexponentialdiscounting}}\label{appen: proof theorem 3}

	The result is a direct application of Theorem~\ref{thm:DistributionalRP}, and Theorem 2.1 in \citet{schennach2014entropic}. For completeness of the proof we present Theorem 2.1 in \citet{schennach2014entropic} using our notation below.

\begin{thmnonumber}[Theorem 2.1, \citealp{schennach2014entropic}]
Assume that the marginal distribution of $\rand{x}$ is supported on some set $X\subseteq\Real^{d_x}$, while the distribution of $\rand e$ conditional on $\rand x=x$ is supported on or inside the set $E \subseteq \Real^{d_e}$ for any $x\in X$. Let $h$, $g$ and $\eta$ satisfy Definition~\ref{def:MEM}. Then
\begin{align*}
		\inf_{\mu\in\mathcal{P}_{E|X}}\norm{\Expt[\mu\times\pi_0]{g(\rand x,\rand e)}}=0 \iff \inf_{\gamma\in\mathbb{R}^{k+q}}\norm{\Expt[\pi_0]{h(\rand x;\gamma)}}=0,
\end{align*}
		where $\pi_0\in\boldsymbol{P}_{\mathcal{X}}$ is the observed distribution
		of $\rand x$.
\end{thmnonumber}

\subsection{Proof of Theorem~\ref{thm:semianal}}\label{appen: proof theorem 4}
Recall that the first $k=\abs{\mathcal{T}}^2-\abs{\mathcal{T}}$ moments correspond to the inequality conditions, and the last $q$ moments correspond to the measurement error centering conditions. Let $\gamma_I=(\gamma_j)_{j=1,\dots,k}$, $g_I=(g_j)_{j=1,\dots,k}$, $\gamma_M=(\gamma_j)_{j=k+1,\dots,k+q}$, and $g_M=(g_j)_{j=k+1,\dots,k+q}$ be sub-vectors of $\gamma$ and $g$ that correspond to inequality and the measurement error centering conditions, respectively.

Step 1. Take a sequence $\{\gamma_{I,l}\}_{l=1}^{+\infty}$ such that every component of $\gamma_{I,l}$ diverges to $+\infty$. Note that since $g_I$ takes values in $\{-1,0\}^{k}$,
\[
\sup_{x,e}\abs{\exp(\gamma_{I,l}\tr g_{I}(x,e))-\Char{g_{I}(x,e)=0}}\leq \exp(-\min_{i=1,\dots,k}\gamma_{I,l,i})\to_{l\to+\infty}0,
\]
where $\gamma_{I,l,i}$ is the $i$-th component of $\gamma_{I,l}$.
Hence, for any function $f\in L^1(\eta(\cdot|x))$
\begin{align*}
&\norm{\int f(e)\exp(\gamma_{I,l}\tr g_{I}(x,e))d\eta(e|x)-\int f(e)\Char{g_{I}(x,e)=0}d\eta(e|x)}\leq\\
&\leq \exp(-\min_{i=1,\dots,k}\gamma_{I,l,i})\int{\norm{f(e)}d\eta(e|x)}\to_{l\to+\infty}0.
\end{align*}
Hence, the sequence of measures $\exp(\gamma_{I,l}\tr g_{I}(x,\cdot))d\eta(\cdot|x)$ converges to the measure 
\[
\Char{g_{I}(x,\cdot)=0}d\eta(\cdot|x)
\]
in total variation. The later measure is well-defined and nontrivial since we assume that $\tilde E|X=\{e:\Char{g_{I}(x,e)=0}\}$ has a positive measure under $\eta(\cdot|x)$. Let $d\tilde{\eta}(\cdot|x)$ denote $\Char{g_{I}(x,\cdot)=0}d\eta(\cdot|x)$.
\par
Step 2. Consider the moment conditions under $d\tilde{\eta}(\cdot|x)$
\[
\tilde h_M(x;\gamma)=\dfrac{\int_{e\in{E}|{X}}g_M(x,e)\exp(\gamma\tr g_M(x,e))d\tilde\eta(e|x)}{\int_{e\in{E}|{X}}\exp(\gamma\tr g_M(x,e))d\tilde\eta(e|x)}.
\]
Definition~\ref{assu:eta}.(iii) together with Assumption~\ref{assu:bounded support} and Step 1 imply that for any compact set $\Gamma\in\Real^{q}$, uniformly in $\gamma_M\in \Gamma$
\begin{align*}
\norm{\Expt[\pi_0]{h(\rand x;(\gamma_{I,l}\tr,\gamma_M\tr)\tr)}}=\norm{\Expt[\pi_0]{\tilde h_M(\rand{x};\gamma_M)}}+o(1).
\end{align*}
Thus, by continuity of $h_M$ in $\gamma_M$, when $l$ goes to infinity, we can work with the reduced optimization problem:
\begin{equation}\label{eq1}
\inf_{\gamma\in\Real^{q}}\norm{\Expt[\pi_0]{\tilde h_M(\rand{x};\gamma)}}.
\end{equation}
\par
Step 3. Note that (\ref{eq1}) is equivalent to the optimization problem in Theorem~\ref{thm:ELVISexponentialdiscounting}. Hence, infimum in (\ref{eq1}) is equal to $0$ if and only if the data is approximately consistent with model $\m$.
\par
We assumed that every component of $g_M$ takes both positive and negative values on some nonzero measure subsets of $\tilde E|X$ (Assumption~\ref{assu:nondegenerate w}). Hence, following the proof of Theorem 2.1 and Lemma A.1 in \citet{schennach2014entropic}, we can conclude that if infimum in (\ref{eq1}) is equal to $0$, then it is achieved at some finite and unique $\gamma_{0,M}$. Otherwise, $\norm{\gamma_M}$ diverges to infinity.

\subsection{Proof of Theorem~\ref{thm:test}}\label{appen: proof theorem 5}

The result is a direct application of Theorem F.1 in \citet{schennach2014entropic}. For completeness of the proof we present the version of it that is applicable to our setting below.
\begin{thmnonumber}[Theorem F.1, \citealp{schennach2014entropic}] Let data be i.i.d.. If  (i) the set
	\[
	\Gamma=\{\gamma\in\Real^{q}:\Exp{\norm{\tilde h_M(\rand{x},\gamma)}}\leq C\}
	\]
	is nonempty for some $C<\infty$; (ii) $\Exp{\norm{\tilde h_M(\rand{x},\gamma)}^2}<\infty$ for all $\gamma\in\Gamma$, then
	\[
	\lim_{n\to\infty}\Prob{\mathrm{TS}_{n}>\chi^2_{q,\alpha}}\leq \alpha.
	\]
\end{thmnonumber}
An i.i.d. sample is assumed. To show the validity of conditions (i) and (ii) note that since $\rand{x}$ has a bounded support (by Assumption~\ref{assu:bounded support}) and $\tilde\eta$ satisfies conditions of Definition~\ref{assu:eta}.(iii), for any finite $\gamma$ there exist finite positive constant $C_1(\gamma)$ such that almost surely in $\rand{x}$
\[
\norm{\tilde h_{M}(\rand{x},\gamma)}^2\leq C_1(\gamma).
\]
Hence, for any nonempty compact set $\Gamma$ one can take $C=\sup_{\gamma\in\Gamma}C_1(\gamma)$. Together with Assumption~\ref{assu:bounded support}, the later implies condition (ii). Similarly, one can use $C$ to bound $\Exp{\norm{\tilde h_M(\rand{x},\gamma)}}$.
\par
Under the alternative hypothesis, $\norm{\hat{\tilde{h}}_M(\gamma)}$ either converges to a positive constant or diverges to infinity. Thus, since eigenvalues of $\tilde\Omega(\gamma)$ are bounded away from zero and are bounded from above the test is consistent.

\subsection{Proof of Theorem~\ref{thm:CollectiveSufficient}}\label{appen: proof theorem 8}

By Theorem~\ref{thm:CollectiveDeterministicCharacterization} we have that the following inequalities hold almost surely:

\[
\rand v_{t,A}-\rand v_{s,A}\geq\frac{1}{\rand d_{A}^{t}}[\rands{\rho}_{t,I}\tr(\rand c_{t,I}^{*}-\rand c_{t,B}^{*}-\rand c_{s,I}^{*}+\rand c_{s,B}^{*})+\frac{\rand p_{t,H}-\rand p_{t,B}}{\prod_{j=1}^{t}(1+\rand r_{j})}'(\rand c_{t,H}^{*}-\rand c_{s,H}^{*})\quad\forall t,s\in\mathcal{T},
\]

\[
\rand v_{t,B}-\rand v_{s,B}\geq\frac{1}{\rand d_{B}^{t}}[\rands{\rho}_{t,I}\tr(\rand c_{t,B}^{*}-\rand c_{s,B}^{*})+\frac{\rand p_{t,B}}{\prod_{j=1}^{t}(1+\rand r_{j})}'(\rand c_{t,H}^{*}-\rand c_{s,H}^{*})\quad\forall t,s\in\mathcal{T}.
\]

Then we multiply the first inequality by $\rand d_{A}^{t}$, this random variable is positive almost surely, so it does not alter the inequalities. We do the same for the second inequality, and multiply it by $\rand d_{B}^{t}$. Then we add-up the two inequalities, to obtain:
\[
\rand d_{A}^{t}(\rand v_{t,A}-\rand v_{s,A})+\rand d_{B}^{t}(\rand v_{t,B}-\rand v_{s,B})\geq\rands{\rho}_{t,I}\tr(\rand c_{t}^{*}-\rand c_{s}^{*})\quad\forall t,s\in\mathcal{T}.
\]

\section{Monte Carlo Experiments}\label{appen:monte carlo}
In this section we study the behavior of our test in two Monte Carlo experiments. In the first one, we provide evidence for overrejection of the exponential discounting model by the deterministic test of \citet{browning1989anonparametric} and correct coverage by our test. In the second experiment, we provide evidence for the power (consistency) of our test against some fixed alternatives. Finally, we conduct some robustness checks of our Markov Chain Monte Carlo (MCMC) integration.

\subsection{Overrejection of Exponential Discounting for Browning's Deterministic Test}\label{appen:monte carlo,browning}

The objective of the Monte Carlo simulation exercise is to test the performance of the methodological procedure developed in this paper against the deterministic benchmark. We are going to provide evidence that a data set generated by a random exponential discounter, when contaminated with measurement error, will be erroneously rejected by deterministic methodologies at the individual level for a sizable fraction of the sample \citep{browning1989anonparametric,blownonparametric2017}. However, our test will not reject it.
\par
We choose our simulation configuration setup to match those of the household characteristics in our application. The Monte Carlo exercise will deal with a moderate size data set of $n=2000$ individuals to show that it works in a data set of the roughly the same size as in our application. The time period is $\mathcal{T}=\{0,1,2,3\}$, and we consider $L=17$ goods. We use the same discounted prices $\{\rho_{i,t}\}_{i=1}^{n}$ as the ones given in \citet{adams2014consume}.\footnote{We use the observed price matrix and sample from it uniformly with repetition at each Monte Carlo experiment.} These are the prices faced by the single-individual/couples households in our application. We consider consumers with the constant elasticity of substitution (CES) instantaneous utility 
\[
u(c_{t})=\sum_{l=1}^{L}\frac{c_{t,l}^{1-\sigma}}{1-\sigma},
\]
where $\rands{\sigma}\sim U[1/15,100]$ is heterogeneous across individuals.
\par
Following \citet{browning1989anonparametric}, the true consumption rule for each consumer and each realized $d$ is given by
\begin{equation*}
c_{t,l}^{*}=\left(\frac{1}{d^{t}}\rho_{t,l}\right)^{-1/\sigma},
\end{equation*}
for all $l=1,\cdots,L$ and $t\in\mathcal{T}$. For the discount factor we considered two different data generating processes (DGPs). For the first DGP $\rand d\sim U[0.8,1]$ (DGP1) and for the second one $\rand{d}=1 \:\as$ (DGP2).  
\par
We perturb the generated consumption with $\rands{\epsilon}_{t,l}\sim \mathrm{i.i.d.}\:U[0.97,1.03]$, which implies that $\Exp{\rands{\epsilon}_{t,l}}=1$. That is, observed consumption is equal to true consumption times the multiplicative perturbation $\rand{c}_{t,l}=\rand{c}_{t,l}^{*}\rands{\epsilon}_{t,l}$.
We define measurement error in consumption as $\rand{w}^c_{t,l}=\rand{c}_{t,l}-\rand{c}_{t,l}^{*}$ and fix $\rand{w}^p_t=0\:\as$. Note that the implied random measurement error $\rand w^c_{t,l}$ is mean-zero by construction and satisfies Assumption~\ref{assu: general measurement error}.\ref{assu:(Measurement-error)Meanbudgetneutrality}. The random vector perturbations $\rands{\epsilon}_{t,l}$ captures incorrect consumption reporting or recording, and can be as high as $1.03$ times the true consumption. This means that relative measurement error is around $3$ percent. This procedure produces a data set $(\rands{\rho}_{t,i},\rand{c}_{t,i})_{i=1,t\in\mathcal{T}}^{i=n}$.
\par
We replicate the experiment $m=1000$ times for both DGPs. The results are presented in Table~\ref{table:browning}. For the deterministic test in \citet{browning1989anonparametric} we use a grid search over $d$ on $[0.1,1]$ with a grid step $0.05$. Searching over a smaller set (e.g., $[0.8,1]$) will only weakly increase the rejection rate of the deterministic test. For DGP1 the deterministic test rejects the exponential discounting model in $53$ percent of the cases on average across experiments. For DGP 2 the average rejection rate across experiments is $79.4$ percent. 
\par
We use our methodology to test for $\s/\ed$-rationality for both DGPs assuming that the support of $\rand{d}$ is known. Assuming bigger support for $\rand{d}$ will only weakly decrease the rejection rate of our test. In other words, the design of our experiment favors the deterministic test. Nevertheless, our methodology cannot reject the correct null hypothesis that all households are consistent with $\s/\ed$-rationality at the $5$ percent significance level. The rejection rate for each DGP1 and DGP2 is $1.2$ percent. As expected, both rejection rates are less than $5$ percent. 
\begin{longtable}{lcccc}
\bottomrule
\midrule\midrule
\endfoot
\caption{Rejection Rates: $\ed$-rationalizability. Number of replications $m=1000$. \label{table:browning}}\\
\toprule
&\multicolumn{2}{c}{rejection rate ($\%$)}\\
&Deterministic test &Our methodology \\
\midrule
DGP1 & $53$&$1.2$\\
DGP2& $79.4$&$1.2$
\end{longtable}

\subsection{Power Analysis}\label{appen:monte carlo,power}
We choose our simulation configuration setup to match Section~\ref{appen:monte carlo,browning}. However, the consumer units are assumed to be couples whose behavior is described by the collective model with exponential discounting described in \citet{adams2014consume}. The individuals in the household are indexed by $A$ and $B$. The random discount factors are $\rand{d}_A$ and $\rand{d}_B$. Individuals face different prices for good $l$ at time period $t$ given their bargaining power $\rands{\mu}_{t,l}$. We observe the sum of these two prices $\rands{\rho}_{t,l}$. That is, $\rands{\rho}_{t,l,A}=\rands{\mu}_{t,l}\rands{\rho}_{t,l}$, and $\rands{\rho}_{t,l,B}=(1-\rands{\mu}_{t,l})\rands{\rho}_{t,l}$. Note that the bargaining power is good and time specific. The random price vectors $\rands{\rho}_{t}$ were drawn from the data set in \citet{adams2014consume} as described in Section~\ref{appen:monte carlo,browning}. 
\par
Similar to the experimental design in Section~\ref{appen:monte carlo,browning}, the consumption rule for each individual and realized $d_j$, $j\in\{A,B\}$ is given by
\begin{equation*}
c_{t,l,j}^{*}=\left(\frac{1}{d^{t}_j}\rho_{t,l,j}\right)^{-1/\sigma_{l,j}},\quad l=1,\cdots,L;\: t\in\mathcal{T},
\end{equation*}
where $\sigma_{l,j}$ is a realization of $\rands{\sigma}_{l,j}\sim\mathrm{i.i.d.}\: U[1/15,100]$, $j\in\{A,B\}$. Then the household consumption data is the sum of individual consumption: $\rand{c}^*_t=\rand{c}^*_A+\rand{c}^*_B\:\as$. The generating process for measurement error coincides with the one presented in Section~\ref{appen:monte carlo,browning}. As a result we generate $(\rands{\rho}_t,\rand{c}_t)_{t\in\mathcal{T}}$ and test whether this data is consistent with $\s/\ed$-rationality.
\par
We consider two different DGPs for the distribution of $\rand{d}_j$, $j\in\{A,B\}$, and $\rands{\mu}_{t,l}$.
\par
\noindent \emph{DGP3}. $\rand{d}_A\sim U[0.1,1]$, $\rand{d}_B\sim U[0.99,1]$, and $\rands{\mu}_{t,l}=1/2\:\as$. Under this DGP household members face the same prices but may have different discount factors.
\par
\noindent \emph{DGP 4}. $\rand{d}_A\sim U[0.1,1]$, $\rand{d}_B\sim U[0.99,1]$, and $\rands{\mu}_{t,l}\sim\mathrm{i.i.d.}\: U[1/3,2/3]$. Under this DGP household members face different prices and may have different discount factors.
\par
We conducted the experiments with each DGP $m=1000$ times for two sample sizes, $n=2000$ and $n=3000$. The supports of the discount factors were assumed unknown and contained inside $[0.1,1]$ interval.  The results are presented in Table~\ref{table:powercouplesresults}. 
\begin{longtable}{lcccc}
\bottomrule
\midrule\midrule
\endfoot
\caption{Rejection Rates: Collective Model. Number of replications $m=1000$. \label{table:powercouplesresults}}\\
\toprule
& prices &discount factors &\multicolumn{2}{c}{rejection rate ($\%$)} \\
&&&$n=2000$&$n=3000$\\
\midrule
DGP3 & same& different  & $32$&$69.1$\\
DGP4& different& different  & $72$&$96.9$
\end{longtable}

For DGP3 with equal bargaining power (same prices) and heterogeneous discount factors, the rejection rate is $32$ percent for the sample size of $n=2000$ and increases to $69.1$ percent for $n=3000$. For DGP4 with asymmetric bargaining power (different prices) and heterogeneous discount factors the rejection rate is even bigger and is equal to $72$ and $96.9$ percent for $n=2000$ and $n=3000$, respectively.
\par
We highlight that DGP3 is compatible with hyperbolic discounting. It is easy to see that consumption behavior of the collective model with symmetric bargaining (i.e., same prices) satisfies the Afriat inequalities for hyperbolic discounting in \citet{blow2013nevermind}. 

\subsection{Robustness of MCMC integration.}
Our testing procedure requires some user-specified parameters: the distribution $\eta$ and the length of the MCMC chain. As mentioned in Section~\ref{sec:elvis}, the choice of $\eta$ has no effect on the value of the test statistic both asymptotically and in finite samples. In other words the difference in values of the test statistics computed using two different $\eta$'s can only be driven by numerical precision of the MCMC integration step and the optimization algorithm used. Thus, we focus on the performance of procedure for different MCMC chain length.
\par
The results in Sections~\ref{appen:monte carlo,browning} and~\ref{appen:monte carlo,power} were obtained using the chain length equal to $\mathrm{cl}=10000$. We decrease the chain length to $\mathrm{cl}=5000$ and for the sample size $n=2000$ we additionally experiment with DGP2, DGP3, and DGP4. The remaining elements of the simulations remain the same as before.  
\par
Table~\ref{table:powercouplesresultschainlength} shows that halving the chain length from $1000$ to $5000$ changes very little the rejection rates of the three DGPs of interest. This is of course desirable as lack of robustness would suggest that the MCMC chain has not converged. This provides reassurance that our choice of chain length $10000$ is appropriate.
\begin{longtable}{lcc}
\bottomrule
\midrule\midrule
\endfoot
\caption{Rejection Rates: ED and Collective Models. Sample size $n=2000$. \label{table:powercouplesresultschainlength}}\\
\toprule
&\multicolumn{2}{c}{rejection rate ($\%$)} \\
&$\mathrm{cl}=10000$ & $\mathrm{cl}=5000$ \\
\midrule
DGP2&  $1.2$ & $2.5$\\
DGP3 &  $32$ & $34.9$\\
DGP4&  $72$ & $71.8$
\end{longtable}

\section{Computational Aspects}\label{appen: coputational aspects}
In this appendix we discuss the computational aspects of our procedure. In Appendix~\ref{appen: pseudo-algorithm} we provide a general pseudo-algorithm to implement our procedure. Appendix~\ref{appen: integration} describes the MCMC procedure used for latent variable integration. Appendix~\ref{appen: hitandrun} provides a description of the ``hit-and-run'' algorithm we used in the construction of the MCMC chain. We provide the specification for $\eta$ and the optimization routines used in our applications and simulations in Appendix~\ref{appen: inputs}. 

\subsection{Pseudo-Algorithm}\label{appen: pseudo-algorithm}
This pseudo-algorithm is based on Schennach's algorithm provided in GAUSS as a supplement to \citet{schennach2014entropic}. The actual implementation of the algorithm has been vectorized and parallelized. 

\begin{algorithmic}[1]

\STATE  \textbf{Step $0$} (Setting parameters)
\begin{itemize}
\item Fix $\mathcal{T}=\{0,\cdots,T\}$, consumer experiments, and $\mathcal{L}=\{1,\cdots,L\}$, set of goods.
\item Fix $g_I$ and $g_M$.
\item Fix $\Lambda$, the support of $(\rands{\lambda}_{t})_{t\in\mathcal{T}}$, and $\Delta$, the support of $(\rands{\delta}_{t})_{t\in\mathcal{T}}$.
\item Fix $\eta\in\mathcal{P}_{E|X}$ (See Appendix~\ref{appen: inputs} for details)
\item Fix $x=(x_i)_{i=1,\dots,n}$, where $x_i=(\rho_{i,t},c_{i,t})_{t\in\mathcal{T}}$ is $i$-th observation and $n$ is the sample size.
\end{itemize}
\STATE  \textbf{end Step $0$.}
\STATE  \textbf{Step $1$}(Integration: Evaluation of the objective function at a given $\gamma\in\Real^{\abs{\mathcal{T}}}$)
\begin{itemize}
\item Set $i=1$.
\end{itemize}
\STATE  \textbf{While $i\leq n$}
\begin{itemize}
\item Define the measure 
\[
\tilde{\eta}(\cdot|x_{i})=\eta(\cdot|x_{i})\Char{g_I(x_i,\cdot)=0}.
\]
\item Integrate latent variables using $\tilde{\eta}(\cdot|x_{i})$ to obtain $\tilde{h}_{M}(x_i,\gamma)$ (See Appendix~\ref{appen: hitandrun} for implementation details).
\item Set $i=i+1$.
\end{itemize}
\STATE \textbf{end While. }
\begin{itemize}
\item Compute 
\[
\hat{\tilde{h}}_{M}(\gamma)=\frac{1}{n}\sum_{i=1}^{n}\tilde{h}_{M}(x_i,\gamma)
\]
and  
\[
\hat{\tilde{\Omega}}(\gamma)=\frac{1}{n}\sum_{i=1}^{n}\tilde{h}_{M}(x_i,\gamma)\tilde{h}_{M}(x_i,\gamma)\tr-\hat{\tilde{h}}_{M}(\gamma)\hat{\tilde{h}}_{M}(\gamma)\tr.
\]
\item Compute $\mathrm{ObjFun}(\gamma)=n\hat{\tilde{h}}_{M}(\gamma)\tr\hat{\tilde{\Omega}}(\gamma)^{-}\hat{\tilde{h}}_{M}(\gamma)$.
\end{itemize}
\STATE  \textbf{end Step $1$.}

\STATE \textbf{Step $3$ }(Optimization Step)
\begin{itemize}
\item Compute $\mathrm{TS}_{n}=\min_{\gamma}{\mathrm{ObjFun}(\gamma)}$.
\end{itemize}
\STATE  \textbf{end Step $3$.}
\end{algorithmic}

\subsection{Latent Variable Integration}\label{appen: integration}
Evaluation of the objective function requires integrating latent variables. We use MCMC methods. For completeness we provide the algorithm for MCMC integration to get $\tilde{h}_M(x_i,\gamma)$.
\begin{algorithmic}[1]
\STATE  \textbf{Inputs}
\begin{itemize}
    \item Fix $\mathrm{cl}$ -- total MCMC chain length; $\mathrm{nburn}$ -- number of ``burned'' chain elements;
    \item Fix  $\eta$, $\gamma$, $x_i$, and the first element of the chain $e_{-\mathrm{nburn}}$ that satisfies the constraints.
    \item Set $r=-\mathrm{nburn}+1$ and $\tilde{h}(x_i,\gamma)_{M}(\gamma)=0$.
\end{itemize}
\STATE  \textbf{While $r\leq \mathrm{nsims}$}
\begin{itemize}
\item Draw $e_{\mathrm{jump}}=((v_{t})_{t\in\mathcal{T}},(\lambda_{t})_{t\in\mathcal{T}},(\delta_{t})_{t\in\mathcal{T}},w^{c},w^{\rho})$
proportional to \newline $\tilde{\eta}(\cdot|x_{i})=\eta(\cdot|x_{i})\Char{g_I(x_i,\cdot)=0}$. 
\item Draw $\alpha$ from $U[0,1]$.
\item Set $e_{r}$ equal to $e_{\mathrm{jump}}$ if $[g_{M}(x_i,e_{\mathrm{jump}})-g_{M}(x_i,e_{r-1})]\tr\gamma>\log (\rands{\alpha})$ and to $e_{r}$ otherwise. 
\STATE \textbf{if $r>0$}
\item Compute $\tilde{h}_{M}(x_i,\gamma)=\tilde{h}_{M}(x_i,\gamma)+g_{M}(x_i,e_{r}))/\mathrm{cl}$
\STATE \textbf{end if}
\item Set $r=r+1$
\end{itemize}
\STATE \textbf{end While. }
\end{algorithmic}

To compute the chain, one always can use ``rejection sampling'': at every step check whether a candidate element of the chain satisfies the inequalities (support constraints). Since our constraints have a simple form, we propose to use a version of the ``hit-and-run'' algorithm that we describe below. 

\subsection{``Hit-and-run'' Algorithm}\label{appen: hitandrun}
Since we use the algorithm presented below in our application with survey data and for concreteness we focus on $\s/\ed$-rationalizability with consumption measurement error. Note that instead of working with measurement errors in consumption, we can equivalently work with true unobserved consumption $c^*_t$. Thus, the latent variables, $e=(d,(c_t^{*\prime},v_t)_{t\in\mathcal{T}})$, have to satisfy the following set of constraints:
\begin{align*}
    &v_t-v_s\geq \dfrac{\rho\tr_t}{d^t}(c^*_t-c^*_s),\\
    &v_t,\:c^*_t\geq0,\\
    &1\geq d\geq \theta_0,
\end{align*}
for all $s,t\in\mathcal{T}$.
\par
The idea behind the ``hit-and-run'' algorithm is (i) to pick some initial point $e^0$ that satisfies the support constraints;\footnote{In our application and simulations we computed initial point by minimizing the norm of $g_M$ subject to the Afriat and sign constraints per observation.} (ii) to construct a candidate point by moving along a random direction within the constrained set on a randomly chosen distance; (iii) to use a user-specified Monte-Carlo acceptance rule to assign to $e^1$ either the initial point $e^0$ or the candidate point; (iv) to apply steps (ii) and (iii) to $e^1$ to construct $e^2$; (v) to repeat until the length of the chain reaches user chosen number.
\par
Take some arbitrary $e^r$ that satisfies the constraints. Let $\xi$ be a direction vector (not necessary unit vector). Thus, the candidate vector is
\[
e^{r+1}=e^{r}+\alpha\xi,
\]
where $\alpha\geq 0$ determines the scale of the perturbation $\alpha\xi$.
\par
\noindent\emph{Sign Constraints.} We start with sign constraints on consumption: $c^*_{t,l}\geq 0$ for all $l$ and $t$. Let $K_c$ be a set of indexes that correspond to $c^*_{t,l}$ in $e^r$. Hence, the constraints take the form
\[
\alpha\xi_{k}\geq -e^{r}_k,\quad\forall k\in K_{c} 
\]
Define $K_{+}=\{k\in K_{c}\::\:\xi_k>0\}$, $K_{-}=\{k\in K_{c}\::\:\xi_k<0\}$, and $K_{0}=\{k\in K_{c}\::\:\xi_k=0\}$. Then, the sign constraints are
\begin{align*}
    \alpha\geq-\dfrac{e^{r}_k}{\xi_k}, \quad \forall k\in K_{+},\\
    \alpha\leq-\dfrac{e^{r}_k}{\xi_k}, \quad \forall k\in K_{-}.
\end{align*}
Note that the constraints that correspond to $k\in K_{0}$ are always satisfied since $e^{r}_k\geq 0$ (i.e., satisfies the constraints). Thus, the sign constraints can be simplified to 
\begin{equation}\label{eq: hitandrun 1}
    \min_{k\in K_{-}}-\dfrac{e^{r}_k}{\xi_k}\geq\alpha\geq\max_{k\in K_{+}}-\dfrac{e^{r}_k}{\xi_k},
\end{equation}
where $\min_{k\in \emptyset}=+\infty$, and $\max_{k\in \emptyset}=-\infty$.
\par
\noindent\emph{Afriat Constraints.} Next, we consider the Afriat inequalities. Let $e^r(v,t)$ and $e^r(c,t)$ be the components of $e^r$ that correspond to $v_t$ and $c_t$, respectively. Assume that $d$ is fixed. Then, the Afriat inequalities are 
\[
e^r(v,t)-e^r(v,s)\geq \dfrac{\rho\tr_t}{d^t} (e^r(c,t)-e^r(c,s)), \forall t,s,
\]
Thus, since $d$ is fixed (the component of $\xi$ that corresponds to $d$ is zero), after plugging in $e^{r+1}$ we get
\[
e^r(v,t)-e^r(v,s)+\alpha (\xi(v,t)-\xi(v,s))\geq \rho_t\tr (e^r(c,t)-e^r(c,s))/d^t+\alpha\rho_t\tr (\xi(c,t)-\xi(c,s))/d^t, \forall t,s,
\]
where $\xi(v,t)$ and $\xi(c,t)$ are the components of $\xi$ that correspond to $v_t$ and $c_t$, respectively.
Hence,
\[
\alpha P_t\tr (\xi(t)-\xi(s))\leq -P_t\tr (e^r(t)-e^r(s)),\forall t,s,
\]
where $P_t=(-1,\rho_t\tr/d^t)\tr$, $\xi(t)=(\xi(v,t),\xi(c,t)\tr)\tr$, and $e^r(t)=(e^r(v,t),e^r(c,t)\tr)\tr$. Similarly to box constraints, we can define 
$T^A_{+}=\{(t,s)\::\:P_t\tr (\xi(t)-\xi(s))>0\}$, $T^A_{-}=\{(t,s)\::\:P_t\tr (\xi(t)-\xi(s))<0\}$, and $T^A_{0}=\{(t,s)\::\:P_t\tr (\xi(t)-\xi(s))=0\}$. Thus, the Afriat constraints are
\begin{equation}\label{eq: hitandrun 2}
\min_{(t,s)\in T^A_{+}}-\dfrac{P_t\tr (e^r(t)-e^r(s))}{P_t\tr (\xi(t)-\xi(s))}\geq \alpha\geq\max_{(t,s)\in T^A_{-}}-\dfrac{P_t\tr (e^r(t)-e^r(s))}{P_t\tr (\xi(t)-\xi(s))}.
\end{equation}
In other words we characterized possible perturbations of $c_t$ and $v_t$ given $d$ that are allowed under Afriat constraints.
Next we want to characterize the set for $d$ given consumption and utility numbers. Note that by assumption $d\in[\theta_0,1]$ and that
\[
d^t(v_t-v_s)\geq \rho_t(c_t-c_s).
\]
Define $T^A_{d,+}=\{(t,s)\::\:(v_t-v_s)>0\}$, $T^A_{d,-}=\{(t,s)\::\:(v_t-v_s)<0\}$, and $T^A_{d,0}=\{(t,s)\::\:(v_t-v_s)=0\}$. Hence, the Afriat ineqaulities are equivalent to
\begin{equation}\label{eq: hitandrun 3}
{\scriptstyle \min\left\{\min_{(t,s)\in T^A_{d,-}}\left(\max\left\{\dfrac{\rho_t(c_t-c_s)}{(v_t-v_s)},0\right\}\right)^{1/t},1\right\}\geq\: d\:\geq\max\left\{\max_{(t,s)\in T^A_{d,+}}\left(\max\left\{\dfrac{\rho_t(c_t-c_s)}{(v_t-v_s)},0\right\}\right)^{1/t},\theta_0\right\}.}
\end{equation}
\par
Inequalities (\ref{eq: hitandrun 1})-(\ref{eq: hitandrun 3}) give sharp restrictions on $\alpha$ that would guarantee that the next draw $e^{r+1}$ satisfies the constraints. Below we provide algorithms how to generate (i) new consumption and utility numbers given prices and discount factor, and (ii) discount factor given prices, consumption, and utility numbers. 

\begin{algorithmic}[1]
\STATE  \textbf{Generating new consumption vector and utility numbers}
\begin{itemize}
    \item Fix the discount factor and prices.
    \item Draw a random direction vector $\xi$ from a uniform distribution on the $[\abs{\mathcal{T}}+\abs{\mathcal{T}}\cdot L]$-dimensional unit sphere.
    \item Compute the interval $A$ using (\ref{eq: hitandrun 1}) and (\ref{eq: hitandrun 2}).
    \item Draw $\alpha$ uniformly from $A$.
    \item Generate new consumption vectors and utility numbers using $\xi$ and $\alpha$.
\end{itemize}
\STATE  \textbf{Generating new discount factor}
\begin{itemize}
    \item Fix prices, consumption, and utility numbers.
    \item Uniformly draw $d$ from the interval that satisfies (\ref{eq: hitandrun 3}).
\end{itemize}
\end{algorithmic}
Thus, we can propose the two approaches to sample from the cone characterized by the Afriat and the sign constraints. If one decides to keep the same $d$ for generating the chain, then one can initially draw several independent draws of $d$, and for every realization of $d$ generate its own chain. The second approach can be thought of as ``double-hit-and-run'': first generate new consumption and utility numbers, and then generate new discount factor using these new consumption and utility numbers. In our application we use double-hit-and-run approach. 
\subsection{User-specified \texorpdfstring{$\tilde\eta$}{Distribution}}\label{appen: inputs}
In this section we specify a particular choice of $\eta$ used in our applications and simulations.
\par
When integrating measurement error, instead of drawing measurement error (e.g., $\rand{w}_t^c$), we draw unobserved true variable (e.g., $\rand{c}_t^*$) and then constructed the measurement error by taking the difference between observed mismeasured and latent true variables (e.g., $\rand{w}_t^c=\rand{c}_t-\rand{c}^*_t$). Note that working with true variables allows us to easily generate measurement errors that imply correct signs for true variables (e.g., $\rand{c}_t^*\geq 0$). In particular, in our applications and simulations we impose sign constraints directly in the sampling stage (Step $2$ in Section~\ref{appen: integration}).  
\par
For our first application (survey data) to build $\tilde{\eta}$ we used the ``hit-and-run'' algorithm described in Appendix~\ref{appen: hitandrun} to produce draws of $\rand{e}$. In particular, $\rand{c}^*=(\rand{c}_t^*)_{t\in \mathcal{T}}$ is such that (i) it satisfies the Afriat-like inequalities and sign constraints, (ii) the user specified distribution over $\rand{w}^c=(\rand{w}_t^*)_{t\in \mathcal{T}}$ is 
\[
d\tilde{\eta}(w^c|x)\propto \exp(-\norm{g_M(x,e)}^2),
\]
where $g_{M}(x,e)=(\rho_t\tr w^c_t)_{t\in \mathcal{T}}$. To achieve this, we use the standard Metropolis Hastings algorithm in each step of the ``hit-and-run'' algorithm to get the draws from the desired distribution. Note that by construction this distribution has the correct support $E|X$.
\par 
The $\tilde{\eta}$ distribution can be adapted to accommodate other moments such as those in our extensions and counterfactual analysis by using the appropriate moment conditions. If the moment conditions, which are not support constraints, include other random variables in $\rand{e}$, the distribution $\tilde{\eta}$ will have to be defined on them, and not only on $\rand{w}^c$ like in our first application. 
\par
For our second application (experimental data) we use a different strategy since (i) the panel is long ($T=50$), (ii) the centering conditions do not depend on $v_t$ and $\delta_t$, and (iii) the static UMT has a simplified characterization in terms of Generalized Axiom of Revealed Preferences (GARP).\footnote{One can replace the Afriat inequalities by the GARP inequalities since GARP is equivalent to $\rat$-rationalizability in this case.} Hence, we can simplify our problem by considering a reduced latent random vector that consists only of true consumption or true prices. We then choose $\tilde{\eta}$ to be a uniform distribution over consumption or prices that satisfy GARP and that produce an expenditure level equal to $1$.\footnote{We impose nonnegativity constraints on consumption and positivity constraints on prices. The requirement to produce expenditure level equal to $1$ makes the support of consumption and price bounded.}  We can do this with the support constraints that we consider in this application. The key for a good computational performance of this step is to check for GARP consistency in an efficient way for each candidate draw of prices or consumption. For this purpose, we use a recursive algorithm to check GARP using an implementation of the deep-first search algorithm with recursive tabu search (see \citealp{Rboelaert2014}). 
\par
In both applications it is trivial to verify that these choices of $\tilde{\eta}$ satisfy the conditions stated in Definition~\ref{def:MEM}.

\subsection{Optimization}
We optimize the objective function specified in the pseudo-algorithm using Bobyqa procedure as implemented in the NLopt library following \citet{powell2009bobyqa}. As an initial guess for the optimization we use the outcome of applying BlackBox Differential Evolution Algorithm to minimize our objective function. Bobyqa performs derivative-free optimization using iteratively constructed quadratic approximations of the objective. We observe that in our simulations this combination of optimizers perform the best in terms of accuracy and speed among similar NLopt alternatives. 
\par
For the second application because the number of moments is larger we use as an initial guess for the optimization the outcome of two-step GMM estimator. Since the objective function of the two-step GMM estimator has a unique minimizer and is locally convex around it, we use Bobqya here as well. Bobqya works well in convex problems as documented in \citet{rios2013derivative}.
Following \citet{schennach2014entropic}, we additionally verified our results using Neldermead.\footnote{See \citet{sasaki2015contraction} for an alternative optimization technique.} 
\par
Another alternative to find good initial values is taking advantage of a convex problem related to our problem. As shown in \citet{schennach2014entropic}, the moment condition in Theorem~\ref{thm:ELVISexponentialdiscounting} is a first-order condition of the following convex optimization problem (Lemma~A.1 in \citealp{schennach2014entropic}):
\[
\min_{\gamma\in\Real^{q}}\Expt[\pi_0]{\ln\Expt[\tilde{\eta}]{\exp(\gamma\tr g_M(\rand{x},\rand{e}))|\rand{x}}}.
\]
Moreover, the norm $\norm{\Expt[\pi_0]{h(\rand x;\cdot)}}$ has a unique global minimum, is convex in the neighborhood of the minimizer if this minimizer is finite, and has no other local minima. Hence, computationally the problem is convenient. 

\section{Analytical Power Results. Robustness to Local Perturbations}\label{appen: counter and robustness}
In this appendix we provide examples of DGPs that will fail to pass our test (Sections~\ref{appen: counterexamples ED} and~\ref{appen: counterexamples GARP}). In Section~\ref{appen: local perturb} we show robustness of UMTs that we consider to local perturbations in observed quantities or prices.

\subsection{s/ED-Rationalizability, Mean-budget Neutrality, Price and Consumption Measurement Error}\label{appen: counterexamples ED}
In this section we construct the data set that can not be $\s/\ed$-rationalized by measurement error in consumption and time invariant measurement error in prices if the centering condition comes in the form 
\[
\Exp{\rands{\alpha}\rands{\rho}_t\tr\rand{c}_t}=\Exp{\rands{\alpha}\rands{\rho}_t^{*\prime}\rand{c}^{*}_t},\quad\forall t\in\mathcal{T},
\]
where $\rands{\alpha}\in(0,1]$, represents individual specific weights. We consider the environment with $2$ time periods and $2$ goods.  We assume that the price measurement error comes in the following form:
\begin{align*}
    \rands{\rho}^{*}_t&=\widetilde{\rand{W}}\rands{\rho}_t,
\end{align*}
where 
\[
\widetilde{\rand{W}}=\left(\begin{array}{cc}
     \widetilde{\rand{w}}_{1}^p&0  \\
     0&\widetilde{\rand{w}}_{2}^p 
\end{array}\right)
\]
is the matrix of time invariant multiplicative price measurement errors.

The above centering condition covers variety of measurement error. For instance, if $\rands{\alpha}=1\:\as$, then we have the centering condition used in our application. The random weight $\rands{\alpha}$ is allowed to be correlated with all observables and measurement error.
\par
Take $\{\rand{c}_t\}_{t=0,1}$ and $\{\rands{\rho}_t\}_{t=0,1}$ such that
\begin{align*}
\rands{\rho}_{0,1}&=\rands{\rho}_{0,2}=1\:\as,\quad\rands{\rho}_{1,1}=\rands{\rho}_{1,2}=2\:\as,\\
\rand{c}_{0,1}&=\rand{c}_{0,2}=1\:\as,\quad\rand{c}_{1,1}=\rand{c}_{1,2}=2\:\as.
\end{align*}
Denote $\rho_0=(1,1)\tr$, $\rho_1=2(1,1)\tr$, $c_0=(1,1)\tr$, and $c_1=2(1,1)\tr$.
By way of contradiction suppose that there exist $\rand{d}\in(0,1]$, $\{\rand{c}^{*}_t,\rands{\rho}_t^{*}\}_{t=0,1}$, $\rands{\alpha}\in(0,1]$, nonnegative $\{\rand v_t\}_{t=0,1}$ such that the Afriat inequalities and the centering conditions are satisfied:
\begin{align*}
    \rand{v}_1-\rand v_0\geq \dfrac{\rands{\rho}_1^{*\prime}}{\rand{d}}&(\rand{c}^*_1-\rand{c}^*_0)\:\as,\quad
    \rand{v}_0-\rand v_1\geq \rands{\rho}^{*\prime}_0(\rand{c}^*_0-\rand{c}^*_1)\:\as,\\
    \Exp{\rands{\alpha}\rands{\rho}_0\tr\rand{c}_0}=&\Exp{\rands{\alpha}\rands{\rho}_0^{*\prime}\rand{c}^{*}_0},\quad
    \Exp{\rands{\alpha}\rands{\rho}_1\tr\rand{c}_1}=\Exp{\rands{\alpha}\rands{\rho}_1^{*\prime}\rand{c}^{*}_1},\\
    &\rands{\rho}^{*}_0=\widetilde{\rand{W}}{\rho}_0,\quad
    \rands{\rho}^{*}_1=\widetilde{\rand{W}}{\rho}_1.
\end{align*}
\par
Note that since $\rand{d}>0$ with probability $1$, the inequalities can be rewritten as
\begin{align*}
    \dfrac{\rand{d}}{2}\left(\rand{v}_1-\rand v_0\right)&\geq \rho_0\tr\widetilde{\rand{W}}(\rand{c}^*_1-\rand{c}^*_0)\:\as,\\
    \rand{v}_0-\rand v_1&\geq \rho_0\tr\widetilde{\rand{W}}(\rand{c}^*_0-\rand{c}^*_1)\:\as.
\end{align*}
Hence,
\[
\dfrac{\rand{d}}{2}\left(\rand{v}_1-\rand v_0\right)\geq\rand{v}_1-\rand v_0.
\]
Thus, since $\rand{d}\in(0,1]$, we can deduce that $\rand{v}_1-\rand v_0\leq 0\:\as$ If we multiply the first Afriat inequality by strictly positive $\rands{\alpha}$ and take expectations from both sides, we get
\begin{align*}
0\geq&\Exp{\rands{\alpha}\dfrac{\rand{d}}{2}\left(\rand{v}_1-\rand v_0\right)}\geq \Exp{\rands{\alpha}\rho_0\tr\rand{W}\rand{c}^*_1}-\Exp{\rands{\alpha}\rho_0\tr\rand{W}\rand{c}^*_0}=\\
&\dfrac{\Exp{\rands{\alpha}\rho_1\tr\rand{W}\rand{c}^*_1}}{2}-\Exp{\rands{\alpha}\rands{\rho}^{*\prime}_0\rand{c}^*_0}=\dfrac{\Exp{\rands{\alpha}\rands{\rho}^{*\prime}_1\rand{c}^*_1}}{2}-\Exp{\rands{\alpha}\rands{\rho}^{*\prime}_0\rand{c}^*_0}=\\
&\dfrac{\Exp{\rands{\alpha}\rands{\rho}^{\prime}_1\rand{c}_1}}{2}-\Exp{\rands{\alpha}\rands{\rho}^{\prime}_0\rand{c}_0}=4\Exp{\rands{\alpha}}-2\Exp{\rands{\alpha}}=2\Exp{\rands\alpha}>0,
\end{align*}
where the equalities come from from the centering conditions and the fact that $\rho_1=2\rho_0$, and the last inequality is implied by $\Exp{\rands{\alpha}}>0$. The above contradiction implies that the constructed data set will never pass our test. 
\par
There are at least two implications of the example constructed in this section. First, we can further restrict $\rands{\alpha}$ by assuming that $\rands{\alpha}=1\:\as$. Thus, the above example demonstrates that without price measurement error the centering condition $\Exp{\rands{\rho}_t\tr\rand{w}^c_t}=0$, $t\in\mathcal{T}$ has empirical content. Second, note that the trembling-hand centering condition ($\Exp{\rand{w}^c_t}=0$, $t\in\mathcal{T}$) implies that in our example $\Exp{\rands{\rho}_t\tr\rand{w}^c_t}=0$, $t\in\mathcal{T}$,
since $\rands{\rho}_t$ has a degenerate distribution. Hence, the trembling-hand centering condition has empirical content as well.
\par
We conclude this section by noting that our example can be used to construct an example with time invariant consumption measurement error and time varying price measurement error because the mean budget neutrality condition and the Afriat inequalities are ``symmetric'' in prices and consumption.

\subsection{GARP and Trembling-Hand Measurement Error in Consumption or Prices}\label{appen: counterexamples GARP}
In the experimental data we use individuals are forced to pick points on the budget lines. That is, $\rands{\rho}^*_t\rand{c}^*_{t}=\rands{\rho}_t\rand{c}_{t}\:\as$ for all $t\in\mathcal{T}$. In this section we construct an example for the GARP with trembling-hand error in consumption. Consider $2$ goods and $2$ time periods environment with deterministic prices.
\begin{align*}
    p_0=(1,2)\tr,\quad p_1=(2,1)\tr.
\end{align*}
The observed consumption vectors are random and satisfy

\begin{tabularx}{\textwidth}{XX}
{\begin{align*}
&\rand{c}_0=\begin{cases}
(\epsilon,1-\epsilon/2)\tr,&\text{ with probability }1/2,\\
(\epsilon,3/4-\epsilon/2)\tr,&\text{ with probability }1/2,
\end{cases}
\end{align*}} 
& 
{\begin{align*}
&\rand{c}_1=\begin{cases}
(1-\epsilon/2,\epsilon)\tr,&\text{ with probability }1/2,\\
(3/4-\epsilon/2,\epsilon)\tr,&\text{ with probability }1/2,
\end{cases}
\end{align*}} 
\end{tabularx}
where $0<\epsilon<1/8$. Hence,
\begin{align*}
    \Exp{\rand{c}_0}=(\epsilon,7/8-\epsilon/2)\tr,\quad \Exp{\rand{c}_1}=(7/8-\epsilon/2,\epsilon)\tr.
\end{align*}
Next, note that observed disposable income $\rand{y}_t$ is a binary random variable:

\begin{tabularx}{\textwidth}{XX}
{\begin{align*}
&\rand{y}_0=\begin{cases}
2,&\text{ with probability }1/2,\\
3/2,&\text{ with probability }1/2,
\end{cases}
\end{align*}} 
& 
{\begin{align*}
&\rand{y}_1=\begin{cases}
2,&\text{ with probability }1/2,\\
3/2&\text{ with probability }1/2.
\end{cases}
\end{align*}} 
\end{tabularx}
\par
\noindent{First time period.} Since mismeasured consumption has to belong to the true budget line ($\rand{p}_t\tr \rand{c}_t=\rand{p}_t\rand{c}^*_t\:\as$ for all $t$) and on average has to agree with the observed consumption, we can conclude that
\begin{align*}
    \Prob{\rand{c}^*_{01}+2\rand{c}^*_{02}=3/2|\rand{y}_0=3/2}&=1,\quad \Prob{\rand{c}^*_{01}+2\rand{c}^*_{02}=2|\rand{y}_0=2}=1,\\
    \Exp{\rand{c}^*_{01}}=&\epsilon,\quad\Exp{\rand{c}^*_{02}}=7/8-\epsilon/2.
\end{align*}
Note that 
\begin{align*}
\Exp{\rand{c}^*_{01}}&=\Exp{\rand{c}^*_{01}|\rand{y}_0=3/2}\Prob{\rand{y}_0=3/2}+\Exp{\rand{c}^*_{01}|\rand{y}_0=2}\Prob{\rand{y}_0=2}\\
&=\dfrac{\Exp{\rand{c}^*_{01}|\rand{y}_0=3/2}}{2}+\dfrac{\Exp{\rand{c}^*_{01}|\rand{y}_0=2}}{2}.
\end{align*}

Hence, since $\rand{c}_t^*\geq 0\:\as$, $\Exp{\rand{c}^*_{01}|\rand{y}_0}\leq2\epsilon\:\as$.
Thus, given that $\rand{p}_t\tr \rand{c}_t=\rand{p}_t\rand{c}^*_t\:\as$ for all $t$ we get that $\rand{c}^*_0\in[0,2\epsilon]\times[3/4-\epsilon,1]$ with positive probability. Similarly, $\rand{c}^*_1\in[3/4-\epsilon,1]\times[0,2\epsilon]$ with positive probability. Thus, since $3/2-4\epsilon>1$ ($\epsilon<1/8$) it means that $\rand{c}^*_{0}$ is also available at $t=1$ with positive probability. Similarly, $\rand{c}^*_1$ is available when $t=0$ with positive probability. The latter violates GARP with positive probability (GARP has to be satisfied with probability $1$). Thus, there is no trembling-hand measurement error that keeps consumption on the same budget and is consistent with GARP. 

\par
We conclude this section by noting that GARP conditions are symmetric in terms of price and consumption vectors. Thus, after relabeling (swapping prices with consumption) the above DGP also will not pass our test if one assumes that there is only mean-zero measurement error in prices.

\subsection{Robustness to Local Perturbations}\label{appen: local perturb}
In this section we show that in many situations our approach is also robust to small measurement errors in observed quantities or prices. Suppose that we fix a model (i.e., the support restrictions on $\rands{\delta}_t$ and $\rands{\lambda}_t$, and the definition of $\rands{\rho}_t^*$). Define the measure of inequality slackness 
\begin{align*}
    \rands{\varepsilon}^*_{t,s}=\rands{\xi}_{t,s}-\rands{\rho}^*_t(\rand{c}^*_t-\rand{c}^*_s)\:\as
\end{align*}
where
\[
\rands{\xi}_{t,s}=\dfrac{\rands{\delta}_t}{\rands{\lambda}_t}(\rand{v}_t-\rand{v}_s).
\]
Suppose that there exist $\{\rand{v}_t,\rands{\delta}_t,\rands{\lambda}_t\}_{t\in\mathcal{T}}$ such that $\rands{\xi}_{t,s}\geq 0$. Next we perturb the true consumption and prices in order to see to what extent the RP inequalities are still valid. Note that the observed potentially mismeasured data $\{\rands{\rho}_t,\rand{c}_t\}_{t\in\mathcal{T}}$ satisfies the constraints with the same $\rands{\xi}_{t,s}$ if
\begin{align*}
    \rands{\varepsilon}^*_{t,s}-\left[\rand{w}^{p\prime}_t(\rand{c}_t-\rand{c}_s)+\rands{\rho}\tr_{t}(\rand{w}^c_t-\rand{w}^c_s)+\rand{w}^{p\prime}_t(\rand{w}^c_s-\rand{w}^c_t)\right]\geq0\:\as
\end{align*}
for all $t$ and $s$. Define
\begin{align*}
    \rands{\alpha}^p&=\max_{t}\norm{\rand{w}^{p}_t},
    \rands{\alpha}^c=\max_{t}\norm{\rand{w}^{c}_t},\\
    \rands{\beta}^p&=\max_{t}\norm{\rands{\rho}_t},\:
    \rands{\beta}^c=\max_{t}\norm{\rand{c}_t}.
\end{align*}
Then by the triangular inequality, the following inequality provides a sufficient restriction on the maximal perturbations of consumption and prices that will not refute the correctly specified model. 
\begin{align*}
    \rands{\varepsilon}^*_{t,s}\geq 6 \max\{\rands{\alpha}^p\rands{\beta}^c,\rands{\alpha}^c\rands{\beta}^p,\rands{\alpha}^p\rands{\alpha}^c\}\:\as.
\end{align*}
In other words, if the UMT leads to an ``interior solution'' (i.e., $\rands{\varepsilon}^*_{t,s}>0$ for all $t\neq s$), then small measurement errors without any centering restrictions will not affect the conclusions based on our testing procedure.

\section{Extensions of s/ED-Rationalizability}\label{appen: extensions of ED }
In this appendix we show that our methodology can cover two important extensions of $\ed$-rationalizability discussed in the main text: (i) $\ed$-rationalizability with income uncertainty (Appendix~\ref{appen: income uncert}); (ii) the collective model of \citet{adams2014consume} (Appendix~\ref{appen:extensions,collective}). 

\subsection{Income Uncertainty}\label{appen: income uncert}
In this section we consider a model of dynamic utility maximization with exponential discounting and income uncertainty. We start with the analysis of the deterministic model and then extend it to stochastic environments. 

\begin{defn}[Dynamic UMT with income uncertainty, $\ed$-IU-rationalizability] 
A deterministic array $(p_t,r_t,c_t)_{t\in \mathcal{T}}$ is $\ed$-rationalizable in the presence of income uncertainty ($\ed$-IU rationalizable) if: (i) There exists a concave, locally nonsatiated, and continuous function $u$. (ii) There exists  a
random income stream $\rand{y}=(\rand{y}_t)_{t\in\mathcal{T}}$. (iii) There exists an array of consumption and saving (policy) functions $\left(c_{t}(\cdot)\right)_{t\in\mathcal{T}}$ and $\left(s_{t}(\cdot)\right)_{t\in\mathcal{T}}$ such that $c_t:\Real_{+}^{\abs{\mathcal{T}}}\to\Real_{+}^{L}\setminus\{0\}$ and $s_t:\Real_{+}^{\abs{\mathcal{T}}}\to\Real_{+}$ for all $t\in \mathcal{T}$. In addition, we restrict these functions to depend only on the income history. That is, for all $t$, $c_t(y')=c_{t}(y)$ and $s_t(y')=s_{t}(y)$ for all $y$ and $y'$ such that $y'_{\tau}=y_{\tau}$ for all $\tau\leq t$.
(iv) The consumption and saving policy functions maximize the expected flow of instantaneous utilities given the budget constraints and history of incomes captured by information $I_t$:
\begin{align*}
    \max_{\left\{c_{\tau}(\cdot),s_{\tau}(\cdot)\right\}_{\tau=t,\dots,T}}\Exp{\sum_{\tau=t}^{T}\delta^{\tau-t}u(c_{\tau}(\rand{y}))\Big|I_t}
\end{align*}
subject to
\[
p_{\tau}^{\tr}c_{\tau}(\rand{y})+s_t(\rand{y})=\rand{y}_{\tau}+(1+r_{\tau})s_{\tau-1}(\rand{y})\:\as,
\]
for all $\tau=t,t+1,\dots,T$.
(v) The consumption stream $(c_t)_{t\in \mathcal{T}}$ at every time period $t$ is equal to the consumption policy function evaluated at a realization $y=(y_t)_{t\in \mathcal{T}}$ of the random income stream (i.e., $c_t= c_t(y)$). (vi) There is initial level of savings $s_0$. 
\end{defn}
$\ed$-IU-rationalizability extends $\ed$-rationalizability in an important direction: in accommodates possible uncertainty in the future income. Since the future income is unobserved, instead of a fixed vector of future consumption and savings, the agent has to come up with the whole consumption and saving functions in order to be ready for all possible realizations of the income stream.
\par
In the case of income uncertainty we can still use the first-order conditions approach with an important modification. Instead of considering a support constraint on the space of marginal utility of wealth, we restrict its law of motion.  
\begin{lem}[FOC for $\ed$-IU-rationalizability] 
A deterministic array $(p_t,r_t,c_t)_{t\in \mathcal{T}}$ is $\ed$-IU rationalizable if and only if there exists a concave, locally nonsatiated, and continuous function $u$, a discount factor $d\in(0,1]$,  and a positive \emph{random} vector $(\rands{\lambda}_t)_{t\in \mathcal{T}}$ such that:
\begin{enumerate}
    \item $\Exp{\rands{\lambda}_{t+1}|I_t}=\rands{\lambda}_t\:\as$, where $I_t$ the information ($\sigma$-algebra) generated by $(\rands{\lambda}_\tau)_{\tau\leq t}$.
    \item $d^{t} \nabla u(c_t)\leq \rands{\lambda}_t \rho_t\:\as$. If $c_{t,j}\neq0$, then $d^{t}\nabla u(c_{t})_j=\rands{\lambda}_{t}\rho_{t,j}\:\as$, where $c_{t,j}$, $\nabla u(c_{t})_j$, and $\rho_{t,j}$ are the $j$-th components of $c_{t}$, $\nabla u(c_{t})$, and $\rho_{t}$, respectively, and $\rho_t=p_t/\prod_{\tau=1}^{t}(1+r_\tau)$.
\end{enumerate}

\end{lem}
\begin{proof}
At every time period $t$ the agent is maximizing the expected flow of instantaneous utilities given the budget constraints and history of incomes captured by $I_t$:
\begin{align*}
    \max_{\left\{c_{\tau}(\cdot),s_{\tau}(\cdot)\right\}_{\tau=t,\dots,T}}\Exp{\sum_{\tau=t}^{T}\delta^{\tau-t}u(c_{\tau}(\rand{y}))\Big|I_t}
\end{align*}
subject to
\[
p_{\tau}^{\tr}c_{\tau}(\rand{y})+s_t(\rand{y})=\rand{y}_{\tau}+(1+r_{\tau})s_{\tau-1}(\rand{y})\:\as,
\]
for all $\tau=t,t+1,\dots,T$.

The Lagrangian function of the above problem takes the form
\begin{align*}
    \Exp{\sum_{\tau=t}^{T}\delta^{\tau-t}u(c_{\tau}(\rand{y}))\Big|I_t}-\sum_{\tau=t}^{T}\Exp{\dfrac{\lambda_{\tau}(\rand{y})}{\delta^{t}\prod_{j=1}^\tau(1+r_j)}\left[p_{\tau}^{\tr}c_{\tau}(\rand{y})+s_t(\rand{y})-\rand{y}_{\tau}-(1+r_{\tau})s_{\tau-1}(\rand{y})\right]\Big|I_{t}},
\end{align*}
where $\{\lambda_{\tau}(\cdot)\}_{\tau=t,\dots,T}$ are lagrange multipliers. The denominator $\delta^{t}\prod_{j=1}^\tau(1+r_j)$ is needed for scaling of $\lambda_{\tau}(\cdot)$. If the instantenious utility function is concave, then, since the constraints are convex, the first-order condition will provide necessary and sufficient conditions for $c_t$ and $s_t$ to be optimal.   
The first-order condition with respect to $c_{\tau}$ is
\[
\Exp{\left[\delta^{\tau}\nabla u(c_{\tau}(\rand{y}))-\lambda_{\tau}(\rand{y})\rho_{\tau}\right]\tr v_{c,\tau}(\rand{y})\Big|I_{t}}=0,
\]
for all $t\in\mathcal{T}$, $\tau=t,\dots,T$, and functions $v_{c,\tau}$, where $\rho_{\tau}=p_{\tau}/\prod_{j=1}^\tau(1+r_j)$. Note that, since for any $j=1,\dots,L$ the first order condition with respect to $c_\tau$ is satisfied with $v_c(\cdot)=(\Char{i=j})_{i=1,\dots,L}$, we have that the first order condition with respect to $c_\tau$ is satisfied if and only if   
\[
\delta^{t}\nabla u(c_{t}(\rand{y}))=\lambda_{t}(\rand{y})\rho_{t}\:\as,
\]
for all $t\in\mathcal{T}$. 
\par
Next, consider the first order condition with respect to $s_{\tau}$:
\[
\Exp{\left[\lambda_{\tau+1}(\rand{y})-\lambda_{\tau}(\rand{y})\right]v_{s,\tau}(\rand{y})\Big|I_{t}}=0
\]
for all $t\in\mathcal{T}$, $\tau=t,\dots,T$, and functions $v_{s,\tau}$. Because of the law of iterated expectations the later is equivalent to 
\[
\Exp{\lambda_{t+1}(\rand{y})-\lambda_{t}(\rand{y})\Big|I_{t}}=0
\]
for all $t\in\mathcal{T}$, since $v_{s,\tau}(\cdot)$ only depends on the history up to moment $\tau$. 
\end{proof}
The first corollary of the lemma above is that without imposing any restriction on income shocks, at the population level, it is impossible to discern whether an array $(p_t,r_t,c_t)_{t\in \mathcal{T}}$ is $\ed$-IU-rationalizable or $\rat$-rationalizable. The reason is that the only implication at the individual level of $\ed$-IU-rationalizability is that the marginal utility of income is positive.\footnote{This observation seems to have been noticed first by \citet{adams2014consume}.} However, if we assume that the latent income shocks are i.i.d., and there is no aggregate shocks, in addition to assuming that preferences are i.i.d. and stable in the time window of interest, we can still have testable implications of the model. This statistical version of the $\ed$-IU model is defined next. 
\par

\begin{defn} [$\s/\ed$-IU-rationalizability]
	A random array $(\rand{p}_t^*,\rand{r}_t^*,\rand c_{t}^{*})_{t\in\mathcal{T}}$
	is $\s/\ed$-IU-rationalizable if there exists a tuple $(\rand{u},(\rands{\lambda}_t,\rand{d})_{t\in\mathcal{T}})$ such that
	\begin{enumerate}
	    \item $\rand u$ is a random, concave, locally nonsatiated, and continuous utility function;
	    \item $\rand{d}$ is a random variable supported on $(0,1]$ interpreted as the time discount factor;
	    \item $(\rands{\lambda_t})_{t\in\mathcal{T}}$ is a positive random vector, interpreted as the marginal utility of income, such that for $t=0,\cdots,|\mathcal{T}|-1$:
	    \[
	    \Exp{\rands{\lambda}_{t+1}\Big|(\rands{\lambda_\tau})_{\tau\leq t},(\rands{\rho}_{\tau}^*)_{\tau\leq t},\rand{u},\rand{d},}=\rands{\lambda}_t\:\as,
	    \]
	    where $\rands{\rho}^*_t=\rand{p}_t^*/\prod_{\tau=1}^t(1+\rand{r}_\tau)$, $t\in\mathcal{T}$;
	    \item $\rand{d}^{t}\nabla \rands{u}(\rand c_{t}^{*})\leq\rands{\lambda_t}\rands{\rho}^*_{t}\:\as$ for all $t\in\mathcal{T}$;
	    \item For every $j=1,\dots,L$ and $t\in\mathcal{T}$, it must be that $\Prob{\rand{c}^*_{t,j}\neq 0,\rands{d}^{t}\nabla \rands{u}(\rand c_{t}^{*})_j<\rands{\lambda_t}\rands{\rho}^*_{t,j}}=0$, where $c_{t,j}^*$, $\rho^*_{t,j}$, and $\nabla {u}(c_{t}^{*})_j$ denote the $j$-th components of $c_t^*$, $\rho^*_t$, and $\nabla{u}(c_{t}^{*})$, respectively.
	\end{enumerate}
\end{defn}
\par
In words, at the beginning of the time-window of interest, a consumer draws a utility function, and a discount factor that are going to remain fixed in time. Every time, given the realized prices, utility, and discount factor the consumer draws a new marginal utility of income, and chooses consumption according to her first-order conditions. The marginal utility of income is a martingale with respect to the known information, which includes the realizations of utility, discount factor, and discounted prices. With this definition at hand we can write the following lemma.
\begin{lem}\label{lem:s-ED-IU-Afriat} 
For a given random array $(\rands{\rho}^*_{t},\rand c_{t}^{*})_{t\in\mathcal{T}}$, the following are equivalent:
	\begin{enumerate}
		\item The random array $(\rands{\rho}^*_{t},\rand c_{t}^{*})_{t\in\mathcal{T}}$ is $\s/\ed\text{-}\mathrm{IU}$-rationalizable.
		\item There exist positive random vectors $(\rand v_{t})_{t\in\mathcal{T}}$ and $(\rands{\lambda_{t}})_{t\in\mathcal{T}}$, and $\rand{d}$ supported on or inside $(0,1]$ such that
		\[
		\rand{v}_{t}-\rand{v}_{s}\geq\dfrac{\rands{\lambda}_{t}}{\rand{d}^t}\rands{\rho}^{*\prime}_{t}(\rand{c}_{t}^{*}-\rand{c}_{s}^{*})\quad\as,\quad\forall s,t\in\mathcal{T};
		\]
		and such that for $t=0,\cdots,|\mathcal{T}|-1$:
		\[
		 \Exp{\rands{\lambda}_{t+1}\Big|(\rands{\lambda_\tau})_{\tau\leq t},(\rands{\rho}_{\tau}^*)_{\tau \leq t},\{\rand{v}_\tau\}_{\tau \leq t},\rand{d}}=\rands{\lambda}_t\:\as.
		\]
	\end{enumerate}
\end{lem}
The proof of Lemma~\ref{lem:s-ED-IU-Afriat} is omitted as it follows trivially from our previous results. 

\subsubsection*{Econometric Framework}
The additional restrictions implied by the income uncertainty in classical $\ed$-rationalizability can be captured as a set of conditional moment conditions that restrict the latent distribution of the marginal utility of income. Note that our framework so far has only dealt with unconditional moments and support constraints. Fortunately, the ELVIS framework can deal with conditional moments too (Theorem 4.1 in \citealp{schennach2014entropic}). The main intuition behind this extension is that a finite set of conditional moments can be written as a (possibly infinite) collection of unconditional moments. 
\par
For simplicity of exposition and for practical purposes, instead of the martingale condition from Lemma~\ref{lem:s-ED-IU-Afriat}, we will use its simplest implication:
\[
\Exp{\rands{\lambda}_{t+1}-\rands{\lambda}_{t}}=0
\]
for all $t\in\mathcal{T}\setminus\{T\}$. Moreover, in order to be able to take expectations of marginal utility of income in the cross-section of individuals, we need to impose a normalization condition such that the marginal utility of income of all individuals is in the same units. A natural normalization (without loss of generality) in the form of a support constraint is
\[
\rands{\lambda}_0=1\:\as.
\]

Recall, that in the benchmark case of perfect-foresight ($\s/\ed$-rationalizability), the marginal utility of income is normalized to $1$ for every time period (i.e., $\rands{\lambda}_t=1\:\as$). In the case of the static UMT ($\rat$-rationalizability) $\rands{\lambda}_t$ is only restricted to be positive. The $\s/\ed$-IU-rationalizability provides a framework that is less restrictive than $\s/\ed$-rationalizability but is more restrictive than $\rat$-rationalizability.

\subsubsection*{Empirical Results}
In our first application we rejected the null hypothesis of $\s/\ed$-rationalizability with perfect-foresight for the case of couples' households (at the $5$ percent significance level). However, we fail to reject the implication of $\s/\ed$-IU-rationalizability captured by the above moment conditions for couples' households at the $5$ percent significance level with a discount factor set at $\rand{d}=1\:\as$. We find that $\mathrm{TS}_n=9.047$ (p-value$=0.249$) is below the $95$ percent quantile of the $\chi^2_{7}$ ($14.07$).
\par
We have not tested all necessary and sufficient conditions for $\s/\ed$-IU-rationalizability. But the evidence we provide suggests a possible explanation of the rejection of the perfect-foresight $\s/\ed$-rationalizability. In short, it may be that couples' households face more income uncertainty than singles. Hence, not taking income uncertainty into account could be a reason why we reject the dynamic UMT in the couples' households case. Indeed, \citet{browning2010uncertainty} points out that risk sharing may be a
benefit from marriage.  Further exploration of this explanation is beyond the scope of this paper.   

\subsection{Collective Exponential Discounting Model}\label{appen:extensions,collective}
The important contribution of \citet{adams2014consume} studies a dynamic collective consumer problem to model the behavior of couple's households. The collective model considers a case in which the household maximizes a utilitarian sum of individual utilities of each member of the couple over a vector of consumption of private and (household) public goods, given the individuals' relative power within the household (Pareto weights). Each individual member of the household is an exponential discounter but the observed consumption is a result of the collective decision making process, and may not be time-consistent. 
We formulate a test for the collective model using our methodology. We fail to reject the null hypothesis of consistency of the data set with the dynamic collective model assuming that the random discount factor is supported on $[0.975,1]$ (this support is the one used in \citet{adams2014consume}). 
\par
Consider a household that consists of two individuals labeled by $A$ and $B$. Partition the vector of goods into publicly consumed goods indexed by $H$ and privately consumed goods indexed by $I$. That is, $c_{t}=(c_{t,I}\tr,c_{t,H}\tr)\tr$ and $p_{t}=(p_{t,I}\tr,p_{t,H}\tr)\tr$. Let $c_{t,A}$ and $c_{t,B}$ be the consumption of the privately consumed goods of individuals $A$ and $B$, respectively ($c_{t,I}=c_{t,A}+c_{t,B}$). Then the
collective household problem with exponential discounting corresponds to the maximization of
\[
V_{\tau}(c)=\omega_{A}u_{A}(c_{\tau,A},c_{\tau,H})+\omega_{B}u_{B}(c_{\tau,B},c_{\tau,H})+\sum_{j=1}^{T-\tau}[d_{A}^{j}\omega_{A}u_{A}(c_{\tau+j,A},c_{\tau+j,H})+d_{B}^{j}\omega_{B}u_{B}(c_{\tau+j,B},c_{\tau+j,H})],
\]
subject to this linear intratemporal budget constraint:
\begin{align*}
p_{\tau,I}\tr c_{\tau,I}+p_{\tau,H}\tr c_{\tau,H}+s_{t}-y_{t}-(1+r_{t})s_{t-1}=0,
\end{align*}
where $\omega_{A},\omega_{B}>0$ are Pareto weights that remain constant across
time and represent the bargaining power of each household member. Individual utility functions, $u_{A}$ and $u_{B}$, are assumed to be continuous, locally nonsatiated and concave. The individual discount factors are similarly denoted by $d_A$ and $d_B$.
The rest of the elements are the same as in our main model.
\par
The quantities $c_{t,A},c_{t,B}$ are assumed to be unobservable
to the econometrician. We observe only $c_{t}$. \citet{adams2014consume} propose one solution to the collective household problem above. They assume full efficiency in the sense that there are personalized Lindahl prices for the publicly consumed goods $p_{t,H}$ that perfectly decentralize the above problem. The Lindahl prices are $p_{t,A}\in\Real_{++}^{L_{H}}$ for household member $A$ and the analogous $p_{t,B}$ such that $p_{t,A}+p_{t,B}=p_{t,H}$. \citet{adams2014consume} established the result which is the analog of Theorem~\ref{thm:Deterministic-thm_exponentialdiscount}. Similar to the case of the single-individual household, define  $\rho_{t,h}={p_{t,h}}/{\prod_{j=1}^{t}(1+r_{j})}$ for
$h\in\{I,H,A,B\}$.
\begin{thm}[\citealp{adams2014consume}]\label{thm:CollectiveDeterministicCharacterization}
An array $(\rho_{t},c_{t})_{t\in\mathcal{T}}$ can be generated by a collective household exponential discounting model with full efficiency
if and only if there exist $d_{A},d_{B}\in(0,1]$; strictly positive vectors $(v_{t,A})_{t\in\mathcal{T}}$, $(v_{t,B})_{t\in\mathcal{T}}$; individual private consumption quantities $(c_{t,A},c_{t,B})_{t\in\mathcal{T}}$ (with $c_{t,A}+c_{t,B}=c_{t,I}$); and personalized Lindahl prices $(p_{t,A},p_{t,B})_{t\in\mathcal{T}}$ (with $p_{t,A}+p_{t,B}=p_{t,H}$) such that for all $s,t\in\mathcal{T}$:
\begin{align*}
&v_{t,A}-v_{s,A}\geq d_{A}^{-t}\left[\rho\tr_{t,I}(c_{t,A}-c_{s,A})+\rho\tr_{t,A}(c_{t,H}-c_{s,H})\right],\\
&v_{t,B}-v_{s,B}\geq d_{B}^{-t}\left[\rho\tr_{t,I}(c_{t,B}-c_{s,B})+\rho\tr_{t,B}(c_{t,H}-c_{s,H})\right].
\end{align*}
\end{thm}
With this result in hand, we can establish our finding in a very straightforward manner. We let $\rands{\rho}_{t}$ and $\rand{c}^*_{t}$ be the random vectors of deflated prices and true consumption. Finally, we define $\rand d_{A}$ and $\rand d_{B}$ as the random discount factors for household members $A$ and $B$, respectively. Also, $\rand u_{A},\rand u_{B}$ and $\rands{\omega}_{A},\rands{\omega}_{B}$ denote the random utility functions and random Pareto weights for each household member. We keep here the assumption about the data-generating process that we maintained for the case of $\s/\ed$-rationalizability, namely, we assume that the preferences and Pareto weights remain stable for each household after being drawn from the joint distribution of $(\rand u_{A},\rand d_{A},\rands{\omega}_{A})$ and $(\rand u_{B},\rand d_{B},\rands{\omega}_{B})$ at the first time period. Now we can establish and prove a stochastic analogue to the result in \citet{adams2014consume}.
\begin{thm}
	\label{thm:CollectiveSufficient} If a random array $(\rands{\rho}_{t},\rand c_{t}^{*})_{t\in\mathcal{T}}$ is generated by a collective household with random exponential discounting under full efficiency, then there exist random variables $\rand d_{A},\rand d_{B}$ which are both supported on or inside $[\theta_0,1]$, and strictly positive random vectors $(\rand v_{t,A})_{t\in\mathcal{T}}$, $(\rand v_{t,B})_{t\in\mathcal{T}}$ that satisfy
	\[
	\rand d_{A}^{t}(\rand v_{t,A}-\rand v_{s,A})+\rand d_{B}^{t}(\rand v_{t,B}-\rand v_{s,B})\geq\rands{\rho}_{t}\tr(\rand c_{t}^{*}-\rand c_{s}^{*})\:\as\quad\forall t,s\in\mathcal{T}.
	\]
\end{thm}
Theorem~\ref{thm:CollectiveSufficient} does not provide sufficient conditions for collective rationalizability. We can provide a stochastic analogue of Theorem~\ref{thm:CollectiveDeterministicCharacterization}, but our choice has several advantages:  (i) one does not need to specify which goods are consumed privately or publicly; (ii) the inequality restrictions in Theorem~\ref{thm:CollectiveSufficient} do not depend on the unobservable Lindahl prices and private consumption vectors, which simplifies implementation; and (iii) we can maintain Assumption~\ref{assu:(Measurement-error)Meanbudgetneutrality} in
a very natural form.  We also assume that prices are measured precisely. Assuming that $\rand{d}_A$ and $\rand{d}_B$ are supported on or inside $[0.975,1]$ we find that $\mathrm{TS}_n=0.018$ ($\text{p-value}>1-10^{-4}$), which is below the $95$ percent quantile of the $\chi^2_{4}$ ($9.49$). Thus, we fail to reject the null hypothesis that the couples' household data set is consistent with the collective exponential discounting model under the assumptions of full efficiency, common support for preferences, and the collective mean budget constraint. The test statistic value for the explored $\theta_0=0.975$ for the collective model is below that of the exponential discounting model for the sample of couples' households. 

\section{Empirical Application (I) Extended: Average Varian Support Set for Budget Shares\label{appen:averagevarianempirical}}
Here we compute bounds on counterfactual average budget shares. Since the null hypothesis of $\s/\ed$-rationalizability cannot be rejected for the case of single-individual households in our first application, we can compute out-of-sample forecasts of average budget shares when the price changes. We focus on one of the categories, petrol. The motivation for this exercise is twofold. First, we want to showcase that our methodology can deliver informative bounds of quantities of interest. Second, fossil fuels prices are usually highly variable. This means that understanding the effects of such price changes on the problem of budget allocation in households is empirically relevant. This question may also be important for some potential users of this methodology, like a regulator trying to impose a green-tax.  We study the following counterfactual question: What would be the average budget share of petrol in time $T+1$, if the price of petrol $p_{T+1,\mathrm{pet}}$ is  $\kappa\cdot100$ percent higher than $p_{T,\mathrm{pet}}$ (i.e., $p_{T+1,\mathrm{pet}}=(1+\kappa)p_{T,\mathrm{pet}}$)? 
\par 
For simplicity, we consider $\kappa$ that takes values in  $\{0,0.01,\cdots,0.10\}$ (i.e., at most a $10$ percent price increase). We set the (random) interest rate faced by the single-individual households $\rand{r}_{T+1}=0.06 \:\as$ (roughly $1$ percent increase over the average interest rate) and the support of the random discount factor to $[0.975,1]$.\footnote{We computed our counterfactual sets with the support of the random discount factor equal to $[0.1,1]$ and $\rand{d}=1\:\as$. The results are similar to those with $[0.975,1]$ and available upon request.} The counterfactual moment is
\[
g_C((\rands{\rho}^*_{t},\rand c^*_{t})_{t\in\mathcal{T}},\rand{c}^*_{T+1};\rands{\rho}^*_{T+1},\theta_j)=\frac{\rands{\rho}^*_{T+1,\mathrm{pet}} \rand{c}^*_{T+1,\mathrm{pet}}}{\rands{\rho}_{T+1}^{*\prime} \rand{c}^*_{T+1}}-\theta_{\mathrm{pet}},
\]
where $\theta_{\mathrm{pet}}\in[0,1]$ is the average budget share of petrol. Note that interest rate cancels out such that the budget share depends only on spot prices. For $\ed$-rationalizability we do not need to specify the expenditure level at $T+1$, as the model endogenously predicts an expenditure level for a new price. This is because $\ed$-rationalizability generalizes quasilinear rationalizabilty by adding discounting. However, adding discounting does not affect the model's budget-free nature \citep{gauthier2018}.
\par 
The $95$ percent bounds for the average counterfactual petrol share are depicted in Figure~\ref{fig:petrol_0.975_bounds}. Note that conditions of the Proposition~\ref{prop:convexity} are satisfied. Hence, the sets of interest are connected, which means it is enough to depict the minimal and maximal shares that are not rejected by our test.\footnote{We searched for average budget share $\theta_{\mathrm{pet}}$ in the grid $\{0.00,0.005,\cdots,1\}$.}
As expected, demand for petrol is decreasing in the price of petrol. The maximal and the minimal drops in shares associated with $10$ percent price increase are $1.25$ to $1.1$ percentage points, respectively.\footnote{The empirical budget share at $t=T$ computed from our (mismeasured) data set for petrol is $6$ percent.}
  \begin{figure}
    \centerline{\includegraphics[width=1\textwidth]{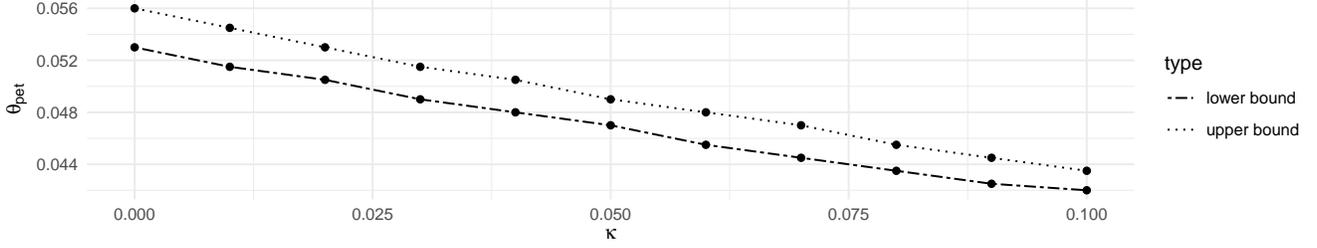}}
    \caption{Average Support Set for Budget Shares: Petrol\label{fig:petrol_0.975_bounds}} 
  \end{figure}

\section{Data Availability}
The data sets and replication codes underlying this article are available in Zenodo, at \url{https://doi.org/10.5281/zenodo.4007866}.



\end{document}